%% file: main.tex
\documentclass{article}

\input packages

\usepackage{geometry}
\usepackage[colorlinks,citecolor=red,linkcolor=blue]{hyperref}
\usepackage{thmtools,cleveref}

\newcommand{\renewtheorem}[1]{%
  \expandafter\let\csname #1\endcsname\relax
  \expandafter\let\csname c@#1\endcsname\relax
  \expandafter\let\csname end#1\endcsname\relax
  \newtheorem{#1}%
}

\theoremstyle{example} 

\theoremstyle{definition}

\newif\ifarxiv
\arxivtrue

\newtheorem{theorem}{Theorem}[section]
\newtheorem{lemma}{Lemma}[section]
\newtheorem{definition}{Definition}[section]
\newtheorem{notation}{Notation}[section]
\newtheorem{remark}{Remark}
\newtheorem{example}{Example}[section]
\newtheorem{proposition}{Proposition}[section]
\newtheorem{corollary}{Corollary}[section]

\ifarxiv
\usepackage{libertine}
\else
\fi 

\crefname{enumi}{Part}{Parts}
\crefname{section}{\S\!\!\,}{\S\!\!\,}%
\crefname{subsection}{\S\!\!\,}{\S\!\!\,}%
\crefname{subsubsection}{\S\!\!\,}{\S\!\!\,}%
\crefname{appendix}{\S\!\!\,}{\S\!\!\,}

\begin{document}

\title{Comparing Type Systems for Deadlock Freedom}
\author{Ornela Dardha \\ University of Glasgow, UK \and  Jorge A. P\'{e}rez \\ University of Groningen, The Netherlands}

\maketitle

\begin{abstract}
Message-passing
software systems exhibit non-trivial forms of concurrency and distribution; 
they are expected to 
follow intended  protocols among communicating services, but also to never ``get stuck''. 
This intuitive requirement has been expressed by {liveness} properties such as progress or (dead)lock freedom and
various type systems ensure these properties for concurrent processes. 
Unfortunately, very little is known about 
the precise relationship between these type systems and the classes of typed processes they induce. 

This paper
puts forward the first comparative study of different type systems for message-passing  processes
that guarantee deadlock freedom.
We compare two  classes of deadlock-free typed processes, here denoted \lcp and \fullKoba.
The class \lcp 
stands out for its canonicity: it 
results from Curry-Howard interpretations of \rev{classical} linear logic propositions as {session types}.
The class \fullKoba, obtained by encoding session types into Kobayashi's {linear types with usages}, 
includes processes not typable in other type systems. 
We show that \lcp is strictly included in \fullKoba, and identify the precise conditions under which they coincide. 
We also provide two type-preserving translations of processes in \fullKoba into processes in \lcp. 
\end{abstract}


%


%

\section{Introduction}
\label{sec:intro}
In this paper, we formally relate different type systems for concurrent processes specified in the $\pi$-calculus~\cite{DBLP:journals/iandc/MilnerPW92a}.
A fundamental model of computation, the $\pi$-calculus stands out for its {expressiveness}, which enables us 
to represent and reason about message-passing programs in 
functional, object-oriented, and distributed paradigms~\cite{Sangio01}. 
Another distinctive aspect of the $\pi$-calculus is its support for
rich {type systems} that discipline process behavior~\cite{DBLP:conf/concur/Sangiorgi02}.
Following Milner's seminal work on \emph{sorting}~\cite{milner-sort}, various type systems for the $\pi$-calculus have revealed a rich landscape of models for concurrency with disciplined communication; examples include
\emph{graph types}~\cite{DBLP:conf/fsttcs/Yoshida96}, 
\emph{linear types}~\cite{DBLP:conf/popl/KobayashiPT96},
\emph{generic types}~\cite{DBLP:journals/tcs/IgarashiK04},
and
\emph{session types}~\cite{H93,THK94,HVK98}. 
In the last decade, logical foundations for message-passing concurrency, in the style of  Curry-Howard correspondence (``propositions as types'') between linear logic and session types, have been put forward~\cite{CairesP10,DBLP:conf/icfp/Wadler12}.
By disciplining the use of channels, types for message-passing processes strongly influence their expressiveness.
Contrasting  
different type systems through the classes of well-typed processes that they induce is a central theme in this work.

Our interest is in \emph{session-based concurrency}, 
the model of concurrency captured by {session types}.
Session types promote
a type-based approach to communication correctness: dialogues between participants are structured into \emph{sessions},  basic   communication units; descriptions of interaction sequences are then abstracted as  session types  which are checked against process specifications. 
In session-based concurrency,  
types enforce correct communications through different safety and liveness properties. 
Two basic (and intertwined) correctness properties are \emph{communication safety} and \emph{session fidelity}: 
while the former ensures absence of errors (e.g., communication mismatches), the latter 
ensures that well-typed processes respect the  protocols  prescribed  by session types.

A very desirable {liveness} property for safe processes is that they should never ``get stuck''.
This is the well-known \emph{progress} property, which 
asserts that a well-typed term either is a final value or can further reduce~\cite{Pierce02}.
In calculi for concurrency, this 
 property 
 admits several formalizations; two of them are 
\emph{deadlock freedom}
and 
\emph{lock freedom}.
Intuitively, deadlock freedom ensures that communications will eventually succeed unless the whole process diverges \cite{K06};
lock freedom is a stronger property: it
guarantees that communications will eventually succeed, regardless of whether processes diverge, modulo some fairness assumption on the scheduler~\cite{K02}.
Notice that in the absence of divergent behaviors, deadlock freedom and lock freedom coincide.

Another formalization, which we call here \emph{session progress}, has been proposed for session-based concurrency~\cite{CD10,CDM14}: ``a process has session progress if combined with another process providing the corresponding co-actions (a so-called \emph{catalyzer}), then the composition reduces''.
We will say that a process is \emph{composable} if  it can be   composed with an appropriate catalyzer for some of its actions, or \emph{uncomposable}  otherwise.
Carbone et al. \cite{CDM14} proved that session progress and (dead)lock freedom coincide for {uncomposable processes}; for composable processes, session progress states \emph{potential} (dead)lock freedom. 
We will return to this informal distinction between composable and uncomposable processes below
(\Cref{d:openclose} will make it formal).

There is a vast literature on type systems for which typability enforces (dead)lock freedom or session progress---see, e.g.,\cite{K02,K06,DLY07,CD10,CairesP10,DBLP:journals/mscs/CairesPT16,P13,DBLP:conf/coordination/VieiraV13,P14,DBLP:conf/concur/GiachinoKL14,DBLP:journals/iandc/0001L17,DG18,DBLP:conf/esop/BalzerTP19}.
Unfortunately, these sophisticated systems 
rely on different principles
and/or 
consider different variants of the (session) $\pi$-calculus.
Also, different papers sometimes use different terminology and notions.
As a result, very little is known about the relationship between these type systems.
These observations led us to our research questions:
\emph{How do the type systems for liveness properties relate to each other?
What classes of deadlock-free processes do they induce?}
%
%

This paper presents the \emph{first formal comparison} between different type systems for the $\pi$-calculus that enforce liveness properties related to {(dead)lock freedom}.
More concretely, 
we tackle the above open questions by comparing
\lcp and \fullKoba,
two salient classes of deadlock-free (session) typed processes:

\begin{enumerate}[label=$\bullet$]
\item The class $\lcp$ contains  session  processes that are well-typed 
under the Curry-Howard correspondence between 
\rev{(classical)} linear logic propositions and session types~\cite{CairesP10,DBLP:journals/mscs/CairesPT16,DBLP:conf/icfp/Wadler12}.
This suffices, because 
the type system derived from such a correspondence simultaneously ensures 
communication safety, 
session fidelity, and 
deadlock freedom.

\item The class \fullKoba contains session processes that 
enjoy communication safety and session fidelity (as ensured by the type system of Vasconcelos~\cite{V12})
as well as satisfy deadlock freedom.
This class of processes is defined indirectly, 
by combining Kobayashi's type system based on {usages}~\cite{K02,K06,K07} with encodability results by Dardha et al.~\cite{DGS12}.
\end{enumerate}

Let us  clarify the nature of processes in \lcp and \fullKoba.
As  \Cref{d:lang} formally states, 
processes in \lcp and \fullKoba are typed under some typing context, possibly {non} empty. 
As such, these classes contain both {composable} processes (if the typing context is not empty) and {uncomposable} processes (otherwise).
Thus, strictly speaking, processes in \lcp and \fullKoba have session progress (as described above), which is strictly weaker than {deadlock freedom}, because a process satisfying session progress does not necessarily satisfy deadlock freedom. This is trivially true for composable processes, which need a catalyzer \cite{CDM14}. However, since we shall focus on uncomposable  processes, for which session progress and deadlock freedom coincide,  we shall refer to \lcp and \fullKoba as classes of deadlock-free processes.

There are good reasons for investigating   \lcp and \fullKoba.
On one hand, due to its deep logical foundation, the class \lcp appears to us as the {canonical} class of deadlock-free session processes, upon which all other classes should be compared. Indeed, this class arguably offers the most principled yardstick for comparisons.
On the other hand, the class \fullKoba integrates session type checking with the sophisticated usage discipline developed by Kobayashi
for $\pi$-calculus processes.
This indirect approach to deadlock freedom, developed in \cite{CDM14,Dardha16}, 
is general: it can capture sessions with subtyping, polymorphism, and higher-order communication.
Also, as discussed in~\cite{CDM14}, \fullKoba strictly includes classes of session typed processes induced by other type systems for deadlock freedom~\cite{DLY07,CD10,P13}. 

\subsubsection*{Contributions}
This paper contributes technical results that, on the one hand, \emph{separate} the classes \lcp and \fullKoba by precisely characterizing the fundamental differences between them and, on the other hand, \emph{unify} these classes by showing how their differences can be overcome to translate processes in \fullKoba into processes into \lcp. More in details:

\begin{itemize}
	\item To \emph{separate} \lcp from \fullKoba, we  define \muKoba: a sub-class of \fullKoba whose definition internalizes the key aspects of the Curry-Howard interpretations of session types. In particular, \muKoba adopts the principle of ``{composition plus hiding}'', a distinguishing feature of the interpretations in~\cite{CairesP10,DBLP:conf/icfp/Wadler12}, by which concurrent cooperation is restricted to the sharing of \emph{exactly one} session channel.
	
	We show that  \lcp and \muKoba coincide (\Cref{t:cppkoba}). 
	This gives us a separation result: there are deadlock-free session processes that \emph{cannot} be typed by  systems derived from the 
Curry-Howard interpretation of session types~\cite{CairesP10,DBLP:journals/mscs/CairesPT16,DBLP:conf/icfp/Wadler12},
but that are admitted as typable by the (indirect) approach of~\cite{CDM14,Dardha16}. 
	
	\item To \emph{unify} \lcp and \fullKoba, we define
two \emph{translations} of processes in \fullKoba into processes in \lcp  
\ifarxiv
(\Cref{def:typed_enc} and~\ref{def:typed_encr}).
\else 
(\Cref{def:typed_enc}).
\fi 
Intuitively, because the difference between \lcp and \fullKoba lies in the forms of parallel composition they admit (restricted in \lcp, liberal in \fullKoba), it is  natural to transform a process in \fullKoba into another, more parallel process in \lcp. 
In essence, the first translation, denoted $\encCP{\cdot}{}$ (\Cref{def:typed_enc}),  exploits type information to replace sequential prefixes with representative parallel components;
the second translation refines this idea by considering \emph{value dependencies}, i.e., causal dependencies between {independent} sessions not captured by types.
\ifarxiv
Our translations satisfy type-preservation and operational correspondence properties 
(Theorems~\ref{thm:L0-L2}, \ref{thm:oc}, \ref{thm:rwtp}, and \ref{thm:rwoc}).
\else
We detail the first translation, which satisfies type-preservation and operational correspondence properties 
(Theorems~\ref{thm:L0-L2} and \ref{thm:oc}).
\fi 
\end{itemize}

Our separation result is significant as it establishes the precise status of logically motivated  systems with respect to previous disciplines, not based on Curry-Howard principles.
It also provides a new characterization for linear logic-based processes, leveraging a very different type system for deadlock freedom \cite{K06}.
Also, our unifying result is insightful because it shows that the differences between the two classes of processes manifest themselves rather subtly at the level of process syntax, and can be eliminated by appropriately exploiting information at the level of types.

\rev{To further illustrate these salient points, consider the process
$$
P\defeq \res{xy}\res{wz}(
x\inp s.w\out s
 \pp 
\ov y\out \unit.\ov z\inp  u)
$$
where,
following the syntax of the type system of Vasconcelos~\cite{V12}, we use the restriction 
 $\res{xy}$ to indicate that $x$ and $y$ are the two endpoints of the same channel (and similarly for  $\res{wz}$). 
 Also, $\unit$ denotes a terminated channel.
 Process $P$ belongs to the class \fullKoba.
It  consists of two processes in parallel, each implementing \emph{two separate sessions} in sequence; as such, $P$ does not respect the monolithic structure imposed by   ``composition plus hiding''.
Our separation result is that $P$ falls outside \muKoba, the strict sub-class of  \fullKoba that adheres to ``composition plus hiding''.
Now, process $P$ has an operationally equivalent cousin, the process
$$
Q \defeq \res{xy}\res{wz}(
x\inp s.w\out s
 \pp 
\ov y\out \unit \pp \ov z\inp  u)
$$
which is structurally equivalent
to one that is more parallel and respects ``composition plus hiding'':
$$Q \equiv \res{wz}(\res{xy}(
x\inp s.w\out s
 \pp 
\ov y\out \unit) \pp \ov z\inp  u)
$$
Our unifying result is that
all processes in \fullKoba but outside of \muKoba have these kind of more parallel,
operationally equivalent cousins, which  respect the ``composition plus
hiding'' principle that characterizes class \lcp.}

To our knowledge, our work is the first to formally compare fundamentally distinct type systems for deadlock freedom in message-passing processes.
Previous comparisons by Carbone et al.~\cite{CDM14} and Caires et al.~\cite[\S6]{DBLP:journals/mscs/CairesPT16}, are informal: they are based on 
representative ``corner cases'', namely examples of deadlock-free session processes typable in
one system but not in some other.

\begin{figure*}[t]
\begin{mdframed}
\centering
\input{figures/diagram.tex}
		\end{mdframed}
	\caption{Summary of type systems, languages with deadlock freedom (DF), and encodings between them (indicated by solid black arrows). Main results are denoted by \textcolor{ForestGreen}{green} lines: our separation result, based on the coincidence of \lcp and \muKoba is indicated by the solid line with reversed arrowheads; our unifying result is indicated by the dashed arrow.}
	\label{f:overview}
\end{figure*}

\rev{\Cref{f:overview} summarizes the different type systems and ingredients needed to define \fullKoba and \lcp.}

\subsubsection*{Paper Organization}
 In \Cref{s:sessions} we summarize the 
    session $\pi$-calculus and the  type system of~\cite{V12}.
In \Cref{s:deadlf} we present the two typed approaches to deadlock freedom for sessions.
 \Cref{s:hier}  defines the classes  \lcp and \fullKoba, the strict sub-class \muKoba,
and shows that $\lcp  = \muKoba$ (\Cref{t:cppkoba}).
\Cref{s:enco}  presents our first translation  of \lkoba into \lcp and establishes its correctness properties. The optimization with value dependencies is discussed in~\Cref{ss:oprw}.
Enforcing deadlock freedom by typing is already challenging for  processes   without constructs such as recursion or replication. For this reason, here we concentrate on finite processes; in \Cref{s:discuss} we discuss the case of processes 
with unbounded behavior (with constructs such as replication and recursion) \nurev{and connections with other logical interpretations of session types}.
\Cref{s:rw} 
compares with related works and 
\Cref{s:concl}  collects some concluding remarks.
Omitted technical details are included in the 
appendices.

This paper is a revised version of the workshop paper~\cite{DBLP:journals/corr/DardhaP15}, extended with new material: we present full technical details and additional examples not presented in~\cite{DBLP:journals/corr/DardhaP15}. Also, we present updated comparisons with related works.
The separation result based on the coincidence with the sub-class \muKoba, given in  \Cref{s:hier}, is new.
Moreover, 
the first translation, presented in~\cite{DBLP:journals/corr/DardhaP15} and given in \Cref{s:enco}, 
 has been substantially  simplified and its correctness properties have been refined.
 \ifarxiv
The second translation, based on value dependencies and discussed in \Cref{ss:oprw}, is new to this paper.
\else
\fi 
The content of \Cref{s:discuss} is also original to this presentation.


\section{Session \p calculus}\label{s:sessions}
We present the session $\pi$-calculus and its corresponding type system: the linear fragment of the type system by Vasconcelos~\cite{V12}, which ensures communication safety and session fidelity (but not progress nor deadlock freedom). 
Below we follow the definitions and results from~\cite{V12}, pointing out differences where appropriate.  

\subsection{Process Model}
\label{ss:pm}
Let $P, Q, \ldots$ range over processes, $x, y, \ldots$ over channel names (or \emph{session endpoints}), and $v, v', \ldots$ over values; 
for simplicity, the sets of values and channels coincide.
In examples, we use $\unit$ to denote a terminated channel that cannot be further used.

\begin{figure}[t]
\begin{mdframed}
  \begin{displaymath}
   \begin{array}[t]{rllllll}
	P,Q ::= 	& \ov x\out v.P				& \mbox{(output)} &\midd 
	          	& \nil	     					& \mbox{(inaction)} &
	          	\\[1mm] 
           		\midd & x\inp y.P					& \mbox{(input)}     &
           		\midd  
           		& P \pp Q  	    		  		& \mbox{(composition)} &  
           		\\[1mm]
           		\midd & \selection x {l_j}.P 			& \mbox{(selection)}  & 
           		\midd
			& \res {xy}P					& \mbox{(session restriction)}&
			\\[1mm]
         	\midd	& \branching xlP				& \mbox{(branching)} &&& 
         	\\[1.5mm]

	v ::=		& x 							& \mbox{(channel)}&
  \end{array}
  \end{displaymath}
\vspace{1.0em}
  \begin{displaymath}
    \begin{array}{rll}  
	\Did{R-Com} & 
		\res {xy}(\ov x\out v.P\pp y\inp z.Q) 
		\to 
		\res {xy}(P\pp Q\substj{v}{z})
		\\[1mm]
	    \Did{R-Case} &
		\res {xy}(\selection x{l_j}.P \pp \branching ylP) 
		\to 
		\res {xy}(P\pp P_j)
	~~  j\in I \qquad
      \\[1mm]      
		\Did{R-Par}&		{P\to Q}\Longrightarrow
							{P\pp R\to Q\pp R}
	\\[1mm]
	\Did{R-Res}&	{P\to Q} \Longrightarrow
							{\res {xy} P\to \res {xy} Q}
       \\[1mm]
	\Did{R-Str}& 	{P\equiv P',\ P\to Q,\ Q'\equiv Q}
								\Longrightarrow
								{P'\to Q'} & 
    \end{array}
\end{displaymath}
\end{mdframed}
  \caption{The session $\pi$-calculus: syntax (top) and semantics (bottom).}
  \label{fig:sessionpi}
\end{figure}

We briefly comment on the syntax of processes, given in  \Cref{fig:sessionpi} (upper part).
The main difference with respect to the syntax in~\cite{V12} is that we do not consider 
boolean values nor conditional processes (if-then-else).  
Process $\ov x\out v.P$ denotes the output of  $v$ along  $x$, with continuation $P$.
Dually, process $x\inp y.P$ denotes an input along  $x$ with continuation $P$, with 
$y$ denoting a placeholder.
Rather than the non-deterministic choice of the untyped $\pi$-calculus~\cite{DBLP:journals/iandc/MilnerPW92a}, the session $\pi$-calculus includes
operators for (deterministic) internal and external labelled choices, 
denoted $\selection x {l_j}.P$ and $\branching xlP$, respectively.
Process $\selection x {l_j}.P$ uses  $x$ to select $l_j$ from a labelled choice process
$\branching xlP$, so as to trigger $P_j$; labels are indexed by the finite set $I$ and are pairwise distinct.
We also have the 
	inactive process (denoted $\nil$),
	the parallel composition of $P$ and $Q$ (denoted $P \pp Q$), 
	and the \emph{double restriction} operator, noted $\res{xy}P$:
	the intention is that $x$ and $y$ denote \emph{dual} session endpoints with scope $P$.
	We omit $\nil$ whenever possible and write, e.g., $\ov x\out \unit$
	instead of $\ov x\out \unit.\nil$.
	We will write \procs to denote the session $\pi$-calculus processes generated by the grammar in \Cref{fig:sessionpi}.
 
Notions of free and bound names in processes are exactly as in~\cite{V12}. That is, name $y$ is bound in $x\inp y.P$ and names $x$ and $y$ are bound in $\res{xy}P$. 
A name that does not occur bound within a process is said to be free; the set of free names of $P$ is denoted $\fv{P}$. We write $P\substj{v}{z}$ to denote the (capture-avoiding) substitution of free occurrences of $z$ in $P$ with $v$. Finally, we follow   \emph{Barendregt's variable convention}, whereby all names in binding occurrences in any mathematical context are pairwise distinct and also distinct from the free names.

The operational semantics is given in terms of a reduction relation, noted $P\to Q$, and defined by the rules in \Cref{fig:sessionpi} (lower part).
Reduction relies on a standard notion of \emph{structural congruence}, noted $\equiv$, 
defined by the following axioms:
\begin{align*}
& P \pp Q \equiv Q \pp P 
\qquad
(P \pp Q) \pp R \equiv  P \pp (Q \pp R)
\qquad
P \pp \nil \equiv P
\\
& \res{xy}P \pp Q \equiv \res{xy}(P \pp Q)
\qquad
\res{xy}\nil \equiv \nil
\qquad
\res{wx}\res{yz}P \equiv \res{yz}\res{wx}P
\end{align*}
Key rules in \Cref{fig:sessionpi}  are~\Did{R-Com} and \Did{R-Case}, denoting the interaction of output/input prefixes and  selection/branching constructs, respectively. 
Observe that interaction involves prefixes with different channels (endpoints), and always occurs in the context of an outermost double restriction.
Rules~\Did{R-Par}, \Did{R-Res}, and~\Did{R-Str} are standard~\cite{V12}.
We write $\mto$ to denote the reflexive, transitive closure of $\to$. 

\subsection{Type System}\label{ss:typesess}
%
The syntax of {session types}, ranged over $T, S, \ldots$, is given by the following grammar.
$$
T,S ::=  \wn T.S \midd \oc T.S \midd \branch lS \midd \select lS \midd \nilT
$$
The type $\wn T.S$ is assigned to an endpoint that first receives a value of type $T$ and then continues according to the protocol described by~$S$.
Dually, type $\oc T.S$ is assigned to an endpoint that first outputs a value of type $T$ and then continues according to the protocol described by $S$.
Type
$\branch lS$ is used for {external choices}, and generalises input types;
dually, 
type 
$\select lS$ is used for {internal choices}, and generalises   output types.
Finally, $\nilT$ is the type of an endpoint with a terminated protocol.
Notice that session types describe protocols as \emph{sequences} of actions; they 
do not admit parallel operators.

With respect to the syntax of types in~\cite{V12},  we only consider channel endpoint types (no ground types such as \textsf{bool}). 
Also, types in the system in~\cite{V12} can be qualified as either \emph{linear} or \emph{unrestricted}. 
Our session types are  {linear}---the only {unrestricted} session type is $\nilT$. 
Focusing on linear types suffices for our purposes, and leads to simplifications in typing rules and auxiliary notions, such as well-formedness (see below).

A central notion in session-based concurrency is  \emph{duality}, 
which relates session types offering opposite (i.e., complementary) behaviors; it stands at the basis of communication safety and session fidelity.
Given a session type $T$, its dual type $\dual{T}$ is defined as follows:
\begin{displaymath}
  \begin{array}{rll}
	\dual{\nilT}				&\defeq& \nilT\\
	\dual{\oc T.S}			&\defeq& \wn T.\dual{S}\\
	\dual{\wn T.S}			&\defeq& \oc T.\dual{S}\\
	\dual{\select lS}			&\defeq& \branch l{\dual S} \\
	\dual{\branch lS}			&\defeq& \select l{\dual S} 
  \end{array}
\end{displaymath}

%

\begin{figure}[t]
\begin{mdframed}
$$\begin{array}{c}
%
\inferrule[\Did{T-Nil}]
{\eend{\Gamma} }
{\Gamma \s \nil}
\qquad
\inferrule[\Did{T-Par}]
{\Gamma_1\s P\quad
\Gamma_2\s Q}
{\Gamma_1,\Gamma_2\s P\pp Q}
\qquad
\inferrule[{\Did{T-Res}}]
	{
	    \Gamma, x:T, y:\ov{T} \s P
	}
	{ 
	     \Gamma \s \res{xy} P
	}
\\\\
%
\inferrule[{\Did{T-In}}]
	{
 	  \Gamma, x:S, y:T \s P
	}
	{
	  \Gamma, x : \wn T.S \s x\inp{y}.P
	}
\qquad
\inferrule[{\Did{T-Out}}]
	{
	   \Gamma, x:S \s P
	}
	{
	   \Gamma,  x : \oc T.S, y:T \s \ov x\out{y}.P
	}
\\\\
\inferrule[{\Did{T-Brch}}]
	{
		\Gamma, x: S_i \s P_i \quad  \forall i\in I
	}
	{
		\Gamma, x:\branch lS \s \branching {x}lP
	}
\qquad
\inferrule[{\Did{T-Sel}}]
	{
		\Gamma,  x: S_j \s P \quad   \exists j \in I
	}
	{ 	
		\Gamma, x:  \select lS \s \selection {x} {l_j}.P
	}
  \end{array}$$
  \end{mdframed}
\caption{Typing rules for the session \p calculus.}
\label{fig:sess_typing}
\end{figure}

Typing contexts, ranged over by $\Gamma, \Gamma'$, are 
produced by the following grammar:
$$
\Gamma, \Gamma' ::= \emptyset \ \midd \ \Gamma, x:T
$$
where `$\emptyset$' denotes the empty context. 
We standardly require the variables appearing in a context to be pairwise distinct.
It is often convenient to treat typing contexts
as sets of typing assignments $x:T$. This way, e.g., we write $x: T \in \Gamma$ if $\Gamma = \Gamma', x:T$, for some $\Gamma'$.
We write $\eend{\Gamma}$ if and only if $x: T \in \Gamma$ implies $T = \nilT$.
We sometimes write $\Gamma^{\eende}$ to indicate that $\eend{\Gamma}$.
Given a context $\Gamma$ and a process $P$, a session typing judgment is of the form $\Gamma \s P$.
If $\Gamma$ is empty, we write $~\s P$.

Typing rules are given in \Cref{fig:sess_typing}.
Rule~\Did{T-Nil} states that $\nil$ is only well-typed under a fully terminated context.
Rule~\Did{T-Par} types the parallel composition of two processes by composing their corresponding typing contexts.\footnote{In the presence of {unrestricted types}, as given in~\cite{V12}, 
Rule \Did{T-Par} requires 
 a \emph{splitting operator}, noted $\circ$ in~\cite{V12}. However, since  we consider only  {linear session types}, the $\circ$ operator boils down to `$,$'.}
Rule~\Did{T-Res} types a restricted process by requiring that the two endpoints have dual types.
Rules~\Did{T-In} and \Did{T-Out} type the receiving and sending of a value over a channel $x$, respectively.
Finally, Rules~\Did{T-Brch} and \Did{T-Sel} are generalizations of input and output over a labelled set of processes.

The main guarantees of the type system in~\cite{V12} are \emph{communication safety} and \emph{session fidelity}, i.e., typed processes respect their ascribed protocols, as
represented by session types.
We have the following results:

\begin{theorem}[Strengthening -- Lemma 7.3 in~\cite{V12}]\label{thm:strength}
Let $\Gamma\s P$ and 
$x \not\in \fv{P}$.
If $\Gamma = \Gamma', x:T$ then $\Gamma'\s P$.
\end{theorem}

\begin{theorem}[Preservation for $\equiv$ -- Lemma 7.4 in~\cite{V12}]\label{thm:subj_cong}
If $\Gamma\s P$ and $P\equiv Q$, then $\Gamma\s Q$.
\end{theorem}

\begin{theorem}[Preservation  -- Theorem 7.2 in~\cite{V12}]\label{thm:subj_red}
If $\Gamma\s P$ and $P\to Q$, then $\Gamma\s Q$.
\end{theorem}

Following~\cite{V12}, we say that processes of the form
$\ov x\out v.P$, $x\inp z.Q$, 
$\selection x{l_j}.P$, and $\branching xlP$
are \emph{prefixed at $x$}. 
We call \emph{redexes} processes of the form $\ov x\out v.P\pp y\inp z.Q$ 
and 
$\selection x{l_j}.P \pp \branching ylP$.
The following notion of well-formed processes, a specialization of the definition in~\cite{V12}, is key to single out meaningful typed processes:
\begin{definition}[Well-Formedness]\label{def:well-form_sessions.}
A process is {\em well-formed} if for each of its structural congruent processes of the form
$\res{{x_1y_1}}\ldots\res{{x_ny_n}}(P\pp Q \pp R)$ the following condition holds.
\begin{enumerate}[label=$\bullet$]

\item
If $P$ is prefixed at $x_1$ and $Q$ is prefixed at $y_1$
then $P\pp Q$ is a redex.
\end{enumerate}
\end{definition}
%
%
%
%

We have the following result: 
\begin{theorem}[Safety -- Theorem 7.3 in~\cite{V12}]\label{thm:safety}
If $~\s P$ then $P$ is well-formed.
\end{theorem}

Therefore, if  $~\s P$ and $P$ reduces to $Q$ in zero or more steps, then $Q$ is well-formed; this is Theorem 7.1 in~\cite{V12}.


\rev{We close by introducing some useful terminology}:

\begin{definition}[(Un)Composable Processes]
\label{d:openclose}
Let $\Gamma \s P$.
If $\Gamma$ is empty, we say that $P$ is  \emph{uncomposable}; 
otherwise, if $\Gamma$ is non-empty, we say $P$ is \emph{composable}.
\end{definition}

\rev{We use the adjectives composable and uncomposable to distinguish typable processes depending on their associated typing context. Note that the adjectives \emph{open} and \emph{closed} have their usual meaning, associated to the free names of a process, possibly untyped.  As an example, process $\nil$ is a closed process that can be typed under a non-empty typing context, and so it is composable.}

\subsection{Deadlock Freedom}
As already motivated, a desirable liveness property for session $\pi$-calculus processes is that they should never ``get stuck''.
Unfortunately, the session type system given in~\cite{V12} (and summarized above) does not exclude deadlocked processes.
Intuitively, this is because typed processes may contain cyclic causal dependencies enforced by communication prefixes in processes but not described by their session types. 
Indeed, a particularly insidious class of deadlocks is due to cyclic interleaving of channels in processes, as illustrated by following example.

\begin{example}[A Deadlocked Process]\label{ex:basic}
Process  
%
$
P\defeq \res{xy}\res{wz}(\ov x\out \unit.\ov w\out \unit\pp z\inp t.y\inp s)
$
represents the implementation of two independent sessions, $xy$ and $wz$, which get intertwined (blocked) due to the 
nesting induced by input and output prefixes.
Process $P$ is well-typed in~\cite{V12} under $\unit:\nilT\s P$, even if $P$ is unable to reduce.
\end{example}

Below we  define deadlock freedom 
in the session $\pi$-calculus by
following~\cite[Def.\,5.2]{DBLP:journals/iandc/0001L17}: 

\begin{definition}[Deadlock Freedom]\label{def:lock}
A process $P$ is \emph{deadlock-free} if the following condition holds:
whenever $P \mto P'$ and one of the following holds
\begin{itemize}
\item $P' \equiv \res{\widetilde{xy}}(\ov x\out v.Q_1 \pp Q_2)$
\item $P' \equiv \res{\widetilde{xy}}(x\inp y.Q_1 \pp Q_2)$
\item $P' \equiv \res{\widetilde{xy}}(\selection x {l_j}.Q_1 \pp Q_2)$
\item $P' \equiv \res{\widetilde{xy}}(\branching {x}lP \pp Q)$
\end{itemize} 
then there exists $R$ such that $P' \to R$.
\end{definition}

\begin{remark}[Defining Deadlock Freedom]
\label{r:df}
\Cref{def:lock} is closely related to the definition of deadlock and deadlock freedom in~\cite{K02} (Definition 2.4), which states that  a process $P$ is in deadlock if it reaches one of the first two items stated in \Cref{def:lock} and cannot reduce from there. Then, a process $P$ defined as deadlock-free if it never reduces to a   deadlocked process.
We shall be following the type system in~\cite{K06}, where the notion of deadlock freedom is defined informally: a process is deadlock-free if 
when ``given a request, it will eventually returns a result unless the process diverges''.
\end{remark}

%
%

\begin{example}[A Deadlock-Free Process]
\label{ex:deadlock-free}
It is easy to see that process $P$  from \Cref{ex:basic} is not deadlock-free as per \Cref{def:lock}.
A deadlock-free variant of process $P$  would be
$P'\defeq \res{xy}\res{wz}(\ov x\out \unit.\ov w\out \unit \pp y\inp s.z\inp t)$, which also is typable: $\unit:\nilT\s P'$.
Observe how the difference between $P$ and $P'$ is in the parallel component on the right-hand side: the two input prefixes have been swapped.
\end{example}



\section{Two Approaches to Deadlock Freedom}\label{s:deadlf}
We introduce two approaches to typing deadlock-free  processes.
The first comes from  interpretations of linear logic propositions as 
session types~\cite{CairesP10,DBLP:journals/mscs/CairesPT16,DBLP:conf/icfp/Wadler12} (\Cref{ss:lf}).
The second approach exploits encodings of session processes and types~\cite{DGS12} into the linear types with \emph{usages} for the $\pi$-calculus (\Cref{ss:dgs}). Based on these approaches, in~\Cref{s:hier} we will formally define the classes \lcp and \fullKoba mentioned in the introduction.

\subsection{Linear Logic Foundations of Session Types}\label{ss:lf}
The linear logic interpretation of session types was introduced by Caires and Pfenning~\cite{CairesP10,DBLP:journals/mscs/CairesPT16}, and developed
by Wadler~\cite{DBLP:conf/icfp/Wadler12} and others. Here we consider 
an interpretation 
based on classical linear logic (\CLL) with mix principles, following~\cite{CairesCFest14,CP17}.

The syntax and semantics of processes are as in \Cref{s:sessions} with the following differences.
First, we have a single restriction construct $\res{x}P$ instead of the
double restriction $\res{xy}P$.
Second, we have a \emph{forwarding process}, denoted $\linkr{x}{y}$, which intuitively ``fuses'' or ``links'' channels/names $x$ and $y$. More formally, we have:
\begin{displaymath}
   \begin{array}[t]{rllllll}
	P,Q ::= 	\ov x\out v.P \midd x\inp y.P \midd \selection x {l_j}.P \midd \branching xlP \midd \res x P \midd \linkr{x}{y} \midd P \pp Q \midd
	\nil
  \end{array}
  \end{displaymath}
\noindent
In what follows, the bound output $\res{y}\ov{x}\out y.P$ will be abbreviated as $\bout{x}{y}P$.
Also, we write $\res{\wt{x}} P$ to abbreviate $\res{x_1}\ldots\res{x_n}P$.

Differences in the reduction rules
are summarized in \Cref{fig:redlcp}.
In particular, observe how interaction of input/output prefixes and of selection/branching constructs is no longer covered by an outermost restriction.

\begin{figure}[t!]
\begin{mdframed}
  \begin{displaymath}
    \begin{array}{rll}  
	\Did{R-ChCom} & 
		\ov x\out v.P\pp x\inp z.Q
		\to 
		P\pp Q\substj{v}{z}
		\\[1mm]
	\Did{R-Fwd} & 
		\res {x}(\linkr{x}{y} \pp P)
		\to 
		P\substj{y}{x}
      	\\[1mm]
	\Did{R-ChCase} &
		\selection x{l_j}.P \pp \branching xlP
		\to 
		P\pp P_j
	\quad  j\in I 
       \\ [1mm]
	\Did{R-ChRes}&	{P\to Q} \Longrightarrow
							{\res {x} P\to \res {x} Q}
    \end{array}
\end{displaymath}
\end{mdframed}
  \caption{Reduction rules for processes in  \lcp.}
  \label{fig:redlcp}
\end{figure}

As for the type system, we consider the so-called \emph{linear logic types}
which correspond to linear logic propositions (without exponentials). They
 are given by the following grammar:
$$
A,B ::= \bot \midd \one \midd A\otimes B  \midd A \parl B \midd \branch lA  \midd  \select lA  
$$
Intuitively, $\bot$ and $\one$ are used to type a terminated endpoint.
Type $A\otimes B$ is associated to an endpoint that first outputs an object of type $A$ and then 
behaves according to $B$. Dually, type 
$A \parl B$ is the type of an endpoint that first inputs an object of type $A$ and then continues as $B$.
The interpretation of  $\branch lA$ and $\select lA$ as types for branching and selection behaviors is precisely as in session types (cf. \Cref{ss:typesess}).

A full duality on linear logic types 
corresponds to the negation operator of \CLL\ $(\cdot)^\perp$.
The \emph{dual} of type $A$, denoted $\dual{A}$, is inductively defined
as follows:
\begin{displaymath}
\begin{array}{rll}
\dual{\one} &\defeq& \bot 
\\ \dual{\bot} &\defeq&  \one 
\\ \dual{A\otimes B}  &\defeq&  \dual{A} \parl \dual{B} 
\\ \dual{A\parl B}  &\defeq&   \dual{A} \otimes \dual{B} 
\\ \dual{\branch lA}  &\defeq&  \select l{\dual A}
\\ \dual{\select lA}  &\defeq&   \branch l{\dual A} 
\end{array}
\end{displaymath}
  
\noindent
Recall that $A \lolli B \triangleq \dual{A} \parl B$.
As explained in~\cite{CairesCFest14},
considering mix principles means admitting $\bot \lolli \one$ and
$\one \lolli \bot$, and therefore $\bot = \one$.
We write $\bullet$ to denote either $\bot$ or $\one$,
and therefore $\dual{\bullet} = \bullet$. 
\rev{That is, we consider the \emph{conflation} of dual types $\bot$ and $\one$ as explored by Atkey et al.~\cite{DBLP:conf/birthday/AtkeyLM16}.} 

\begin{figure*}[t]
\begin{mdframed}
    \vspace{-1em}
\centering
{
$$
\begin{array}{c}
\inferrule*[left=\Did{T-$\one$}]{
}{ \nil \cp x{:}\bullet}
\qquad
\inferrule*[left=\Did{T-$\bot$}]
{P \cp  \Delta}
{P \cp x{:}\bullet, \Delta}
\qquad
\inferrule*[left=\Did{T-$\tid$}]{
}
{\linkr{x}{y} \cp x{:}A, y{:}\dual{A}}
\\\\
\inferrule[\Did{T-$\parl$}]
{P \cp \Delta, y{:}A, x{:}B}
{x\inp {y}.P \cp \Delta, x{:} A\parl B}
\qquad
\inferrule[\Did{T-$\otimes$}]
{P \cp  \Delta, y{:}{ A}  \and Q \cp \Delta', x{:}B}
{\bout{x}{y}. (P\para Q) \cp  \Delta, \Delta', x{:}  A\otimes B}
\\\\
\inferrule[\Did{T-$\oplus$}]
{P \cp \Delta, x{:}A_j  \and j\in I}
{ \selection {x} {l_j}.P  \cp \Delta, x{:}\select lA}
\qquad
\inferrule[\Did{T-$\with$}]
{P_i \cp \Delta, x{:}A_i  \and \forall i\in I}
{\branching xlP \cp \Delta, x{:}\branch lA}
\\\\
\inferrule[\Did{T-$\cut$}]
{P \cp \Delta, x{:}\dual{A}  \and Q \cp  \Delta', x{:}A}
{\res {x}(P\para Q) \cp\Delta, \Delta'}
\qquad
\inferrule[\Did{T-$\mix$}]
{P \cp \Delta  \and Q \cp  \Delta'}
{ P\para Q \cp\Delta, \Delta'}
\end{array}
$$
}
\end{mdframed}
\caption{\label{fig:type-system-cll}
Typing rules for the \p calculus with linear logic types.}
\end{figure*}

Typing contexts, ranged over $\Delta, \Delta', \ldots$, are produced by the following grammar:
$$
\Delta, \Delta' ::= \cdot \ \midd \ \Delta, x:A
$$
where  `$\,\cdot\,$' denotes the empty typing context.

Typing judgments are of the form $P \cp \D$.
\Cref{fig:type-system-cll} gives the corresponding typing rules.
One salient point is Rule~$\Did{T-$\cut$}$, which types two processes that have exactly one channel of dual type in common ($x$ in the rule) by composing them in parallel and immediately restricting this common channel.
This implements the  \emph{``composition plus hiding''} principle,  
which monolithically integrates parallel composition and restriction in a single rule.
Indeed, unlike the system in~\cite{V12}, there is no dedicated rule for restriction, which also appears in the bound output induced by Rule~\Did{T-$\otimes$}.
Also, Rule~\Did{T-$\mix$} conveniently enables to type the 
\emph{independent} parallel composition of processes, i.e., the composition of two processes that do not share any sessions.
\rev{(This is referred to as \emph{communication-free concurrency} in \cite{DBLP:conf/birthday/AtkeyLM16}.)}
\nurev{Following \Cref{d:openclose}, we say that process $P \cp \D$ is composable if $\Delta \neq \cdot$ and uncomposable otherwise.}

We now collect main results for this type system; see~\cite{DBLP:journals/mscs/CairesPT16,CairesCFest14,CP17} for details.
We first state type preservation:
\begin{theorem}[Type Preservation]
\label{thm:LL-type-preservation}
If $P \cp \D$ and $P \redd Q$ then $Q \cp \D$.
\end{theorem}

We now state deadlock freedom.
For any $P$, define $live(P)$ if and only if $P \equiv (\nub \wt{n})(\pi.Q \para R)$,  where $\pi$ is an input, output, selection, or branching prefix. 

\begin{theorem}[Deadlock Freedom]\label{t:progress}
If $ P \cp \cdot $ and $live(P)$ 
then $P \redd Q$,  for some $Q$.
\end{theorem}

\subsection{Deadlock Freedom by Encodability}\label{ss:dgs}
The second approach to deadlock-free session processes is \emph{indirect},
in  that  
establishing deadlock freedom for session processes
appeals to encodings into a dyadic $\pi$-calculus whose type system 
enforces deadlock freedom by 
exploiting \emph{usages}, \emph{obligations}, and \emph{capabilities}~\cite{K02,K06}.

We follow closely the definitions and results in~\cite{K06}.
Next, we introduce 
the syntax of the (dyadic) $\pi$-calculus 
(\Cref{sss:proc}), 
types with usages (\Cref{sss:usage}), 
typing rules (\Cref{sss:usage}), 
and 
finally the technical results leading to deadlock freedom (\Cref{sss:techres}).
The encodings of session processes into dyadic processes and of session types into  types with usages are given in \Cref{sss:enco}.

\subsubsection{The Dyadic $\pi$-calculus}\label{sss:proc}
\paragraph{Syntax}
The syntax of the dyadic \p calculus is as follows:
  \begin{displaymath}
   \begin{array}[t]{rllllll}
	v ::=		&x 							& \mbox{(channel)} & \midd 
			&\vv {j}v						& \mbox{(variant value)}
\\ [2mm]
	P,Q ::= 	& \ov x\out {\wt v}.P			& \mbox{(output)} & 
	          	\midd & \nil	     					& \mbox{(inaction)} &
	          	\\[1mm] 
           		\midd & x\inp {\wt z}.P				& \mbox{(input)}     
           		& 
           		\midd & P \pp Q  	    		  		& \mbox{(composition)} &  \\[1mm]
           	\midd 	& \picase v{x_i}{P_i} 			& \mbox{(case)}  & 
			\midd & \res {x}P					& \mbox{(session restriction)}&
  \end{array}
  \end{displaymath}
We discuss differences with respect to~\Cref{s:sessions} and~\cite{K06}:
\begin{itemize}
	\item 
While the syntax of processes given in \Cref{s:sessions} (and in~\cite{K06}) supports monadic communication, we consider \emph{dyadic} communication: 
an output prefix involves a tuple of values $v_1, v_2$  
and an input prefix involves 
a tuple of variables $z_1,  z_2$.
For the sake of notational uniformity, we write $\wt{v}$ and $\wt{z}$ 
to stand for $v_1, v_2$  and $z_1,  z_2$, respectively. 
Dyadic communication is convenient, as the encoding of session processes into ``standard'' $\pi$-calculus processes \cite{DGS12} requires transmitting tuples whose length is at most two. 

\item While in~\cite{K06} input and output prefixes are annotated with a capability annotation $t \in \mathbb{N} \cup \infty$, we omit such annotations to enhance clarity.
\item Rather than the branching and selection constructs in \Cref{s:sessions}  (and the if-then-else process considered in~\cite{K06}), in the dyadic $\pi$-calculus presented above we have the \emph{case construct} $\picase v{x_i}{P_i}$, which uses the \emph{variant value} $\vv {j}v$~\cite{Sangio01}.
\item We do not consider the let construct and replication processes in~\cite{K06}.
\item In line with~\Cref{ss:lf} and \cite{K06} (but differently from \Cref{s:sessions}) we use $\res x P$ as restriction operator.
\end{itemize}

As before, we write $\res{\wt{x}}P$ to denote the process $\res{x_1}\cdots\res{x_n}P$.
Notions of free and bound names are as usual; we write $\fv{P}$ to denote the set of free names of $P$.

\paragraph{Reduction Semantics}
Following~\cite{K06}, the reduction semantics of dyadic $\pi$-calculus processes relies on the structural relation $\preceq$:

\begin{definition}[Def A2 in~\cite{K06}]
\label{d:kobastruct}
The \emph{structural relation} $\preceq$ is 
the least reflexive and transitive
relation closed under the following rules (where $P \equiv Q$ denotes $(P \preceq Q) \land (Q \preceq P))$:
\begin{align*}
& 
P \pp Q \equiv Q \pp P 
\qquad\qquad
(P \pp Q) \pp R \equiv  P \pp (Q \pp R)
\\
&
P \pp \nil \equiv P
\qquad\qquad
\res{x}P \pp Q \equiv \res{x}(P \pp Q) \quad \text{if $x$ is not free in $Q$}
\\
&
\res{x}\nil \equiv \nil
\qquad\qquad
\res{x}\res{y}P \equiv \res{y}\res{x}P
\\
&
{P\preceq Q}\Longrightarrow
							{P\pp R\preceq Q\pp R}
\qquad\qquad
{P\preceq Q}\Longrightarrow
							{\res{x}P \preceq \res{x}Q}
\end{align*}
\end{definition}
\noindent Thus,  $P \preceq Q$ intuitively means that $P$ can be restructured into $Q$ by using the above rules.
We shall refer to $\equiv$ as structural equivalence.
The reduction rules are then as follows:
  \begin{displaymath}
 \begin{array}{lllll}  
	\Did{R\p Com} & 
		\ov x\out {\wt v}.P\pp x\inp {\wt z}.Q
		\to 
		P\pp Q\substj{\wt v}{\wt z}
	\\[1mm]
	\Did{R\p Case} &
		 \picase {\vv jv}{x_i}{P_i}	
		\to 
		P_j\substj{v}{x_j}
	\quad  j\in I 
	\\[1mm]
			\Did{R\p Par}&		{P\to Q}\Longrightarrow
							{P\pp R\to Q\pp R}
	\\[1mm]
	\Did{R\p Res}&	{P\to Q} \Longrightarrow
							{\res {x} P\to \res {x} Q}
       \\[1mm]
	\Did{R\p Str}& 	{P\preceq P',\ P\to Q,\ Q'\preceq Q}
								\Longrightarrow
								{P'\to Q'} & 
    \end{array}
\end{displaymath}
Rules are self-explanatory. 
\rev{We only discuss Rule~\Did{R\p Case}, which is the main difference with respect to the reduction semantics in~\cite{K06}. 
Note that by using the variant value $\vv jv$, this rule simultaneously selects $P_j$ \emph{and} induces the substitution $\substj{v}{x_j}$. This is different from Rule \Did{R-Case} for selection and branching (cf. \Cref{fig:sessionpi}), which  selects a branch but does not involve a communication.}
As before, we write $\mto$ to denote the reflexive, transitive closure of $\to$.

Since the definition of deadlock freedom in~\cite{K06} is only informal (cf. \Cref{r:df}), we shall adopt the following 
definition, which mirrors \Cref{def:lock}:

\begin{definition}[Deadlock Freedom]\label{def:lockpi}
A process $P$ is \emph{deadlock-free} if the following condition holds:
whenever $P \mto P'$ and one of the following holds
\begin{itemize}
\item $P' \equiv \res{\wt{x}}(\ov x\out {\wt v}.Q_1 \pp Q_2)$
\item $P' \equiv \res{\wt{x}}(x\inp{\wt{y}}.Q_1 \pp Q_2)$
\item $P' \equiv \res{\wt{x}}(\picase {\vv jv}{x_i}{P_i} \pp Q)$
\end{itemize} 
then there exists $R$ such that $P' \to R$.
\end{definition}

%

\subsubsection{Types with Usages}\label{sss:usage}
The type system for deadlock freedom in~\cite{K06} exploits types with usages. 
Usages rely on \emph{obligations} and \emph{capabilities}, which are endowed with  \emph{levels} to describe inter-channel dependencies:
\begin{enumerate}[label=$\bullet$]
\item
An obligation of level $n$ must be fulfilled by using only capabilities of level \emph{less than} $n$.
Said differently, an action of obligation $n$ may be prefixed by actions of capabilities less than $n$.

\item
For an action with capability of level $n$, there must exist a co-action with obligation of level \emph{less than} $n$ or \emph{equal to} $n$.
\end{enumerate}

We shall rely on usages as defined next, a strict subset of those defined in~\cite{K06}.

\begin{definition}[Usages]
\label{d:usages}
The syntax of usages $U, U', \ldots$ is defined  by the following grammar:
  \begin{displaymath}
  \begin{array}[t]{rllllll}
	U ::=
				&  {\zusage}  			&  \mbox{(not usable)} 
				\\
                             \ \midd \      &  \wn^{\ob}_{\ca}.U  		& \mbox{(used in input)}
\\[1mm]
                          \ \midd \   & \oc^{\ob}_{\ca}.U  	&  \mbox{(used in output)}
                          \\[1mm]
                             \ \midd \ 
                             & (U_1 \pp U_2) 			&  \mbox{(used in parallel)}
                             \\[1mm]
                            \ \midd \ &  \uparrow^{\,t} U & \mbox{(lift obligation levels of $U$ up to $t$)}
\end{array}
  \end{displaymath}
  where the  \emph{obligation} $\ob$ and the \emph{capability} $\ca$ range over the set $\mathbb{N} \cup \infty$.
  We shall refer to usages generated with the first three productions above (not usable, input, output) as \emph{sequential usages}.
\end{definition}
  
Usage $\zusage$ describes a channel that cannot be used at all. 
A usage $\wn^{\ob}_{\ca}.U$ (resp. $\oc^{\ob}_{\ca}.U$)
is associated to a channel that can be used once for input (resp. output)   and then according to usage $U$. 
The usage $U_1\pp U_2$ can be associated to a channel that is used according to $U_1$ and $U_2$, possibly in parallel.
The usage $\uparrow^{\,t} U$ acts as an operator that lifts the obligation levels in $U$ up to $t$. 
We let $\alpha$ range over `$\wn$' and `$\oc$'.
We will often omit $\zusage$, and so we will write, e.g., $\alpha^{\ob}_{\ca}$ instead of $\alpha^{\ob}_{\ca}.\zusage$.

\begin{notation}[Co-actions]
\label{not:coact}
We write $\ov{\alpha}$ to denote the \emph{co-action} of $\alpha$, i.e., $\ov{\oc} = \wn$ and $\ov{\wn} = \oc$.
\end{notation}


We rely on a number of auxiliary definitions for usages; they all follow~\cite{K06}:

\begin{definition}[Capabilities and Obligations]
\label{def:obligations}
Let $U$ be a usage.
The input and output \emph{capability levels} (resp. \emph{obligation levels}) of $U$, written
 $\caps{\wn}{U}$ and $\caps{\oc}{U}$
(resp. $\obss{\wn}{U}$ and $\obss{\oc}{U}$),
are defined as:
$$
\begin{array}{rclcrcl}
\caps\alpha{\zusage}		 &\defeq& \infty
& \quad & 
\obss\alpha{\zusage}		 &\defeq&   \infty 
\\
\caps\alpha{\ov{\alpha}^{\ob}_{\ca}.U}		 &\defeq&\infty 
& \quad & 
\obss\alpha{\ov{\alpha}^{\ob}_{\ca}.U}		 &\defeq&   \infty 
\\
\caps\alpha{\alpha^{\ob}_{\ca}.U}		 &\defeq& \ca 
& \quad & 
\obss\alpha{\alpha^{\ob}_{\ca}.U}		 &\defeq&   \ob 
\\
\caps\alpha{U_1\pp U_2}				 &\defeq&  \min(\caps\alpha{U_1},\caps\alpha{U_2})
& ~~& 
\obss\alpha{U_1\pp U_2}				 &\defeq&  \min(\obss\alpha{U_1},\obss\alpha{U_2})
\\
\caps\alpha{\uparrow^{\,t} U} &\defeq&  \caps\alpha U 
& ~~& 
\obss\alpha{\uparrow^{\,t} U} &\defeq&  \max(t, \obss\alpha U) 
\end{array}
$$
We write $\obss{}{U}$ for $\max(\obss{\wn}{U}, \obss{\oc}{U})$.
\end{definition}
%

%
\noindent
The reduction relation on usages, noted $U\to U'$, 
intuitively says that if a channel with usage $U$ is used for communication then
it should be used according to   $U'$ afterwards.
It relies on an auxiliary structural relation on usages.

\begin{definition}[Structural relation on usages \cite{K06}]
Let $\preceq$ be the least reflexive and transitive relation on usages defined by the following rules:
\begin{eqnarray*}
U_1 \pp U_2 \preceq U_2 \pp U_1
\quad
\uparrow^{\,t}(U_1 \pp U_2) \preceq (\uparrow^{\,t} U_1) \pp (\uparrow^{\,t} U_2) 
\quad
(U_1 \pp U_2) \pp U_3 \preceq U_1 \pp (U_2 \pp U_3)
\\
U_1 \preceq U'_1 \land U_2 \preceq U'_2 \Longrightarrow U_1 \pp U'_1 \preceq U_2 \pp U'_2
\quad
\uparrow^{\,t}{\alpha^{\ob}_{\ca}.U} \preceq \alpha^{\max(\ob,t)}_{\ca}.U
\quad 
U \preceq U' \Longrightarrow \uparrow^{\,t}U \preceq \uparrow^{\,t}U'
\end{eqnarray*}
\end{definition}


\begin{definition}
\label{def:usage_red}
The reduction relation $\to$ on usages is the smallest relation closed under the following rules:
\begin{displaymath}
\begin{array}{rll}
	\Did{U-Com} & 
		\wn^{\ob}_{\ca}.U_1 \pp \oc^{\ob'}_{\ca'}.U_2
		\to
		U_1\pp U_2
      \\[1mm]      
	\Did{U-Par} &
		U\to U'
		\Longrightarrow 
		U\pp U'' \to U'\pp U''
       \\[1mm]      
	\Did{U-SubStruct}&	{U\preceq U_1,\ U_1\to U_2,\ U_2\preceq U'}
								\Longrightarrow
								{U\to U'}
      \\[1mm]
\end{array}
\end{displaymath}
The reflexive, transitive closure of $\to$ (written $\mto$) 
is
defined as expected.
\end{definition}

The following key definition ensures that if some action has a capability of level $n$ then the obligation level of its co-actions should be at most $n$.
		
\begin{definition}[Reliability]\label{d:reli}
Let $\alpha$ and $\ov{\alpha}$ be co-actions (cf. \Cref{not:coact}).
We write $\conpar{\alpha}{U}$ when $\obss{\ov\alpha}{U} \leq \caps{\alpha}{U}$.
We write $\con{U}$ when  $\conpar\wn U$ and $\conpar\oc U$ hold.
Usage $U$ is \emph{reliable}, noted $\rel U$, if $\con {U'}$ holds 
for all $U'$ such that $U\mto U'$.
\end{definition}

\begin{example}
 We illustrate reliability and obligation/capability levels.
 Consider the usage $U = \wn^{\ob_1}_{\ca_1}.\zusage \pp !^{\ob_2}_{\ca_2}.\zusage$. 
 We establish the conditions required for $\rel{U}$ to hold. This in turn requires determining a few ingredients:
 \begin{itemize}
     \item $\obss{!}{U}  = \min(\obss{!}{\wn^{\ob_1}_{\ca_1}.\zusage},\obss{!}{!^{\ob_2}_{\ca_2}.\zusage})
          =  \min(\infty, \ob_2)   = \ob_2$      
     \item $\caps{\wn}{U}  = \min(\caps{\wn}{\wn^{\ob_1}_{\ca_1}.\zusage},\caps{\wn}{!^{\ob_2}_{\ca_2}.\zusage})
          =  \min(\ca_1, \infty)   = \ca_1$      
     \item $\obss{\wn}{U}  = \min(\obss{\wn}{\wn^{\ob_1}_{\ca_1}.\zusage},\obss{\wn}{!^{\ob_2}_{\ca_2}.\zusage})
          =  \min(\ob_1, \infty)   = \ob_1$      
     \item $\caps{!}{U}  = \min(\caps{!}{\wn^{\ob_1}_{\ca_1}.\zusage},\caps{!}{!^{\ob_2}_{\ca_2}.\zusage})
          =  \min(\infty, \ca_2 )   = \ca_2$      
 \end{itemize}
This way, $\conpar{\wn}{U}$ denotes $\ob_2 \leq \ca_1$; similarly, 
   $\conpar{\oc}{U}$ denotes $\ob_1 \leq \ca_2$.
 Then we have that $\con{U}$ holds when both $\ob_2 \leq \ca_1$ and $\ob_1 \leq \ca_2$ hold. 
 This is indeed the condition needed for  $\rel{U}$ to hold; notice that $U \to \zusage \pp \zusage$ is the only reduction possible, and that $\con{\zusage \pp \zusage}$ trivially holds.
 Notice that letting $\ob_1 = \ca_1 = \ob_2 = \ca_2 = 0$ suffices for $\rel{U}$ to hold.
\end{example}

\noindent
Having defined usages (and their associated notions), we move to define types.

\begin{definition}[Types with Usages]
\label{d:typesusages}
	The syntax of  types $\ctype, \ctype', \ldots$ builds upon usages as follows:
  \begin{displaymath}
  \begin{array}[t]{rllllll}
    	\ctype ::=
    				& \utype{U}{\widetilde{\ctype}}			& \mbox{(channel types)}
				\\
    			\ \midd \	& \variant {l_i}{\ctype_i}			& \mbox{(variant type)}
\end{array}
  \end{displaymath}
\end{definition}

Above, $\wt \ctype$ indicates a sequence of types of length at most two. 
Type $\utype{U}{\widetilde{\ctype}}$ is associated to a channel that behaves according to usage $U$ to exchange a tuple of values with types $\ctype_1,  \ctype_2$. 
Notice that 
$\widetilde{\ctype}$
can be empty; in that case, we write $\utype{U}{\temp}$: a channel with this type behaves according to $U$ without exchanging any values.

Differences with respect to the syntax of types in~\cite{K06} are: (i)~we do not consider Boolean nor product types, and (ii)~we consider the \emph{variant type} $\variant {l_i}{\ctype_i}$ from~\cite{Sangio01} to denote the disjoint union of labeled types, where labels $l_i$ ($i\in I$) are pairwise distinct. Variant types are essential to encode selection and branching in session types~\cite{DGS12}.

Typing contexts, ranged over $\Gamma, \Gamma', \ldots$, are 
produced by the following grammar:
$$
\Gamma, \Gamma' ::= \emp\ \midd \ \Gamma, x:\ctype
$$
where `$\emp$' denotes the empty context. 
Given a context $\Gamma = x_1:\ctype_1, \cdots, x_n:\ctype_n$, we write $\dom(\Gamma)$ to denote its domain, i.e., the set $\{x_1, \ldots, x_n\}$.

Following~\cite{K06}, we use `$\prec$' to denote a partial order that statically tracks the order in which channels are created. That is, $x \prec y$ means that $x$ was created more recently than $y$.
Typing judgments are indexed by `$\prec$'; they are of the form
$\Gamma \lf \wt{v}$ (for values) 
and 
$\Gamma \lf P$
(for processes).

Before commenting on the typing rules, given in \Cref{f:pityping}, we present some important auxiliary notions, extracted from~\cite{K06}.

\begin{definition}[Auxiliary Operators on Types]
\label{def:auxoper}
We define the following auxiliary operators:
\begin{enumerate}
\item
The unary operation $\uparrow^{\,t}$ on usages extends to types as follows:
$$\uparrow^{\,t} (\utype{U}{\wt{\ctype}}) = \utype{\uparrow^{\,t} U}{\wt{\ctype}}$$

\item 
The composition operation on types, denoted $\pp$, is defined as follows:
\begin{align*}
\utype{U_1}{\wt \ctype} \pp \utype{U_2}{\wt \ctype}
&\defeq \utype{(U_1\pp U_2)}{\wt \ctype}
\\
\variant {l_i}{\ctype_i} \pp \variant {l_i}{\ctype_i}
&\defeq \variant {l_i}{\ctype_i}
\end{align*}
The generalisation of $\pp$ to typing contexts, denoted $(\Gamma_1 \pp \Gamma_2)(x)$, is defined as follows:
$$
(\Gamma_1 \pp \Gamma_2)(x)
=
\begin{cases}
    \Gamma_1(x) \pp \Gamma_2(x) & \text{if $x \in \dom(\Gamma_1) \cap \dom(\Gamma_2)$}
\\
    \Gamma_1(x) & \text{if $x \in \dom(\Gamma_1) \setminus \dom(\Gamma_2)$}
\\
    \Gamma_2(x) & \text{if $x \in \dom(\Gamma_2) \setminus \dom(\Gamma_1)$}
\end{cases}
$$

\item 
The operator ``$\semi$''
combines a type assignment $x:\utype{\alpha^{\ob}_{\ca}}{\wt \ctype}$
and a context $\Gamma$ into a new context.
Precisely, $x:\utype{\alpha^{\ob}_{\ca}}{\wt \ctype} \semi \Gamma$ represents the context
 $\Gamma'$, defined 
as follows:
\begin{align*}
   \dom(\Gamma') &\defeq \{x\} \cup \dom(\Gamma) 
   \\
   \Gamma'(x) &\defeq\left\{ 
    \begin{aligned}
      &\utype{\alpha^{\ob}_{\ca}.U}{\widetilde \ctype}		&& \mbox{if 
$\Gamma(x)=\utype{U}{\widetilde \ctype}$}\\
      &\utype{\alpha^{\ob}_{\ca}}{\widetilde \ctype}		&& \mbox{if $x\notin 
\dom(\Gamma)$}
    \end{aligned}
  \right. 
  \\
   \Gamma'(y) &\defeq 
   \begin{cases}
   \uparrow^{\ca}\Gamma(y) & \text{if $y\neq x \land x \prec y$}
   \\
      \uparrow^{\ca+1}\Gamma(y) & \text{if $y\neq x \land x \not\prec y $}
      \end{cases}
\end{align*}
\end{enumerate}
\end{definition}

%

\begin{figure*}[t]
  \begin{mdframed}[style=newtight]
      \vspace{-1em}
\centering
$$\begin{array}{c}
\inferrule[{\Did{T$\pi$-Var}}]{
}{
 x:\ctype \lf x:\ctype
}
\quad
\inferrule[{\Did{T$\pi$-Tup}}]{
\Gamma_1 \lf v_1:\ctype_1 \and \Gamma_2 \lf v_2:\ctype_2  \and \wt{v} = v_1,   v_2~~\wt{\ctype} = \ctype_1, \ctype_2
}{
\Gamma_1 \pp \Gamma_2 \lf \wt{v}:\wt{\ctype}
}
\quad
\inferrule[{\Did{T$\pi$-LVal}}]{
      \Gamma \lf v: \ctype_j \quad \exists j\in I
}{
      \Gamma \lf   {l_j}\_v: \variant {l_i}{\ctype_i}
}
\\\\
\inferrule[{\Did{T$\pi$-Nil}}]{
}{
 \emp  \lf {\nil}
}
\qquad
\inferrule[{\Did{T$\pi$-Res}}]{
      \Gamma, x: \utype{U}{\widetilde \ctype} \lfpar{\prec \cup \{(x,y) \,|\, y \in \fv{P}\setminus \{x\} \}}{} P\quad \rel{U}
}{
      \Gamma \lf \res{x} P
}
\qquad
%
%
\inferrule[{\Did{T$\pi$-Par}}]{
	\Gamma_1 \lf P\and \Gamma_2 \lf Q
}{ 
      \Gamma_1 \pp \Gamma_2 \lf  P\pp Q
}
\\\\
\inferrule[{\Did{T$\pi$-Out}}]{
     \Gamma_1 \lf P \and   \Gamma_2 \lf \wt{v}: \wt{\ctype}
}{
       x: \utype{\oc^{0}_{\ca}}{\wt{\ctype}} \semi (\Gamma_1\pp \Gamma_2) \lf {\ov x}\out{ \tl v}.P
}

\qquad

\inferrule[{\Did{T$\pi$-In}}]{
      \Gamma, \widetilde y:\widetilde \ctype\lf P
}{
      x:\utype{\wn^{0}_{\ca}}{\wt{\ctype}} \semi \Gamma \lf x\inp{\wt y}.P
}
\\\\
\inferrule[{\Did{T$\pi$-Case}}]{
      \Gamma_1 \lf v:  \variant {l_i}{\ctype_i} \and 
      \Gamma_2, x_i:\ctype_i \lf P_i \quad \forall i\in I
}{
      \Gamma_1 \pp \Gamma_2 \lf  \picase v{x_i}{P_i}
}
  \end{array}$$  	
  \end{mdframed}
  \caption{Typing rules for the \p calculus (\Cref{sss:proc}). 
  Rules for values appear in the first line; the remaining rules are for processes.}\label{f:pityping}
\end{figure*}



%

The typing rules for values and processes are given in \Cref{f:pityping}.
All rules are as in~\cite{K06}, except for the new rules 
\Did{T\p LVal} and \Did{T\p Case}, 
which type a choice: the former types a variant value with a variant type; the latter types a case process using a variant value as its guard.
We discuss the remaining rules.
Rules \Did{T\p Var} and \Did{T\p Tup} are standard.
Rule~\Did{T\p Nil} states that the terminated process is typed under the empty context.
Rule~\Did{T\p Res} states that $\res xP$ is well-typed if the usage for $x$
 is reliable  (cf. \Cref{d:reli}).
 Rule~\Did{T$\pi$-Par} 
states that the parallel composition of processes $P$ 
and $Q$ (typable under $\Gamma_1$ and $\Gamma_2$, respectively)
is well-typed under the composed typing context $\Gamma_1 \pp \Gamma_2$ (\Cref{def:auxoper}(2)).
\rev{Note that, unlike Rule~$\Did{T-$\cut$}$ in \Cref{fig:type-system-cll}, $P$ and $Q$ need not to share any channels to be composed.}
Rules~\Did{T\p Out} and \Did{T\p In} type output and input processes with dyadic communication in a typing context where the `$\semi$' operator (\Cref{def:auxoper}(3)) is used to
increase the obligation level of the channels in continuation $P$.

%

\subsubsection{Properties}
\label{sss:techres}

Type soundness of the type system given in \Cref{f:pityping} implies that well-typed processes are deadlock-free (cf.  \Cref{def:lockpi}).
We now state these technical results from~\cite{K06} and discuss changes in their proofs, which 
require minimal modifications.

First,  typing is preserved by the structural relation $\preceq$ (cf. \Cref{d:kobastruct}):

\begin{lemma}
\label{l:kobastruct}
If $\Gamma \lf P$ and $P \preceq Q$ then $\Gamma \lf Q$.
\end{lemma}

The proof of type preservation in~\cite{K06} relies on other auxiliary results (such as substitution), which we do not recall here. To state the type preservation result, we need the following auxiliary definition:

\begin{definition}[Context Reduction]
\label{d:envredk}
We write $\Gamma\to \Gamma'$
when one of the following hold:
\begin{enumerate}
\item $\Gamma = \Gamma_1, x:\utype{U}{\widetilde{\ctype}}$
and
$\Gamma' = \Gamma_1, x:\utype{U'}{\widetilde{\ctype}}$
with $U \to U'$ (cf.~\Cref{def:usage_red}), for some $\Gamma_1$, $x$, $\widetilde{\ctype}$, $U$ and $U'$.

\item $\Gamma = \Gamma_1, x:\variant {l_i}{\ctype_i}$
and
$\Gamma' = \Gamma_1, x:\ctype_j$, 
with $j \in I$, for some $\Gamma_1$ and $x$.
\end{enumerate}
\end{definition}
Above, Item (1) is as in~\cite{K06}.
Item (2) is required for characterizing reductions of the case construct.
We now state type preservation. 

\begin{theorem}[Type Preservation]
\label{thm: preservation}
If $\Gamma \lf P$ and $P\to Q$, then
$\Gamma'\lf Q$ for some $\Gamma'$ such that 
$\Gamma' = \Gamma$ or
$\Gamma\to \Gamma'$.
\end{theorem}

The proof is by induction on the derivation $P\to Q$, following closely the proof in~\cite{K06}. An additional case is needed due to the reduction rule of the case construct (Rule~\Did{R\p Case}):
$$
\picase {\vv jv}{x_i}{P_i}	
		\to 
		P_j\substj{v}{x_j}
	\quad  (j\in I)
$$
As a result of this reduction, we have $\Gamma \to \Gamma'$ because of \Cref{d:envredk}(2).


The following important result extends Theorem 2 in \cite{K06} with 
the case construct: 
\begin{theorem}[Deadlock Freedom]\label{t:dfk}
If $\emptyset \lf P$ and
one of the following holds 
\begin{itemize}
\item $P\preceq \res{\wt{x}}(\ov x\out {\wt v}.Q_1 \pp Q_2)$
\item $P\preceq \res{\wt{x}}(x\inp {\wt z}.Q_1 \pp Q_2)$
\item $P \preceq \res{\wt{x}}(\picase {\vv jv}{x_i}{P_i} \pp Q)$
\end{itemize}
then $P\to R$, for some process $R$.
\end{theorem}
In~\cite{K06}, the proof of this theorem relies on a notion of \emph{normal form} in which input and output prefixes act as guards for if-expressions and let-expressions (not present in our syntax); our case construct can be easily accommodated in such normal forms. The proof in~\cite{K06} also uses: (i)~an extension of $\preceq$ for replicated processes, (ii)~mechanisms for uniquely identifying bound names, and (iii)~an ordering on processes. In our case, we do not need (i)  and can re-use (ii) and (iii) as they are. With these elements, the proof argument proceeds as in~\cite{K06}.
We finally have:

\begin{corollary}
If $\emptyset \lf P$ then $P$ is deadlock-free, in the sense of \Cref{def:lockpi}.
\end{corollary}

It is worth noticing how both {\Cref{t:progress} 
and \Cref{t:dfk} 
have similar formulations:   
both properties state that processes 
can always reduce if they 
are well-typed (under the empty typing context) and have an appropriate structure (cf., 
condition 
$live(P)$ in \Cref{t:progress} and 
condition 
$P\preceq \res{\wt{x}}(x\inp {\wt z}.Q \pp R )$ or $P\preceq \res{\wt{x}}(\ov x\out {\wt v}.Q \pp R )$
in \Cref{t:dfk}).}

\subsubsection{On Deadlock Freedom by Encoding}
\label{sss:enco}
\begin{figure}[t]
  \begin{mdframed}
      \vspace{-1em}
\begin{align*}
	\encf{\nil}						&\defeq \nil
\\
	 \encf{\res{xy}P}					&\defeq  \res {c} \encfp{P}{f,\{x,y \mapsto c\}}
\\
	 \encf{P\pp Q}					&\defeq \encf{P} \pp \encf{Q}
\\
	\encf{\ov x\out v.P}				&\defeq \res c \ov{\f x}\out {\f v, {c}}.\encx{P}{c} 
\\
	\encf{x\inp y.P}					&\defeq \f x\inp{y,c}.\encx{P}{c}
\\
 	 \encf{\selection x {l_j}.P}			&\defeq   \res c \ov{\f x}\out {\vv{j}c}.\encx{P}{c}
\\
	 \encf{\branching xlP}				&\defeq  \f x \inp y.\ \picase {y}{c}{\encx {P_i}c}
\end{align*}
  \end{mdframed}
\caption{Encoding session $\pi$-calculus processes into the dyadic $\pi$-calculus, under a renaming function $f$ on names/channels.}
\label{d:encdgs}
\end{figure}

We use encodings to relate classes of (typed) processes induced by the type systems given so far.
To translate a session typed process into a usage typed process, we 
follow the encoding suggested in~\cite{K07} and developed in~\cite{DGS12}, 
which 
mimics the sequential structure of a session by sending its continuation as a payload over a channel.
This continuation-passing style encoding of processes is 
denoted $\encf{\cdot}$, where $f$ is a renaming function from channels to fresh names; see~\Cref{d:encdgs}.
We write $f_x$ to stand for $f(x)$ and $f, \{x \mapsto c\}$ to denote that the entry for $x$ in the renaming $f$ is updated to $c$. We write $f, \{x,y \mapsto c\}$ to mean that both $x$ and $y$ are updated to $c$.

We also need to formally relate  session types to  usage types. 
To this end, we define  $\encob{\cdot}{\su}$ in \Cref{f:enctypes}.
Two points are noteworthy: (i)~it suffices to generate a usage type with obligations and capabilities equal to 0, and (ii)~only sequential usages are needed to encode session types.
We now extend this encoding from types to typing contexts:

\begin{figure}[t]
  \begin{mdframed}
      \vspace{-1em}
\centering
\begin{align*}
\encob{\nilT}{\su} & =  \utype{\zusage}{\temp} \\
\encob{\wn T.S}{\su} & =  \utype{\inpuse{0}{0}}{\encob{T}{\su}, \encob{S}{\su}}\\
\encob{\oc T.S}{\su} & = \utype{\outuse{0}{0}}{\encob{T}{\su},\encob{\dual S}{\su}}\\
\encob{\branch lS}{\su} & =  \utype{\inpuse{0}{0}}{\variant {l_i}{\encob{S_i}{\su}}}\\
\encob{\select lS}{\su} & =   \utype{\outuse{0}{0}}{\variant {l_i}{\encob{\dual{S_i}}{\su}}}
\end{align*}
  \end{mdframed}
\caption{Encoding  session types into usage types.
\label{f:enctypes}}
\end{figure}

\begin{definition}
\label{def:enc_env_su}
Given $\encob{\cdot}{\su}$ as in \Cref{f:enctypes}, and 
with a slight abuse of notation, 
  we write $\encf{\cdot}$ to denote 
  the encoding
  of session type contexts $\Gamma$ into
usage typing contexts  that is
inductively defined as follows:
\[
\encf{\emp} \defeq \emp
\qquad 
\encf{\Gamma,x:T}\defeq \encf{\Gamma} , \f x:\encob{T}{\su}
\]
\end{definition}
%

The next results relate deadlock freedom, typing and the encodings in~\cite{DGS12}, thus formalising the indirect approach to deadlock freedom.
\begin{proposition}
Let $P$ be a deadlock-free session process.
Then $\encf P$ is a deadlock-free \p process.
\end{proposition}
\begin{proof}
The encoding of terms given in \Cref{d:encdgs} preserves the nature of the prefixes in $P$ (outputs are directly translated as outputs, and similarly for inputs) and is defined homomorphically with respect to parallel composition. 
Then, it is easy to see that if $P$ reduces then $\encf P$ can immediately match that reduction; thanks to the definitions of deadlock freedom in each language (\Cref{def:lock} and \Cref{def:lockpi}, respectively), this suffices to conclude the thesis.
\end{proof}

The following result states deadlock freedom by encodability, following~\cite{CDM14}.
\begin{corollary}\label{c:df}
Let $\s P$ be a session process.
If $\lf \encf P$ is deadlock-free then $P$ is deadlock-free.
\end{corollary}


\section{Separating Classes of Deadlock-Free Typed Processes}\label{s:hier}
Up to here, we have summarized three existing type systems for the $\pi$-calculus: 
\begin{itemize}
    \item Session types, with judgment $\Gamma \s P$ (\Cref{ss:typesess});
    \item Session types based on linear logic, with judgment $P \cp \D$ (\Cref{ss:lf});
    \item Usage types, with judgment $\Gamma \lf P$ (\Cref{sss:usage}).
\end{itemize}
Here we establish formal relationships between these type systems.
As already mentioned, our approach consists in defining \lcp and \fullKoba, the class of deadlock-free session-typed processes induced by the type systems in \Cref{ss:lf} and \Cref{sss:usage}, respectively. 
Our main result is that $\lcp \subset \lkoba$ (Corollary \ref{c:subset}), a strict inclusion that \emph{separates} these two classes of processes.
To obtain this separation result, we define \muKoba, a strict sub-class of \lkoba which we show to coincide with \lcp. 

\subsection{Classes of Deadlock-Free Processes}
\subsubsection{The Classes \lcp and \fullKoba}
\label{ss:classes}
We start by defining classes \lcp and \fullKoba, for which 
we require some auxiliary definitions.
The following encoding addresses 
minor syntactic differences between session typed processes (cf. \Cref{s:sessions}) 
and the processes typable in the linear logic interpretation of session types (cf.~\Cref{ss:lf}).
Such differences concern free output actions and the double restriction operator:

\begin{figure}[t]
  \begin{mdframed}
      \vspace{-1em}
  	\begin{align*}
\chr{\overline x\out y.P} &\defeq \bout{x}{z}.(\linkr{z}{y} \pp \chr{P}) 
\\
\chr{\res{xy}P} &\defeq  \res{w}\chr{P}\substj{w}{x}\substj{w}{y} \quad w \not\in \fn{P}
\end{align*}
  \end{mdframed}

\caption{Encoding session $\pi$-calculus processes into the linear logic processes (cf. \Cref{ss:lf}). It is defined as an homomorphism for the other process constructs.\label{f:sesstocp}}
\end{figure}

\begin{definition}\label{d:trans}
Let $P$ be a session process. The auxiliary encoding $\chr{\cdot}$ 
from the session processes (cf. \Cref{s:sessions}) into the linear logic processes (cf. \Cref{ss:lf})
is defined in \Cref{f:sesstocp}.
\end{definition}

We also need an encoding of session types into linear logic propositions.
The encoding, also denoted $\encob{\cdot}{\ct}$ and given in \Cref{f:enctypesct}, simply follows the essence of the linear logic interpretation.
We extend it to typing contexts as follows:

\begin{definition}
\label{def:enc_env_sc}
The encoding $\encob{\cdot}{\ct}$ 
of session type contexts $\Gamma$
into linear logic typing contexts is defined as:
\[
\encob{\emp}{\ct} \defeq \emp
\qquad 
\encob{\Gamma,x:T}{\ct}	\defeq \encob{\Gamma}{\ct}, x:\encob{T}{\ct}
\]
\end{definition}

It is not difficult to see that the encoding in \Cref{d:trans} is operationally correct.
That is, $\encob{\cdot}{\ct}$ preserves and reflects reductions of session processes. It also preserves the encoding of types in \Cref{f:enctypesct} and \Cref{def:enc_env_sc}.

The next definition allows us to abstract away from differences in the way type systems handle  process $\nil$: in the usages type system, process $\nil$ can only be typed under the empty typing context, whereas in the session type systems the typing context can be non-empty (subject to conditions).

\begin{definition}[Core Context]
Given $\Gamma \s P$, we write $\core{\Gamma}$ to denote the \emph{core context} of $\Gamma$ with respect to $P$. Context  $\core{\Gamma}$ is defined so that ${\Gamma} = \core{\Gamma}, \Gamma'$ holds for some $\Gamma'$, satisfying the following conditions: (i)~$\eend{\Gamma'}$; (ii)~$x:\nilT \in \Gamma'$ implies $x \not\in \fv{P}$;
and 
(ii)~$x:T \in \core{\Gamma}$ implies $x \in \fv{P}$.
\end{definition}
 Notice that, by  \Cref{thm:strength} (strengthening), $\Gamma \s P$ implies $\core{\Gamma} \s P$.

\begin{figure}
  \begin{mdframed}
      \vspace{-1em}
  \begin{align*}
\encob{\nilT}{\ct} & =  \bullet \\
\encob{\wn T.S}{\ct} & =   \encob{{T}}{\ct} \parl \encob{S}{\ct}\\
\encob{\oc T.S}{\ct} & =  \encob{\dual T}{\ct} \otimes \encob{S}{\ct}\\
\encob{\branch lS}{\ct} & =  \&\big\{l_i:\encob{S_i}\ct\big\}_{i\in I} \\
\encob{\select lS}{\ct} & =   \oplus\big\{l_i:\encob{S_i}\ct\big\}_{i\in I} 
\end{align*}
  \end{mdframed}
\caption{Encoding session types into linear logic types.
\label{f:enctypesct}}
\end{figure}

\smallskip

Recall that we write \procs to denote the class of session $\pi$-calculus processes generated by the grammar in
 \Cref{fig:sessionpi}.
Formally,  \lcp and \lkoba are classes of processes in \procs, defined below.

\begin{definition}[\lcp and \fullKoba]\label{d:lang}
The classes \lcp and \fullKoba are defined as follows:
\begin{align*}
\lcp  &\defeq    \Big\{P \in \procs \suchthat \exists \Gamma.\ (\Gamma \s P \,\land\, \chr{P} \cp \encob{\Gamma}{\ct}) \Big\} 
\\
\lkoba &\defeq  \Big\{  P \in \procs \suchthat \exists \Gamma,f.\ (\Gamma \s P \,\land\, \encf{\core{\Gamma}} \lf \encf P) \Big\} 
\end{align*}
\end{definition}

In words, \lcp contains those well-typed session $\pi$-calculus processes whose corresponding translation as a process in \Cref{ss:lf}  (using $\chr{\cdot}$ as in \Cref{d:trans}) is also typable in the linear logic interpretation,  with propositions obtained using the encoding on types $\encob{\cdot}{\ct}$ (\Cref{f:enctypesct}).
Similarly,   \lkoba contains those well-typed session $\pi$-calculus processes   whose corresponding translation as a dyadic $\pi$-calculus process (using $\encf{\cdot}$ as in \Cref{d:encdgs})
is also typable under Kobayashi's type system, with usage types derived using the encoding on types $\encob{\cdot}{\su}$
(\Cref{f:enctypes}).

Notice that \lcp and \lkoba contain both composable and uncomposable processes (cf. \Cref{d:openclose}).
As informally discussed in the Introduction,
processes in \lcp and \lkoba satisfy the \emph{progress} property, defined in~\cite{CD10} and further studied in \cite{CDM14}.
As a consequence:
\begin{itemize}
\item \emph{Uncomposable processes}  in \lcp (typable with $\Delta=\cdot$) are deadlock-free
by following \Cref{t:progress}. Similarly, 
 \emph{uncomposable processes}
in \lkoba (typable with $\Gamma=\emptyset$) are deadlock-free by the indirect approach formalised by \Cref{t:dfk} and \Cref{def:lockpi} and \Cref{c:df}.
\item
 \emph{Composable processes}  in \lcp 
 (typable with $\Delta\neq\cdot$)
 and \lkoba (typable with $\Gamma\neq\emptyset$) may be stuck, because 
  they lack communicating counterparts as described by their (non-empty) typing context. These missing counterparts will be formalized as a \emph{catalyzer} \cite{CD10} that allows a  process to further reduce, thereby ``unstucking it''.
 \end{itemize}
Although we are interested in the (sub)class of processes that satisfy deadlock freedom, we have defined \lcp and \lkoba more generally as processes in \procs satisfying progress; this simplifies the definition and presentation of our technical contributions.


\subsubsection{The Class \muKoba}
\label{ss:mukoba}

\begin{figure*}[t]
    \begin{mdframed}[style=newtight]
        \vspace{-1em}
\centering
$$\begin{array}{c}
\inferrule[{\Did{T$\pi$-Var}}]{
}{
 x:\ctype \mlf x:\ctype
}
\quad
\inferrule[{\Did{T$\pi$-Tup}}]{
\Gamma_1 \mlf v_1:\ctype_1 \and \Gamma_2 \mlf v_2:\ctype_2  
\and \wt{v} = v_1,  v_2~~\wt{\ctype} = \ctype_1, \ctype_2  
}{
\Gamma_1 \pp  \Gamma_2 \mlf \wt{v}: \wt{\ctype}
}
\quad
\inferrule[{\Did{T$\pi$-LVal}}]{
      \Gamma \mlf v: \ctype_j \quad \exists j\in I
}{
      \Gamma \mlf   {l_j}\_v: \variant {l_i}{\ctype_i}
}
\\\\
\inferrule[{\Did{T$\pi$-Nil}}]{
}{
 \emp  \mlf {\nil}
}
\qquad
\inferrule[{\Did{T$\pi$-IndPar}}]{
	\Gamma_1 \mlf P \and \Gamma_2 \mlf Q
    \and  
    \dom(\Gamma_1) \cap \dom(\Gamma_2) =  \emptyset
}{ 
      \Gamma_1, \Gamma_2 \mlf  P\pp Q
}
\\\\
	\inferrule[{\Did{T$\pi$-Par+Res}}]{
      \Gamma_1, x: \utype{U_1}{\ctype} \mlfpar{\prec \cup \prec_1} P_1
      \and 
      \Gamma_2, x: \utype{U_2}{\ctype} \mlfpar{\prec \cup \prec_2} P_2
      \and 
          \dom(\Gamma_1) \cap \dom(\Gamma_2) =  \emptyset
                \\ 
      i \in \{1,2\}
      \and 
      \prec_i =   \{(x,y) \,|\, y \in \fv{P_i}\setminus \{x\} \}
      \and 
      \rel{U_1 \pp U_2}
}{
      \Gamma_1 \pp \Gamma_2 \mlf \res{x} (P_1 \pp P_2)
}
\\\\
%
%
\inferrule[{\Did{T$\pi$-BOut}}]{
     \Gamma_1, y: \utype{U_1}{\ctype} 
     \mlfpar{\prec'} 
     y: \utype{U_1}{\ctype}, y': \utype{U}{\ctype'}       
     \and 
     \Gamma_2, y: \utype{U_2}{\ctype}  \mlfpar{\prec'}  P 
     \\   
     \rel{U_1 \pp U_2}
     \and 
     \prec' = \prec \cup \{(y,z) \,|\, z \in \fv{P}\setminus \{y\} \}
     \and
     \wt{y} = y, y'~~\wt{\ctype} = \ctype, \ctype'
}{
       x: \utype{\oc^{0}_{\ca}}{\wt{\ctype}} \semi (\Gamma_1\pp \Gamma_2) 
       \mlf 
       \res{y}{\ov x}\out{\wt{y}}.P
}

\\\\\

\inferrule[{\Did{T$\pi$-In}}]{
      \Gamma, \wt{z}:\wt{\ctype} \mlf P
}{
      x:\utype{\wn^{0}_{\ca}}{\wt{\ctype}} \semi \Gamma \mlf x\inp{\wt{z}}.P
}
\qquad
\inferrule[{\Did{T$\pi$-Case}}]{
      \Gamma_1 \mlf v:  \variant {l_i}{\ctype_i} \and 
      \Gamma_2, x_i:\ctype_i \mlf P_i \quad \forall i\in I
}{
      \Gamma_1 \pp \Gamma_2 \mlf  \picase v{x_i}{P_i}
}
  \end{array}$$
    \end{mdframed}

  \caption{Typing rules for the \p calculus (\Cref{sss:proc}), which specialize those in \Cref{f:pityping} to define the class \muKoba. 
  Typing rules for values appear in the first line; the remaining typing rules are for processes.}\label{f:mupityping}
\end{figure*}

We now define \muKoba, a strict class of \fullKoba obtained by internalizing key features of linear logic interpretation in~\Cref{ss:lf} into the rules of \Cref{f:pityping}.
Considering the type system is summarized in \Cref{ss:dgs}, \muKoba will arise from a type system modified as follows:
\begin{description}
	
	\item[Types] Because session types are strictly sequential, we restrict \Cref{d:typesusages} to channel types of the form $\utype{U}{\ctype_1, \ctype_2}$ where $U$ is a \emph{sequential usage} (cf.~\Cref{d:usages}). 
	All other notions (obligations and capabilities, the semantics of usages, the notion of reliability) are kept unchanged.
	
	\item[Typing Rules] The type system considers judgments of the form $\Gamma \mlf P$. 
	The typing rules are given in \Cref{f:mupityping}; they are based on those in \Cref{f:pityping} with the following modifications. 
	First, we  consider Rule~$\Did{T$\pi$-BOut}$, which  accounts for bound output as used in the type discipline in \Cref{ss:lf}. 
	Second, to account for the ``composition plus hiding'' principle, we merge Rules $\Did{T$\pi$-Res}$ and $\Did{T$\pi$-Par}$ into Rule~\Did{T$\pi$-Par+Res}. In this modified rule, the parallel usage appears exclusively for the purpose of ensuring reliability. 
	Finally, we include Rule $\Did{T$\pi$-IndPar}$ to account for independent parallel composition as expressed by the mix principle in \Cref{ss:lf}.
	\end{description}

We use this modified type system to define $\muKoba$, following \Cref{d:lang}:	
	\begin{definition}[\muKoba]\label{d:langmu}
	The class of processes \muKoba is defined as follows:
	$$
	\muKoba \defeq  \Big\{ P \in \procs \suchthat \exists \Gamma,f.\ (\Gamma \s P \,\land\, \encf{\core{\Gamma}} \mlf \encf P) \Big\} 
$$
\end{definition}
Considerations about (un)composable processes in \muKoba hold exactly as described for \fullKoba.

\subsection{Main Results} 
We start by separating \fullKoba and \muKoba. First, we have the following:
 
 \begin{lemma}
 \label{l:muKoba}
 Let $P$ be a dyadic $\pi$-calculus process as in \Cref{sss:proc}.
If $\Gamma \mlf P$ then $\Gamma \lf P$.
 \end{lemma}
 
 \begin{proof}
By induction on the type derivation for $P$, with a case analysis on the last applied typing rule. 
We rely on the analogue property for values (if $\Gamma \mlf v$ then $\Gamma \lf v$), which is immediate.
We discuss only the cases featuring key differences between $\mlf$  (\Cref{f:mupityping}) and  $\lf$ (\Cref{f:pityping}) arise:
\begin{itemize}
    \item If the last applied rule is \Did{T$\pi$-BOut}, then $P = \res{y}{\ov x}\out{y,y'}.P'$
    and
    $$
    \inferrule
    {
     \Gamma_1, y: \utype{U_1}{\ctype} 
     \mlfpar{\prec'} 
     y: \utype{U_1}{\ctype}, y': \utype{U}{\ctype'}       
     \and 
     \Gamma_2, y: \utype{U_2}{\ctype}  \mlfpar{\prec'}  P 
     \\   
     \rel{U_1 \pp U_2}
     \and 
     \prec' = \prec \cup \{(y,z) \,|\, z \in \fv{P}\setminus \{y\} \}
}{
       \Gamma 
       \mlf 
       \res{y}{\ov x}\out{y,y'}.P
}
$$
with $\Gamma = x: \utype{\oc^{0}_{\ca}}{\ctype, \ctype'} \semi (\Gamma_1\pp \Gamma_2)$.
We have 
$\Gamma_1, y: \utype{U_1}{\ctype} 
     \lfpar{\prec'}
     y: \utype{U_1}{\ctype}, y': \utype{U}{\ctype'}$.
     Also, by IH:
$\Gamma_2, y: \utype{U_2}{\ctype}  \lfpar{\prec'}  P$.

We show that $P$ can be typed using Rules~$\Did{T$\pi$-Out}$ and $\Did{T$\pi$-Res}$ in sequence. 
First we have:
$$
\inferrule{
\Gamma_1, y: \utype{U_1}{\ctype} 
     \lfpar{\prec'}
     y: \utype{U_1}{\ctype}, y': \utype{U}{\ctype'}
     \and 
    \Gamma_2, y: \utype{U_2}{\ctype}  \lfpar{\prec'}  P 
}{
       x: \utype{\oc^{0}_{\ca}}{\ctype, \ctype'} \semi (\Gamma_1 \pp \Gamma_2 \pp y: \utype{U_1 \pp U_2}{\ctype}) 
       \lf 
       {\ov x}\out{y,y'}.P
}
$$
Now, because by assumption we have $\rel{U_1 \pp U_2}$, we can use Rule~$\Did{T$\pi$-Res}$  to derive 
$$
\inferrule{
       x: \utype{\oc^{0}_{\ca}}{\ctype, \ctype'} \semi (\Gamma_1 \pp \Gamma_2 \pp y: \utype{U_1 \pp U_2}{\ctype}) 
       \lf 
       {\ov x}\out{y,y'}.P
       \and 
       \rel{U_1 \pp U_2}
}
{
       x: \utype{\oc^{0}_{\ca}}{\ctype, \ctype'} \semi (\Gamma_1 \pp \Gamma_2) 
       \mlf 
       \res{y}{\ov x}\out{y,y'}.P
}
$$

    \item If the last applied rule is \Did{T$\pi$-Par+Res}, then $P = \res{x} (P_1 \pp P_2)$ and
   $$\inferrule{
      \Gamma_1, x: \utype{U_1}{\ctype} \mlfpar{\prec \cup \prec_1} P_1
      \and 
      \Gamma_2, x: \utype{U_2}{\ctype} \mlfpar{\prec \cup \prec_2} P_2
      \and 
          \dom(\Gamma_1) \cap \dom(\Gamma_2) =  \emptyset
      \\ 
      i \in \{1,2\}
      \and 
      \prec_i =   \{(x,y) \,|\, y \in \fv{P_i}\setminus \{x\} \}
      \and 
      \rel{U_1 \pp U_2}
}{
     \Gamma \mlf \res{x} (P_1 \pp P_2)
}$$
    where $\Gamma = \Gamma_1 \pp \Gamma_2$. 
    By IH, we have 
    $\Gamma_i, x: \utype{U_i}{\ctype} \lfpar{\prec \cup  \prec_i} P_i$, with $i \in \{1,2\}$.
    
    We show that $P$ can be typed using Rules~$\Did{T$\pi$-Par}$ and $\Did{T$\pi$-Res}$ in sequence.
    First, we have:
    $$
    \inferrule{
    \Gamma_1, x: \utype{U_1}{\ctype} \lfpar{\prec_ \cup \prec_1 \cup \prec_2} P_1
    \and
    \Gamma_2, x: \utype{U_2}{\ctype} \lfpar{\prec_ \cup \prec_1 \cup \prec_2} P_2
    }{
    \Gamma_1 \pp \Gamma_2 \pp x: \utype{U_1 \pp U_2}{\ctype} \lfpar{\prec \cup \prec_1 \cup \prec_2} P_1 \pp P_2
    }
    $$
    where we have weakened, in both cases, the partial orders for $P_1$ and $P_2$. Now, because by assumption we have $\rel{U_1 \pp U_2}$, we can use Rule~$\Did{T$\pi$-Res}$  to derive 
$$
\inferrule{
         \Gamma_1 \pp \Gamma_2 \pp  x: \utype{U_1 \pp U_2}{\ctype} \lfpar{\prec \cup \prec_1 \cup \prec_2} P_1 \pp P_2
         \and 
         \rel{U_1 \pp U_2}
}
{
\Gamma_1 \pp \Gamma_2 \lfpar{\prec} \res{x}{(P_1 \pp P_2)}
}
$$
    
    \item Finally, if the last applied rule is \Did{T$\pi$-IndPar}, then process $P = P_1 \pp P_2$, which can be immediately typed using Rule~$\Did{T$\pi$-Par}$.  
\end{itemize}
\end{proof}
 
  As a result of  \Cref{l:muKoba}, properties derived from typing in \Cref{ss:dgs} (type preservation and deadlock freedom) hold in the modified type system $\mlf$ and apply to processes in \muKoba.
  Thus, $\muKoba \subseteq \fullKoba$.

The modifications defined in \Cref{ss:mukoba} make $\mlf$  a strict subsystem of the system $\lf$ in \Cref{ss:dgs}.
Indeed, because the converse of \Cref{l:muKoba} does not hold, there are processes in $\fullKoba$ but not in $\muKoba$:
\begin{lemma}
\label{lem:K1inK2}
 $\muKoba \subset \fullKoba$.
 \end{lemma}
 \begin{proof}
  Because $\muKoba$ is obtained by restricting the typing rules of $\fullKoba$, it is immediate that $\muKoba \subseteq \fullKoba$. In the following, we show that this inclusion is strict, by showing that $\fullKoba$ contains (deadlock-free) session processes not in $\muKoba$.
 A representative example is: 
\begin{align*}
P & \defeq   \res{a_1b_1}\res{a_2b_2}(a_1\inp x.\;\ov{a_2}\out{x}
        \pp \ov{b_1}\out{\unit}.\;b_2\inp{z})
\\
  \encf{P} & =   \res{c}\res{d}(c\inp{x,w}.\;\res{d'}\ov{d}\out{x,d'}
        \pp \res{c'}\ov{c}\out{\unit,c'}.\;d\inp{z,u})
\end{align*}
%
\noindent
This process is not in $\muKoba$ because $\encf{P}$
involves the composition of two parallel processes which share two sessions. 
Hence, $\encf{P}$ is typable in $\lf$ 
(using Rules~$\Did{T$\pi$-Par}$ and $\Did{T$\pi$-Res}$ in \Cref{f:pityping})
but not in $\mlf$, because of the forms of composition admitted by Rule~$\Did{T$\pi$-Par+Res}$ in \Cref{f:mupityping}.
 \end{proof}
 
%

We now show that $\lcp$ and $\muKoba$ coincide.
We need an important property, given by \Cref{p:relia}, which  connects our encodings of (dual) session types into usage types
with  reliability (\Cref{d:reli}), a central notion to the type systems for deadlock freedom in \Cref{f:pityping} and \Cref{f:mupityping}.
First we have the following result, which connects our encodings of types and the notion of duality.

\begin{lemma}
\label{lem:dualenc}
Let $T,S$ be  session types. Then, the following hold:
(i)~ $\dual T =  S$  if and only if ${\encob{\dual T}{\ct}} = \encob{S}{\ct}$;
(ii)~$\dual T =  S$  if and only if ${\encob{\dual T}{\su}} = \encob{S}{\su}$.
%
%
\end{lemma}
\begin{proof}
By induction on the duality relation of session types.
\end{proof}
%

Given a channel type $\ctype = \utype{U}{\widetilde{\ctype}}$, we write $\usa{\ctype}$ to denote the usage $U$.

\begin{lemma}
\label{p:relia}
Let $T$ be a session type.
Then $\rel{\usa{\encob{T}{\su}} \pp \usa{\encob{\dual{T}}{\su}}}$ holds.
\end{lemma}
\begin{proof}
The proof proceeds by induction on the structure of   $T$,  using 
\Cref{lem:dualenc}.
\begin{enumerate}[label=$\bullet$]
\item
$T = \nilT$. 
By duality, $\dual{T} = \nilT$.
Then
$\encob{T}{\su} = \encob{\dual{T}}{\su} = \utype{\zusage}{\temp}$.
Thus, $\usa{\encob{T}{\su}} \pp \usa{\encob{\dual{T}}{\su}} = \zusage \pp\zusage$. 
Notice that 
$\zusage \pp \zusage \not \to$ and that
$\con{\zusage \pp \zusage}$ holds trivially.
Therefore, by \Cref{d:reli}, $\rel{\zusage \pp\zusage}$.

\item 
$T = \oc T_1.{T_2}$ for some $T_1, T_2$.
By definition of duality, $\dual{T} = \wn T_1.\dual{T_2}$.
Then,
\begin{align*}
\encob{T}{\su}  & =
\utype{\outuse{0}{0}}{\encob{T_1}{\su}, \encob{\dual {T_2}}{\su}}  
\\
\encob{\dual{T}}{\su} & =
\utype{\inpuse{0}{0}}{\encob{T_1}{\su}, \encob{\dual {T_2}}{\su}}
\end{align*}
Thus, 
$\usa{\encob{T}{\su}} \pp \usa{\encob{\dual{T}}{\su}} = {\outuse{0}{0} \pp \inpuse{0}{0} }$.
Notice that $\outuse{0}{0} \pp \inpuse{0}{0}  \to \zusage \pp \zusage \not \to$.
We examine $\con{\cdot}$ for both usages.
First, $\con{\outuse{0}{0} \pp \inpuse{0}{0}}$ holds because $0 \leq 0$; 
then, $\con{\zusage \pp \zusage}$ trivially holds.
Therefore, by \Cref{d:reli}, $\rel{\outuse{0}{0} \pp \inpuse{0}{0}}$ holds.

\item 
$T = \wn T_1.{T_2}$, for some $T_1, T_2$. 
By definition of duality, $\dual{T} = \oc T_1.\dual{T_2}$.
Then,
\begin{align*}
\encob{T}{\su}   & =
\utype{\inpuse{0}{0}}{\encob{T_1}{\su}, \encob{{T_2}}{\su}}
\\
 \encob{\dual{T}}{\su} & =
\utype{\outuse{0}{0}}{\encob{T_1}{\su}, \encob{{T_2}}{\su}} 
\end{align*}
Thus, 
$\usa{\encob{T}{\su}} \pp \usa{\encob{\dual{T}}{\su}} =
{\inpuse{0}{0} \pp \outuse{0}{0}}$
and the thesis holds as in the previous case.

\item
$T = \branch lS$, for some $S_i$.
By definition of duality, $\dual T = \select l{\dual S}$.
Then,
\begin{align*}
\encob{T}{\su}    & = 
 \utype{\inpuse{0}{0}}{\variant {l_i}{\encob{S_i}{\su}}}
 \\
 \encob{\dual{T}}{\su}  & = 
 \utype{\outuse{0}{0}}{\variant{l_i}{\encob{S_i}{\su}}}
\end{align*}
Thus, 
$\usa{\encob{T}{\su}} \pp \usa{\encob{\dual{T}}{\su}} =
{\inpuse{0}{0} \pp \outuse{0}{0}}$
and the thesis holds just as in case $T = \wn T_1.{T_2}$.

\item 
$T = \select lS$, for some $S_i$ and $i\in I$: similar to the previous cases.
\end{enumerate}
\end{proof}

We then have the following result:
\begin{theorem}\label{t:cppkoba}
$\lcp = \muKoba$.
\end{theorem}

\begin{proof}[Proof (Sketch)]
We prove two lemmas:
(i)~If $P \in \lcp$ then $P \in \muKoba$
and
(ii)~If $P \in \muKoba$ then $P \in \lcp$.
Both lemmas are proven by structural induction on $P$; 
see \Cref{appt:cppkoba} (Page~\pageref{appt:cppkoba}) for details.
\end{proof}

%
%
%
%
%
%
%
%
%

Therefore, we have the following corollary, which attests that 
the class of deadlock-free session processes  induced by 
linear logic interpretations of session types (cf. \Cref{ss:lf}) is strictly included in the class
induced by the indirect approach of~\cite{DGS12} (cf. \Cref{ss:dgs}).

\begin{corollary}\label{c:subset}
$\lcp \subset \lkoba$.
\end{corollary}
\noindent 
The fact that (deadlock-free) processes such as $P$ (cf.~\Cref{lem:K1inK2}) are not included in \lcp is informally discussed by Caires et al.~in~\cite[\S6]{DBLP:journals/mscs/CairesPT16}.
\rev{The discussion in~\cite{DBLP:journals/mscs/CairesPT16} highlights the principle of ``composition plus hiding'' (enforced by Rule~$\Did{T-$\cut$}$, cf. \Cref{fig:type-system-cll}) as a feature that distinguishes logically motivated session type systems from other type systems (in particular, the one in~\cite{GH05}), which can type $P$ but also deadlocked variants of it (cf. \Cref{ex:basic}). However, Caires et al.~give no formal comparisons with other classes of deadlock-free processes.}

%

\section{Translating \lkoba into \lcp}\label{s:enco}
Having established and characterized the differences between deadlock-free session $\pi$-calculus processes in \lcp and \lkoba, in this section we explore how fundamental these differences really are. 
To this end, we define a type-preserving translation  from processes in \lkoba into processes in \lcp. 

One leading motivation for looking into a translation is that the separation result established by \Cref{c:subset} can be seen as being \emph{subtle}, in the following sense.
Consider a process $P$ that belongs to $\fullKoba$ but not to $\muKoba$ because one of its subprocesses does not conform to the (restrictive) form of typed parallel composition enforced by Rule~\Did{T$\pi$-Par+Res} in \Cref{f:mupityping}. The fact that $P \in \fullKoba \setminus \muKoba$ implies that $P$ features forms of session cooperation that are intrinsically sequential and admitted by Rules~\Did{T$\pi$-Res} and~\Did{T$\pi$-Par} in \Cref{f:pityping}; this means that subprocesses of $P$ must become more independent (concurrent) to be admitted in \muKoba.
We illustrate this intuition with an example:
\begin{example}[Idea of the Translation]\label{ex:simple}
Recall
process
$P$ in \Cref{lem:K1inK2}:
$$
P \defeq    \res{a_1b_1}\res{a_2b_2}(a_1\inp x.\;\ov{a_2}\out{x}
        \pp \ov{b_1}\out{\unit}.\;b_2\inp{z})
$$
We have that $P \in \fullKoba$ but  $P \not\in \muKoba$:
each sub-process features two independent sessions occurring in sequence.

Consider now the following variant of $P$, 
in which the 
left subprocess has been kept unchanged, but the
right subprocess has been modified to increase concurrency: 
$$
P' \defeq  \res{a_2b_2}(\res{a_1b_1}(a_1\inp x.\;\ov{a_2}\out{x}
        \pp \ov{b_1}\out{\unit}.\nil) \pp b_2\inp{z}.\nil)
$$
Indeed, by replacing 
$\ov{b_1}\out{\unit}.\;b_2\inp{z}$
with 
$\ov{b_1}\out{\unit}.\nil \pp b_2\inp{z}.\nil$, 
we have that $P' \in \muKoba$.
\end{example}

Here we propose a \emph{translation} that converts any typable session process into a process in $\lcp$ (i.e. $\muKoba$).
The translation, given in \Cref{ss:rw}, follows the idea of \Cref{ex:simple}: given a parallel process as input, return as output a process in 
which 
one of the components is kept unchanged, but the other is translated by using  representatives of the sessions implemented in it, composed in parallel.
Such parallel representatives are formally defined as \emph{characteristic processes} and \emph{catalyzers}, which we introduce next.

\subsection{Characteristic Processes and Catalyzers}
We need some preliminary notions.
A \emph{characteristic process} of a session type $T$  represents the smallest process in $\lcp$ inhabiting it.

\begin{definition}[Characteristic Processes of a Session Type]\label{def:charprocess}
Given a name $x$, 
the set of \emph{characteristic processes} of session type $T$, denoted \chrp{T}{x}, is inductively defined as follows:
\begin{align*}
\chrp{\nilT}{x} 		&\defeq 	{\{\nil\}} \\
\chrp{\wn T'.S}{x} 	&\defeq    {\big\{x(y).(P \para Q) \suchthat P \in \chrp{T'}{y} \land Q \in \chrp{S}{x}\big\}}\\
\chrp{\oc T'.S}{x} 		&\defeq   \big\{\bout{x}{y}.(P \para Q) \suchthat P \in \chrp{\dual{T'}}{y} \land Q \in \chrp{S}{x}\big\}\\
\chrp{\branch lS}{x}	&\defeq    \big\{\branching xlP \suchthat \forall i\in I.\ P_i\in \chrp{S_i}x\big\}\\
\chrp{\select lS}x 		&\defeq    \bigcup_{i\in I}\big\{ \selection x{l _i}.{P_i} \suchthat P_i\in \chrp{S_i}x\big\}
\end{align*}
\end{definition}

\noindent
The previous definition extends to typing contexts by 
composing in parallel independent characteristic processes, one for each of the session types declared in the context.
This reflects that sessions in a context declare independent structures of communication.

\begin{definition}[Characteristic Processes of a Session Typing Context]\label{def:charenv}
Given a context $\Gamma = w_1{:}T_1, \ldots, w_n{:}T_n$, 
we shall write $\chrpk{\Gamma}{}$
to stand for the set
$
\left\{
(P_1 \para \cdots \para P_n)
\ \suchthat
P_i \in \chrp{T_i}{w_i}
\right\}
$.
\end{definition}

Characteristic processes are well-typed in the system of \Cref{ss:lf} (cf. \Cref{fig:type-system-cll}):

\begin{lemma}
\label{lem: wt_characteristic}
\label{cor:context_char}
Let $T$ be a session type and $\Gamma$ be a session context.
\begin{enumerate}
\item\label{item1: wt_characteristic} For all $P \in \chrp{T}{x}$, we have  $P \cp x:\encob{T}{\ct}$.
\item\label{item2: wt_characteristic} For all $P\in \chrpk{\Gamma}{}$, we have $P\cp \encob{\Gamma}{\ct}$.
\end{enumerate}
\end{lemma}
\begin{proof}
The proof of \Cref{item1: wt_characteristic} is
by induction on the structure of $T$.
The proof of \Cref{item2: wt_characteristic} is 
by induction on the size of $\Gamma$, 
using \Cref{item1: wt_characteristic}.
See \Cref{applem: wt_characteristic} (Page~\pageref{applem: wt_characteristic}) for details.
\end{proof}

{
Let us use `$[\cdot]$' to denote a \emph{hole} and $C[\cdot], C'[\cdot], \ldots$ to denote process contexts (i.e., a process with a hole).
Building upon characteristic processes, a \emph{catalyzer}
for a typing context is a process context
that implements the  behaviors it declares.}

\begin{definition}[Catalyzers of a Session Typing Context]\label{def:catalyser}
Given a session typing context $\Gamma$, we define its {set of} associated \emph{catalyzers}, noted $\mathcal{C}_\Gamma$, inductively as 
follows:
$$
\mathcal{C}_\Gamma \defeq 
\begin{cases}
\big\{[\cdot]\big\} & \text{if $\Gamma = \emptyset$} \\
\big\{ (\nub x)(C[\cdot] \para P) \suchthat C[\cdot] \in \mathcal{C}_{\Gamma'} \land  P \in \chrp{{T}}{x} \big \} & \text{if $\Gamma = \Gamma', x:T$}
\end{cases}
$$
\end{definition}

\noindent
Given a context $\Gamma = x_1: T_1, \ldots, x_n: T_n$, let us write $\dual{\Gamma}$ to denote the context 
$x_1: \dual{T_1}, \ldots, x_n: \dual{T_n}$, i.e., the context
obtained by ``dualising'' all the types in $\Gamma$.
The following statement formalizes the complementarity, in terms of session behaviors, between a well-typed process in $\lcp$ and its associated catalyzers:

%
%
%
\begin{lemma}[Catalyzers Preserve Typing]
\label{lem:typability_catal}
Let $P\cp \encob{\Gamma}{\ct}, \encob{\Gamma'}{\ct}$ and $C[\cdot] \in \mathcal{C}_{\dual {\Gamma}}$.
Then
$C[P]\cp  \encob{\Gamma'}{\ct}$.
\end{lemma}
\begin{proof}
Follows from \Cref{def:catalyser}, which defines $C[\cdot]$ as a composition of characteristic processes, and \Cref{lem: wt_characteristic} (\Cref{item1: wt_characteristic}), which ensures the appropriate type for each of them.
\end{proof}



\subsection{Translating \fullKoba into \lcp}\label{ss:rw}
{Our translation, given in \Cref{def:typed_enc}, transforms a session-typed process in \fullKoba into a \emph{set} of \lcp processes. 
Unsurprisingly, a delicate point in this translation is the treatment of parallel composition, which follows the intuition of the transformation of $P$ into $P'$ motivated in \Cref{ex:simple}: keep half of a parallel process unchanged, and increase the parallelism in the other half. \rev{This way, our translation returns a set of processes because we wish to account for both these two alternatives, for the sake of generality.}

Intuitively, given a well-typed session process $P_1 \para P_2$, our translation produces 
two sets of processes: the first one collects processes of the form $Q_1 \para G_2$, where 
$Q_1$ composes the translation of $P_1$ within an appropriate catalyzer
and
$G_2$ is a characteristic process that implements all sessions declared in the typing of $P_2$.
Similarly, the second set collects  processes of the form 
$G_1 \para Q_2$, where 
$Q_2$ composes the translation of $P_2$ within an appropriate catalyzer
and
$G_1$ is a characteristic process that implements all sessions declared in the typing for $P_1$.
This way, by translating one subprocess and replacing the other with parallel representatives ($G_1$ and $G_2$), 
the translated processes are more independent, and  the circular dependencies---the heart of deadlocked processes---are systematically ruled out. 

We require some auxiliary notations.

\begin{notation}[Bound Names] 
\label{n:bn}
Concerning bound names and their session types:
\begin{itemize}
\item We annotate bound names with their session types:
we write $\res {x_1y_1:T_1}\cdots\res {x_ny_n:T_n}P$
and $x\inp {y:T}.P$, for some session types $T,T_1, \ldots, T_n$.
\item We sometimes write 
$\Gamma, \wt{x:T}\s P$
as shorthand notation for 
$\Gamma, x_1:T_1, \ldots, x_n:T_n \s P$.
Similarly, we write 
$\res {\wt{x}\wt{y}:\wt{T}}P$ 
to abbreviate 
$\res {x_1y_1:T_1}\cdots\res {x_ny_n:T_n}P$.
\end{itemize}
\end{notation}

Using these notations for bound names,
we introduce the following notation for well-typed parallel processes, in which ``hidden'' sessions are explicitly denoted by brackets:
}



\begin{notation}[Hidden/Bracketed Sessions]
\label{not:para}
We shall write 
$${\Gamma_1, \restricted{\wt{x:S}} \circc \Gamma_2, \restricted{\wt{y:T}}}\s \res {\wt{x}\wt{y}:\wt{S}}(P_1\para P_2)$$
whenever
${\Gamma_1,\Gamma_2}\s \res {x_1y_1}\cdots \res {x_ny_n}(P_1\para P_2)$
holds
with 
$\Gamma_1, x_1:S_1, \ldots,  x_n:S_n \s P_1$, 
$\Gamma_2, y_1:T_1, \ldots, y_n:T_n \s P_2$,
and 
$ S_i = \dual{T_i}$, for all $i \in \{1, \ldots, n\}$.
\end{notation}
We are now ready to give the first translation from session processes into \lcp:

\begin{definition}[Translation into $\lcp$]\label{def:typed_enc}
Let $P$
be such that $\Gamma\s P$ and $P \in \lkoba$. 
The set of \lcp processes $\encoCP{\Gamma\s P}{}$ is defined in \Cref{figure:first_rewriting}.
\begin{figure}[!t]
    \begin{mdframed}
    \vspace{-1em}
\begin{align*}
\encCP{\Gamma^\eende \s \nil}								& \defeq  \big\{\nil \big\}
\\[1mm]
\encCP{{\Gamma}, x:\oc T.S, v:T  \s\dual{x}\out v.P'}		& \defeq 
\big\{ \dual{x}(z).(\linkr{v}{z}\para Q)  \suchthat Q \in \encCP{{{\Gamma}, x:{S}}\s P'}\big\}
\\[1mm]
{\encCP{{\Gamma_1,\Gamma_2}, x:\oc T.S \s\res{zy}\dual{x}\out y.(P_1\para P_2)}}	&\defeq\\
\big\{ \dual{x}(y).(Q_1\para Q_2)  \suchthat~ & Q_1 \in \encCP{{{\Gamma_1},  z:\ov T\s P_1}}
\wedge  Q_2 \in \encCP{{{\Gamma_2},x:{S}\s P_2}}\big\}
\\[1mm]
\encCP{{\Gamma}, x:\wn T.S	 \s x\inp {y:T}.P'}			& \defeq
\big\{x(y).Q  \suchthat Q \in \encCP{{\Gamma,x:S,y:T}\s P'}	\big\}			
\\[1mm]
\encCP{{\Gamma}, x:\select lS \s\selection x{l_j}.P'}		& \defeq
\big\{\selection x{l_j}.Q  \suchthat Q \in {\encCP{{\Gamma,x:S_j}\s P'}}\big\}								 
\\[1mm]  	
\encCP{{\Gamma}, x:\branch lS \s \branching xlP}		&  \defeq
\big\{\parbranching x{l_i}{Q_i}  \suchthat Q_i \in \encCP{{\Gamma,x:S_i}\s P_i}\big\}
\\[2mm]  
\encCP{\Gamma_1, \restricted{\wt{x:S}}\circc \Gamma_2, \restricted{\wt{y:T}}\s \res {\wt{x}\wt{y}:\wt{S}}(P_1\para P_2)}
&\defeq  \\
\big\{
C_1[Q_1] \para G_2  \suchthat~ &Q_1 \in \encCP{{\Gamma_1, {\wt{x:S}}}\s P_1}, \, C_1 \in \mathcal{C}_{\wt{x:T}},\, G_2 \in  \chrpk{\Gamma_2}{}  \big\}  
\\[1mm]							
\cup
\\
\big\{G_1 \para C_2[Q_2]   \suchthat~ &   Q_2 \in \encCP{\Gamma_2,\wt{y:T}\s P_2}, \, C_2 \in \mathcal{C}_{\wt{y:S}},\,
G_1 \in \chrpk{\Gamma_1}{}
\big\}
\end{align*}
    \end{mdframed}

\caption{Translation $\encCP{\cdot}{}$ (cf. \Cref{def:typed_enc}).\label{figure:first_rewriting}}
\end{figure}
\end{definition}
Our  translation operates on typing judgments: a well-typed session process is translated using the information declared in its typing context. Although the translation could be defined for arbitrary  session processes (even  deadlocked ones), membership in $\lkoba$ plays a role in operational correspondence (see below).
We discuss the different cases in \Cref{figure:first_rewriting}:
\begin{enumerate}[label=$\bullet$]
\item The process   $\nil$  is translated into the singleton set $\{\nil\}$ provided that the 
associated typing context $\Gamma$ contains only completed sessions (recall that   $\Gamma^\eende$ stands for $\eend{\Gamma}$). 
\item The translation of output- and input-prefixed processes is self-explanatory; in the former case, we translate the free output available in $\lkoba$ by exploiting a forwarding process in \lcp (cf. \mydefref{d:trans}).
The translation of selection and branching processes also follows expected lines.
\item The last case of the definition handles processes in parallel, possibly with restricted sessions; we use \Cref{not:para} to make
such sessions explicit.
As hinted at above, the translation of a parallel process $\res{\wt{x}\wt{y}:\wt{S}}(P_1\para P_2)$ 
in \lkoba results into two different sets of \lcp processes:
the first set contains processes of the form $C_1[Q_1] \para G_2$, where, intuitively:
\begin{enumerate}[label=$\bullet$]
\item $Q_1$ belongs to the set that results from translating subprocess $P_1$ with an appropriate typing judgment, which includes $\wt{x:S}$.
\item $C_1$ belongs to the set of catalyzers that implement the 
context $\wt{x:T}$, i.e., 
the dual behaviors of the sessions implemented by $P_1$ (cf. \mydefref{def:catalyser}).
This step thus removes the double restriction operator. 
\item $G_2$ belongs to the set of characteristic processes for $\Gamma_2$, which describes the sessions implemented by $P_2$ (cf. \mydefref{def:charenv}).
\end{enumerate}
The  explanation for the processes of the form $G_1 \para C_2[Q_2]$ in the second set is completely dual.

As we will see, processes $C_1[Q_1] \para G_2$ (and $G_1 \para C_2[Q_2]$) preserve by construction the typing of $\res{\wt{x}\wt{y}:\wt{S}}(P_1\para P_2)$:
process $C_1[Q_1]$ (resp. $C_2[Q_2]$) is typable with context $\Gamma_1$ (resp. $\Gamma_2$); 
process $G_2$ (resp. $G_1$) is typable with context $\Gamma_2$ (resp. $\Gamma_1$)---see Theorem\,\ref{thm:L0-L2} below.
\end{enumerate}


We illustrate the translation by means of an example. 


\begin{example}[The Translation $\encCP{\cdot}$ at Work]
\label{ex::K1inK2}
Consider again the process $P$ used in \Cref{lem:K1inK2}.
Let $T \defeq \oc\nilT.\nilT$ and $S \defeq \wn\nilT.\nilT$. Clearly, $\dual{S} = T$.
We have the following derivation. 
%
\medskip
\begin{displaymath}
\infer[\Did{T-Res}]
{
\unit:\nilT\s  \res{a_1b_1}\res{a_2b_2}(a_1\inp x.\;\ov{a_2}\out{x}.\zero
        \pp \ov{b_1}\out{\unit}.\;b_2\inp{z}.\zero)
}
{
	\infer[\Did{T-Res}]
	{a_1: \wn\nilT.\nilT, b_1:\oc\nilT.\nilT, \unit:\nilT\s
	\res{a_2b_2}(a_1\inp x.\;\ov{a_2}\out{x}.\zero\pp \ov{b_1}\out{\unit}.\;b_2\inp{z}.\zero)}
	{
		\infer[\Did{T-Par}]
		{a_2:T, b_2:S, a_1: S, b_1:T, \unit:\nilT \s a_1\inp x.\;\ov{a_2}\out{x}.\zero \pp \ov{b_1}\out{\unit}.\;b_2\inp{z}.\zero}
		{
			\inferrule*[left=\Did{T-In}]{ }{a_1: S, a_2: T \s a_1\inp x.\;\ov{a_2}\out{x}.\zero}
			\qquad
			\inferrule*[right=\Did{T-Out}]{ }{b_1:T,  b_2:S, \unit:\nilT \s \ov{b_1}\out{\unit}.\;b_2\inp{z}.\zero}
		}
	}
}
\end{displaymath}
\smallskip

Before detailing the set $\encCP{\unit:\nilT\s  P_2}$, we spell out the main ingredients required:
\begin{align*}
\encCP{a_1: S, a_2: T \s a_1\inp x.\;\ov{a_2}\out{x}.\zero} &= 
\{a_1\inp x.\dual{a_2}(z).(\linkr{x}{z}\para \nil)\} 
\\
\encCP{b_1:T,  b_2:S, \unit:\nilT \s \ov{b_1}\out{\unit}.\;b_2\inp{z}.\zero} &=
\{\dual{b_1}(u).(\linkr{\unit}{u}\para b_2\inp{z}.\nil)\}
\\
\chrp{T}{x} & = \{\dual{x}(z).(\nil \para \nil)\}
\\
\chrp{S}{x} & = \{x\inp{z}.(\nil \para \nil)\}
\\
\mathcal{C}_{a_1:T, a_2:S} &= \{ (\nub a_1)(R_1 \para (\nub a_2)(R_2 \para [\cdot])) \suchthat R_1 \in \chrp{S}{a_1}, R_2 \in \chrp{T}{a_2}\}
\\
&=  \{ (\nub a_1)(a_1\inp{z}.(\nil \para \nil) \para (\nub a_2)(\dual{a_2}(z).(\nil \para \nil) \para [\cdot])) \}
\\
\mathcal{C}_{b_1:S, b_2:T} &= \{ (\nub b_1)(Q_1 \para (\nub b_2)(Q_2 \para [\cdot])) \suchthat Q_1 \in \chrp{T}{b_1}, Q_2 \in \chrp{S}{b_2}\}
\\
&=  \{ (\nub b_1)(\dual{b_1}(z).(\nil \para \nil) \para (\nub b_2)(b_2\inp{z}.(\nil \para \nil) \para [\cdot])) \}
\end{align*}
Exploiting 
\Cref{not:para}, 
judgement 
$\unit:\nilT\s  P$
can be written as 
$$\restricted{a_1:S, a_2: T} \circc \unit:\nilT, \restricted{b_1:T, b_2: S} \s  \res{a_1b_1}\res{a_2b_2}(a_1\inp x.\;\ov{a_2}\out{x}.\zero
        \pp \ov{b_1}\out{\unit}.\;b_2\inp{z}.\zero)$$
We may now define the translation of $P$ into $\lcp$:
\begin{align*}
& \encCP{[a_1:S, a_2: T] \circc \unit:\nilT, [b_1:T, b_2: S] \s  \res{a_1b_1}\res{a_2b_2}(a_1\inp x.\;\ov{a_2}\out{x}.\zero
        \pp \ov{b_1}\out{\unit}.\;b_2\inp{z}.\zero)} 
 \\
 =~& \big\{
C_1[Q_1] \pp \chrpk{\unit:\nilT}{\unit}  ~~\suchthat C_1 \in \mathcal{C}_{a_1:T,a_2:S}, Q_1 \in \encCP{a_1: S, a_2: T \s a_1\inp x.\;\ov{a_2}\out{x}.\zero}
\big\} 
\\
& \cup 
\big\{
C_2[Q_2] \suchthat C_2 \in \mathcal{C}_{b_1:S,b_2:T}, Q_2 \in \encCP{b_1:T,  b_2:S, \unit:\nilT \s \ov{b_1}\out{\unit}.\;b_2\inp{z}.\zero} 
\big\} 
\\
=~& \big\{
(\nub a_1)(\dual{a_1}(z).(\nil \para \nil) \para (\nub a_2)(a_2\inp{z}.(\nil \para \nil) \para a_1\inp x.\dual{a_2}(z).(\linkr{x}{z}\para \nil))) \pp \chrpk{\unit:\nilT}{\unit} 
\, , \, 
\\
& \quad
(\nub b_1)(b_1\inp{z}.(\nil \para \nil) \para (\nub b_2)(\dual{b_2}(z).(\nil \para \nil) \para \dual{b_1}(u).(\linkr{\unit}{u}\para b_2\inp{z}.\nil))) 
\big\}
\end{align*}
Above, $P$ has two parallel components and so set $\encCP{\unit:\nilT\s  P}$ has two elements, representing
the two different possibilities for ``dividing'' the sequential structure of $P$ into more parallel processes.

\end{example}
\subsubsection{Properties}
We present two important results about our translation.
First, it is type preserving, up to the  encoding of types given in \Cref{f:enctypesct}:
\begin{theorem}[The Translation $\encCP{\cdot}$ is Type Preserving]
\label{thm:L0-L2}
Let $\Gamma\s P$.
Then, for all $Q\in \encCP{{\Gamma}\s P}$, we have that $Q\cp \encob{\Gamma}{\ct}$.
\end{theorem}
\begin{proof}
By induction on the derivation $\Gamma\s P$.
See \Cref{appthm:L0-L2} (Page~\pageref{appthm:L0-L2}) for details.
\end{proof}

\Cref{thm:L0-L2} is meaningful, for it says that the session type ``interface'' of a process (i.e., the set of sessions it implements) is not modified by the translation. 
That is, $\encCP{\cdot}$ modifies the process structure by closely following the 
session typing discipline.

To establish properties with respect to reduction, we start with a useful notation:

\begin{notation}
Let us write $\Gamma \s P_1, P_2$ whenever both $\Gamma \s P_1$ and
$\Gamma \s P_2$ hold. Similarly,
let us write $P_1,P_2 \cp \Gamma$, whenever both $P_1\cp \Gamma$ and $P_2\cp \Gamma$ hold.
\end{notation}

Before showing  how our translation satisfies  an operational correspondence result (cf. \Cref{thm:oc}),
we illustrate the need for a useful auxiliary definition, which will allow us to relate processes typable under the same typing context but featuring a different parallel structure
(cf. \Cref{def:uptok}).

\begin{example}[The Parallel Structure of Typable Processes]
\label{ex:uptok}
Let 
$P \cp \Delta$
where 
$P \defeq \res x ({\ov x\out v.P_1} \pp {x\inp z.P_2})$
and 
$\Delta = \Delta_1,\Delta_2,v:T$.
Consider the following reduction from $P$, obtained 
using the rules in \Cref{fig:redlcp}:
$$
\res x (\,\underbrace{\ov x\out v.P_1}_\text{$\Delta_1,v:T$} \pp \underbrace{x\inp z.P_2}_\text{$\Delta_2$}\,)
\to
\res x(\,\underbrace{P_1}_\text{$\Delta_1$}\pp \underbrace{P_2\substj{v}{z}}_\text{$\Delta_2,v:T$}\,) \defeq Q
$$
Here we are using a free output process, namely ${\ov x\out v.P_1}$, as a shortcut for 
$\dual{x}(z).(\linkr{v}{z}\para P_1)$; recall that this process is typable using Rules \Did{T-$\tid$} and  \Did{T-$\otimes$}  in \Cref{fig:type-system-cll}.
Observe that name $z$ is free with scope $P_2$; for typing to hold, 
it must be the case that $z:T$ in $P_2$.

By \Cref{thm:LL-type-preservation}, $Q\cp \Delta_1,\Delta_2,v:T$. 
Let us consider the typing of   $P$ and $Q$ in relation to their parallel structure.
We can notice that even though the typing context $\Delta$ remains the same under reduction, the ``parallel decomposition'' changes after reduction, as  highlighted by the under brackets. 
In particular, the type assignment $v:T$, at first related to the left-hand side component in $P$, 
``jumps'' after reduction
to the right-hand side   component in $Q$.
This phenomenon is due to 
value passing:
after reduction, a substitution occurs in the continuation process $P_2$, where $v$ might be used in different ways.
For instance, $v$ could occur within a sequential position within $P_2$ or at top-level in one of its parallel subprocesses.

Now, consider process $Q'$, which contains subprocesses $P_1$ and $P_2$ from $Q$ but has a different form of parallelism:
$$
Q'\defeq \res x(\,\underbrace{P_1}_\text{$\Delta_1$}\pp \underbrace{\res z (P_2 \para P_z)}_\text{$\Delta_2$} \pp \underbrace{P_v}_\text{$v:T$}\,)
$$
Here, $P_z$ and $P_v$ are characteristic processes implementing  $T$ along $z$ and $v$, respectively.
Clearly, $Q' \cp\Delta_1,\Delta_2,v:T$, just as $Q$, but its parallel structure is different:
in $Q'$ we have two separate subprocesses in parallel, one typed according to $\Delta_2$, the other according to $v:T$.
In $\res z (P_2 \para P_z)$, process $P_z$ provides one of the counterparts for $P_2$; this interaction is hidden by restriction $\res z$ and so the resulting typing context is $\Delta_2$.
By using $P_v$,  the type interface of $Q'$ is exactly as that of $Q$; 
by using it in parallel position, their parallel structure will be different whenever $z$  occurs in $P_2$ in a sequential position. 
\end{example}

We are interested in capturing the \emph{parallelization relation} between processes such as $Q$ and $Q'$ in \Cref{ex:uptok}, because it serves to explain how translated processes behave under reduction. 
This relation is formalised in \Cref{def:uptok} below:

\begin{definition}[Parallelization Relation]
\label{def:uptok}
Let $P$ and $Q$ be processes such that $P,Q\cp \Gamma$.
We write $P\uptok Q$ if and only if there exist processes
$P_1, P_2, Q_1, Q_2$ and contexts  $\Gamma_1, \Gamma_2$ such that the following hold:
\begin{align*}
P = P_1\pp P_2
\qquad
Q = Q_1\pp Q_2
\qquad
P_1, Q_1\cp \Gamma_1
\qquad
P_2, Q_2\cp \Gamma_2
\qquad
\Gamma = \Gamma_1,\Gamma_2
\end{align*}
\end{definition}

This definition says that two processes $P$ and $Q$, typed under the typing context $\Gamma$, are related by $\uptok$, if they can be decomposed into parallel subprocesses, which are typed under the same decomposition of $\Gamma$. This way, for processes $Q$ and $Q'$ from \Cref{ex:uptok}, relation $Q\uptok Q'$ holds.

\nurev{By definition, the relation $\uptok$ is reflexive.} It will appear in our operational correspondence result, given next. 
Below,  let $\fred$ denote structural congruence (cf \Cref{ss:pm}) extended with a reduction by Rule~\Did{R-Fwd} (cf.~\Cref{ss:lf}). 
We may now state:
\begin{theorem}[Operational Correspondence for $\encCP\cdot$]
\label{thm:oc}
Let $P$ be such that $\Gamma\s P$ for some typing context $\Gamma$. Then, we have:
\begin{enumerate}
\item\label{item1:oc}
\nurev{If $P\to P'$, then for all $Q \in \encCP{{\Gamma}\s P}$ there exist $Q', R$ such that
$Q\to\fred Q'$, $Q' \uptok R$, and $R \in \encCP{{\Gamma}\s P'}$.}

\item\label{item2:oc}
\nurev{If $Q \in \encCP{{\Gamma}\s P}$, such that $P\in \lkoba$, and $Q \to\fred Q'$, then there exist
$P', R$ such that
$P\to P'$,
$Q' \uptok R$, and $R \in \encCP{{\Gamma}\s P'}$.}
\end{enumerate}
\end{theorem}
\begin{proof}
By induction on the length of the derivations $P\to P'$ and $Q \to Q'$.
See \Cref{appthm:oc} for  details.
\end{proof}

{\noindent \Cref{item1:oc} of \Cref{thm:oc} certifies that our translation tightly preserves the behavior of the session process given as input. The parallelization relation $\uptok$ (cf. \Cref{def:uptok}) is crucial when the reduction $P\to P'$ involves value passing: in that case, process $Q'$, obtained from the translated process $Q$, may have a different parallel structure than $R$---just as $Q$ and $Q'$ from \Cref{ex:uptok}.
 \Cref{item2:oc} relates the behavior of a translated process with respect to that of the {session process} given as input, with $\uptok$  playing a role similar as in  \Cref{item1:oc}. Unlike \Cref{item1:oc}, in \Cref{item2:oc} we require $P$ to be in $\lkoba$, i.e., $P$ must be deadlock-free for the correspondence to hold. Indeed, if $P$ is not in $\lkoba$ then  $Q$ (its deadlock-free, translated variant) could have reductions not enabled in $P$ due to deadlocks.}

\subsection{Discussion: Translating \lkoba into \lcp Exploiting Value Dependencies}\label{ss:oprw}


Process prefixes can induce \emph{causality relations} not accounted for by session types.
One kind of causality relations are at the heart of deadlocked processes, in which session interleavings in process prefixes  leads to circular dependencies between independent sessions (cf. \Cref{ex:basic}).
Another kind of causality relations are the \emph{value dependencies} induced by value exchanges: they occur when 
a value received in one session is sent along a different one.

\begin{example}[A Value Dependency]\label{ex:vd1}
Let $P$ be the process 
$$P \defeq \res{a_0b_0}
(
\ov{a_0}\out{\unit}.a_1\inp u.\ov{a_2}\out{u}.\nil \pp 
b_0\inp{v}.(b_2\inp{y}.y\inp{x}.\nil \pp 
\res{wz}(\ov{b_1}\out{w}.\ov{z}\out{\unit}.\nil)))$$
Also, let $U \defeq \wn \nilT.\nilT$. Consider the typing judgment for the 
leftmost parallel process
$\ov{a_0}\out{\unit}.a_1\inp u.\ov{a_2}\out{u}.\nil$:
\begin{align*}
a_0:\oc \nilT.\nilT, a_1:\wn U.\nilT, a_2:\oc U.\nilT & \s 
\ov{a_0}\out{\unit}.a_1\inp u.\;\ov{a_2}\out{u}.\nil
\end{align*}
Although the typing for   $a_1$ and $a_2$ states that they are independent  sessions, 
actually they are causally related: the process forwards along $a_2$ the value of type $U$
received on $a_1$.
By considering the typing for $P$
\begin{align*}
a_1:\wn U.\nilT, a_2:\oc U.\nilT, b_2:\wn U.\nilT, b_1:\oc U.\nilT & \s P
\end{align*}
we see that the forwarded value is $w$, which is delegated by the rightmost parallel process, $\res{wz}(\ov{b_1}\out{w}.\ov{z}\out{\unit}.\nil)$.
\end{example}
In the terminology of Boreale and Sangiorgi~\cite{DBLP:journals/acta/BorealeS98}, value dependencies are both \emph{subject} and 
\emph{object} dependencies. In $\ov{a_0}\out{\unit}.a_1\inp u.\ov{a_2}\out{u}.\nil$
there is a subject dependency: the input on $a_1$ enables the output on $a_2$;
there is also an object dependency: the  name received on $a_1$ is used as object in the output on $a_2$.
Indeed, Boreale and Sangiorgi argue that in most cases an object dependency is also a subject dependency.


While intuitive, the translation $\encCP{\cdot}$   does not preserve {value dependencies}:
in translating a process with parallel components, a subprocess containing a value dependency can be replaced by 
an equally typed process in which such a dependency is no longer present:

\begin{example}[Value Dependencies in $\encCP{\cdot}$]
\label{ex:vd2}
Let process $P$ and type $U$ be as in \Cref{ex:vd1}. 
 Consider the set of processes $\encCP{a_1:\wn U.\nilT, a_2:\oc U.\nilT, b_1:\oc U.\nilT, b_2:\wn U.\nilT  \s P}$ obtained using 
 \Cref{def:typed_enc}.
 One   process in this set is the following:
\begin{align*}
Q = G_{1} \para G_{2} \para \mathcal{C}_{b_0:\wn \nilT.\nilT}\big[b_0\inp{v}.(b_2\inp{y}.y\inp{x}.\nil \pp 
\res{wz}(\ov{b_1}\out{w}.\ov{z}\out{\unit}.\nil))\big]  
\end{align*}
where 
$G_{1} \in \chrp{\wn U.\nilT}{a_1}$ and $G_{2} \in \chrp{\oc U.\nilT}{a_2}$.
Since $G_{1}$ and $G_{2}$ are independently defined, 
the name received along $a_1$ in $G_{1}$
cannot be the same session sent along $a_2$ in $G_{2}$.
Thus, the value dependence between $a_1$ and $a_2$ in $P$, has disappeared in its translation as $Q$.
\end{example}


To address this limitation of $\encCP{\cdot}$, we have defined an optimized translation that preserves value dependencies. 
\ifarxiv
Here we discuss the key elements of this optimization; the full technical development is presented in Appendix~\ref{sec:appendix-6}. 
\else 
Here we discuss the key elements of this optimization; the full technical development is presented in the online appendix~\cite{DP18}. 
\fi 
The optimization is obtained as follows:
\begin{enumerate}
\item 
Detecting value dependencies requires gathering further information on how types are implemented by process prefixes.
To this end, we   extend the type system of \Cref{s:sessions}
with 
\emph{annotated} output and input session types, written $\oc^n S.T$ and $\wn^n S.T$ (with $n \geq 0$), which specify the position of an associated prefix within the process.
This way, e.g., type $\oc^0 S.T$ is associated to an output prefix at top-level.
Also, we have typing judgments 
$$\Gamma \compsi\Psi \s P$$
where the typing 
context $\Gamma$ is as before and 
$\Psi$ is a 
\emph{dependency context}  
that describes all communication prefixes in $P$ and their distance to top-level.
This context contains tuples of the form 
$(a_1, u, 1)$ 
and 
$\langle a_2, u, 2 \rangle$;
the former can be read as 
``there is an input on $a_1$ with placeholder $u$ one prefix away from top-level'', whereas the 
latter can be read as
 ``there is an output on $a_2$ with object $u$ two prefixes away from top-level''.
Using $\Psi$, we formally define value dependencies as pairs:  
we write $\vdep{a^n}{b^m}$ to denote a 
  value dependency of an (output) prefix along session $b$ on an (input) prefix along session $a$.
  
\item We exploit the value dependencies in  $\Psi$ to refine the definitions 
of characteristic processes (\mydefref{def:charprocess}) and  catalyzer contexts (\Cref{def:catalyser}) of a typing context $\Gamma$:
\begin{itemize}
\item The set of characteristic processes of $\Gamma$ that exploits $\Psi$ is denoted $\chrpv{\Gamma}{}{\Psi}$:
it implements value dependencies in $\Psi$ using so-called \emph{bridging sessions} that connect the characteristic processes of the two types involved in the dependency. 
\item The set of catalyzer contexts of $\Gamma$  handles value dependencies in  $\Psi$ by ``isolating'' them using dedicated forwarder processes and characteristic processes. 
\end{itemize}
 
\item Using these refined definitions, we define the optimized translation~$\encCPalt{\cdot}{}$.
Main differences with respect to the translation defined in \Cref{figure:first_rewriting} appear in the translation of 
parallel processes. As before, process $\res {\wt{x}\wt{y}:\wt{S}}(P_1\para P_2)$ is translated into processes of two forms: $C_1[Q_1] \para G_2$ and $ G_1 \para C_2[Q_2]$, where $Q_1$ (resp. $Q_2$) stands for the translation of $P_1$ (resp. $P_2$). The difference is that the catalyzers $C_1, C_2$ and the characteristic processes $G_1, G_2$ are now obtained exploiting value dependencies (as explained in~(2)).
As $\encCP{\cdot}$, 
the optimized translation~$\encCPalt{\cdot}{}$ 
satisfies type preservation and operational correspondence. 
\end{enumerate}

\medskip
\noindent
The following example illustrates how $\encCPalt{\cdot}{}$ improves over $\encCP{\cdot}{}$.

\begin{example}[Revisiting Examples~\ref{ex:vd1} and \ref{ex:vd2}]
Let $P$ be as in \Cref{ex:vd1}:
$$P \defeq \res{a_0b_0}(
\ov{a_0}\out{\unit}.a_1\inp u.\ov{a_2}\out{u}.\nil \pp 
b_0\inp{v}.(b_2\inp{y}.y\inp{x}.\nil \pp 
\res{wz}(\ov{b_1}\out{w}.\ov{z}\out{\unit}.\nil)))
$$
In the extended type system, we have
$\Gamma \compsi\Psi \s P$ 
where
\begin{align*}
\Gamma & \defeq a_1:\wn^{1} U.\nilT, a_2:\oc^{2} U.\nilT, b_2:\wn^{1} U.\nilT, b_1:\oc^{1} U.\nilT 
\\
\Psi & \defeq \langle a_0, \unit, 0 \rangle, (a_1, u, 1), \langle a_2, u, 2 \rangle, 
(b_0, v, 0), (b_2, y, 1), (y, x, 2),  
\langle b_1, w, 1 \rangle, \langle z, \unit, 2 \rangle
\end{align*}
Using $\Psi$, we can detect the value dependence $\vdep{a_{1}^1}{a_{2}^2}$.
 Consider now $Q'$, one particular process included in the set $\encCPalt{\Gamma \compsi\Psi \s P}$:
\begin{align*}
Q' = (\nub c_{a_1a_2})(G'_{a_1} \para G'_{a_2}) \para \mathcal{C}_{b_0:\wn \nilT.\nilT}\big[b_0\inp{v}.(b_2\inp{y}.y\inp{x}.\nil \pp 
\res{wz}(\ov{b_1}\out{w}.\ov{z}\out{\unit}.\nil))\big]  
\end{align*}
where 
$G'_{a_1} \in \chrpv{\wn U.\nilT}{a_1}{\Psi}$ and $G'_{a_2} \in \chrpv{\oc U.\nilT}{a_2}{\Psi}$
are obtained using the refined definition of characteristic processes.
With the refined definition, we have that 
\begin{align*}
G'_{a_1} &= a_1(y).\bout{c_{a_1a_2}}{w}.(\linkr{y}{w} \para \nil)
\\
G'_{a_2} &= c_{a_1a_2}(y).\bout{a_2}{w}.(\linkr{y}{w} \para \nil)
\end{align*}
Thus, the characteristic processes of a type are still independently defined but now implement a bridging session $c_{a_1a_2}$, which ensures that the name received along $a_1$ in $G'_{a_1}$
is the same name outputted along $a_2$ in $G'_{a_2}$. The characteristic process of a typing context guarantees that bridging sessions are properly composed.
This way, the value dependence between $a_1$ and $a_2$, present in $P$, has been preserved in $Q'$ as a result of the refined translation.
\end{example}

Although intuitive, the way in which 
$\encCPalt{\cdot}{}$ improves over $\encCP{\cdot}{}$
is somewhat implicit, because value dependencies are accounted for by the revised definitions of characteristic processes and catalyzers as restricted (hidden) sessions. 
Therefore, the implemented  value dependencies in these definitions does not influence the observable behavior of the translated process. Formalizing the advantages of 
$\encCPalt{\cdot}{}$ over $\encCP{\cdot}{}$ is an interesting question for future work, because it
requires devising new notions of observables for typed processes that ``look inside'' reductions to ensure that values are forwarded appropriately between the two independent sessions.

\section{Discussion}\label{s:discuss}

\paragraph{Processes with Unbounded Behavior}
Our investigation has been motivated by the proliferation of type systems
for ensuring safety and liveness properties of mobile, concurrent processes. 
Different type systems enforce different such properties, which 
include various forms of (dead)lock-freedom, termination, and confluence.
In this work, our criteria have been twofold. On the one hand, we have aimed at obtaining
objective 
formal comparisons between well-established type  systems, sticking to their original formulations as much as possible. 
On the other hand, we have concentrated on the intrinsic challenges of statically enforcing deadlock freedom. 
We have focused on  typed processes without 
constructs for expressing processes with unbounded behavior, such as replication or recursion. 
This focus was useful to be consistent with these criteria; next we discuss some of the issues involved in going beyond this class of finite processes. 

In the Curry-Howard correspondences for session types,  sharing exponentials $!$ and $?$ at the level of 
propositions/types can be interpreted as input-guarded replication $!x\inp z.Q$ at the level of proofs/processes. 
In the classical setting  considered here, 
a channel typed with $!A$ denotes a \emph{server} able to offer an arbitrary number of copies (including zero) of a process with behavior of type $A$ upon request (an input); 
dually, a channel typed with $?A$ denotes a \emph{request} to such a server (an output). 
This requires an additional cut rule for exponential (unrestricted) contexts (denoted $\Theta$) as well as typing rules for realizing the sharing semantics of $!A$ and $?A$. Other rules are adjusted to account for  $\Theta$; for instance, the mix rule becomes:
\begin{align*}
\inferrule
{P \cp \Delta; \Theta  \and Q \cp  \Delta'; \Theta}
{ P\para Q \cp\Delta, \Delta'; \Theta}
\end{align*}
The resulting logically justified  reduction rule for replication 
is as follows:
\begin{align}
\res{x}(\bout{x}{y}.P\pp !x\inp z.Q) 
\to 
\res{x}(\res{y}(P \pp Q) \pp !x\inp z.Q) 
\label{eq:repl}
\end{align}
The process before the reduction features one linear cut along session $x$; 
after the reduction, the process contains two cuts:
the cut on $y$ is linear, whereas the one on $x$ is exponential.

The definition of the typed languages \lcp and \lkoba (\Cref{d:lang}) can be extended to consider typed processes with replication, which is used in \cite{CairesP10,DBLP:journals/mscs/CairesPT16} (in the input-guarded variant described above) and in \cite{K06} (where unguarded replication $!P$ is accounted for). 
Since the type system $\s$ in \cite{V12} admits rather expressive forms of recursion (that go well beyond the server-request interactions enabled by input-guarded replication), the class of ``full \lcp'' processes would be a strict sub-class of session-typed processes.


Now, considering processes with unbounded behavior entails including \emph{termination} properties  into the analysis of (dead)lock-freedom. 
Crucially, the full Curry-Howard interpretation for session types, including exponentials as explained above, is known to be strongly normalizing and confluent~\cite{DBLP:journals/iandc/PerezCPT14,CP17}.
Therefore, unbounded behavior in full \lcp concerns unboundedly many copies of finite, deterministic  interactive behaviors. 
The fact that a single type system simultaneously enforces deadlock freedom, termination, and confluence sharply contrasts to the situation for non-logical type systems: to our knowledge, only the hybrid type system in~\cite{KS10} simultaneously ensures these three properties (see below).
Clearly, it would be unfair to compare processes 
that enjoy different properties, i.e., processes
in $\lcp$ against  
well-typed processes in type systems that enforce some, but not all, of 
deadlock freedom, termination, and confluence. 
By focusing on finite processes,  we have found a fair ground to objectively compare different type systems.

The integration of non-logical type systems for (dead)lock-freedom, termination, and confluence is far from trivial, and requires advanced mechanisms.
Kobayashi and Sangiorgi~\cite{KS10} targeted this goal by defining a parametric hybrid type system based on usages, which enforces the three properties through different methods (not necessarily type systems).
As in~\cite{K02}, the syntax of usage types in~\cite{KS10} includes 
 replicated usages  $\ast U$, i.e., unboundedly many parallel copies of usage $U$. (The authors remark that the recursive usages  $\mu \alpha. U$ from \cite{K06} are also sound.)
To enable fair comparisons against  full \lcp,  Kobayashi and Sangiorgi's type system should be restricted so that it allows only  input-guarded replicated  processes (cf.~\eqref{eq:repl}) and 
simultaneously enforces (dead)lock-freedom, termination, and confluence. 
Identifying the syntactic/semantic conditions that enable this restriction seems challenging. 
Defining such a restricted variant of~\cite{KS10}   
 would most likely mean developing a very different  type system,
therefore departing from the compact abstractions given by usage types.
Hence, comparing such a different type system against the canonical language full $\lcp$ would also be unfair. 

We notice that Curry-Howard interpretations of session types have been extended with forms of (co)-recursive types,  which extend (full) \lcp by admitting as typable certain  forms of (productive) unbounded behavior---see the works by Toninho et al.\cite{DBLP:conf/tgc/ToninhoCP14} and by Lindley and Morris~\cite{DBLP:conf/icfp/LindleyM16}.
Indeed, the framework in~\cite{DBLP:conf/icfp/LindleyM16} can be roughly seen as the extension of the logic-based type system in \Cref{ss:lf} with co-recursion.
Although  these extensions bring closer logically motivated and non-logical type systems (which often do not ensure termination), the forms of co-recursive unbounded behavior enabled by~\cite{DBLP:conf/tgc/ToninhoCP14,DBLP:conf/icfp/LindleyM16} preserve the intrinsic confluent, deterministic behavior inherited from logical foundations. 
This prevents fair comparisons with 
 type systems for deadlock freedom of unbounded/cyclic 
communication structures
which do not guarantee confluence/determinism, such as those by Giachino et al.~\cite{DBLP:conf/concur/GiachinoKL14,DBLP:journals/iandc/0001L17}.
In contrast, the type system for deadlock freedom by Padovani~\cite{P14} ensures a form of \emph{partial confluence}, inherited from~\cite{DBLP:journals/toplas/KobayashiPT99}.
Hence, 
it would seem that there is common ground for 
comparing the linear type system in~\cite{P14} and the logically motivated session type system in~\cite{DBLP:conf/icfp/LindleyM16}.
The actual feasibility of relating these different type systems remains unclear, and should be established in future work.


\paragraph{Other Curry-Howard Interpretations of Session Types}
\nurev{As already discussed, our work concerns the Curry-Howard correspondences between (classical) linear logic propositions and session types developed in~\cite{CairesP10,DBLP:journals/mscs/CairesPT16,DBLP:conf/icfp/Wadler12}. In the following, we briefly discuss our results in the context of other  correspondences between (variants of) linear logic and session types~\cite{DG18,DBLP:journals/pacmpl/KokkeMP19,DBLP:journals/pacmpl/QianKB21}.}

\nurev{Dardha and Gay~\cite{DG18} developed Priority-based Classical Processes (PCP), a session type system that rests upon an extension of classical linear logic with Kobayashi's obligations/capabilities, simplified to priorities by following Padovani~\cite{P14}. Unlike the type system in~\cite{DBLP:conf/icfp/Wadler12}, the type system of PCP admits processes with \emph{safe} cyclic topologies. To this end, PCP cuts ties with the ``composition plus hiding'' principle that is at the heart of \lcp: the Rule~\Did{T-$\cut$}  is replaced by rules  \Did{T-$\mix$} and \Did{T-$\cycle$} that separately treat parallel composition and restriction, respectively. As a result, a \emph{multicut} rule, which allows composing two processes that may share more than one session, is derivable in PCP. This way, the class of typable processes induced by PCP strictly includes $\lcp$.}

\nurev{Kokke et al.~\cite{DBLP:journals/pacmpl/KokkeMP19} developed Hypersequent Classical Processes (HCP), an interpretation of session types based on the same classical linear logic as in~\cite{DBLP:conf/icfp/Wadler12} but with a hypersequent presentation. The design of HCP neatly induces a labeled transition semantics for typed processes; it has been further investigated in, e.g., \cite{DBLP:journals/pacmpl/QianKB21,DBLP:conf/concur/0001KDLM21}.  Also, similarly to PCP, HCP includes separate typing rules for restriction and parallel composition. However, HCP cannot type the kind of safe cyclic process topologies that are typable in PCP: as such, the separate treatment for restriction and parallel composition (enabled by the hypersequent presentation of CLL) does not enhance by itself the topology of processes, and the  typable processes in HCP have a tree topology as in \lcp. We conjecture that our separation and unifying results  also hold, with minor modifications, for a variant of \lcp based on HCP. We leave such technical investigation as future work.} 

\nurev{Qian et al.~\cite{DBLP:journals/pacmpl/QianKB21} recently proposed Client-Server Linear Logic (CSLL) and a corresponding session type system for $\pi$-calculus processes, which uses the hypersequent presentation of~\cite{DBLP:journals/pacmpl/KokkeMP19}. By introducing \emph{coexponentials}, CSLL captures more expressive client-server behaviours than those usually admitted by interpretations of CLL as session types. Their approach consists in abandoning the mix principles as adopted by most systems based on CLL (including the one we define in \Cref{ss:lf}). This way, the mix rules in CSLL only concern disjoint concurrency, as captured by hyperenvironments. As already discussed, in our work the Rule~\Did{T-$\mix$} enables us to type the 
{independent} parallel composition of processes, as needed in several aspects of our development, such as our separation result based on \muKoba. Exploring how to adapt our results to a system without mix principles is an interesting item for future work.}

\paragraph{Correctness Criteria for the Translation}
\nurev{
The translation $\encCP\cdot$ defined in \Cref{s:enco} is supported by an operational correspondence result, \Cref{thm:oc}. Operational correspondence is a well-established correctness criterion, which allows us to certify the quality of a language translation. There exist other correctness criteria, which typically serve different purposes~\cite{Gorla10,DBLP:journals/entcs/Parrow08}. A notable criterion is \emph{full abstraction}, i.e., that a translation from a source to a target language preserves the respective equivalences. While full abstraction is not informative enough to assert the quality of a translation (see \cite{DBLP:journals/mscs/GorlaN16,DBLP:journals/mscs/Parrow16} for a discussion), it is quite appropriate to formalise the transfer of reasoning techniques between different calculi. Studying $\encCP\cdot$ from the perspective of full abstraction would be insightful and requires developing (typed) behavioural equivalences for \lkoba and \lcp. We leave such technical developments as future work.
}

\section{Related Work}\label{s:rw}
The analysis of deadlock freedom in concurrency has a rather long and rich history, as discussed in details in~\cite{DBLP:conf/tacs/AbramskyGN97}.
Focusing on 
type-based approaches to deadlock freedom for communicating processes, early works are by
Kobayashi~\cite{DBLP:conf/lics/Kobayashi97}
and by Abramsky et al.~\cite{DBLP:conf/tacs/AbramskyGN97}.
The work in~\cite{DBLP:conf/tacs/AbramskyGN97} develops a semantic approach to deadlock freedom for asynchronous communicating processes, building upon categorical foundations.
The work in~\cite{DBLP:conf/lics/Kobayashi97} proposes a type system for the $\pi$-calculus that builds upon two key ideas: (i)~the introduction of usage of channels as types (\emph{usage types}), and (ii)~the classification of channels into reliable and unreliable (where reliable channels are ensured to have deadlock-free interactions).
These ideas have proved rather influential: based on them, a number of extensions and enhancements to type systems for deadlock-free, message-passing processes have been introduced. Kobayashi~\cite{K07} offers a unified presentation of these developments. The work of Kobayashi and Laneve~\cite{DBLP:journals/iandc/0001L17} builds upon notions of usage types and reliable channels for the type-based analysis of unbounded process networks.

The introduction of type systems based on usage types coincided in time with the introduction of (binary) session types as a type-based approach to ensure safe structured communications~\cite{H93,THK94,HVK98}. In their original formulation, session types for the $\pi$-calculus ensure communication safety and session fidelity, therefore ruling out some (but not all!) of the causes of deadlocked behaviors in communicating processes.
In session-based concurrency, deadlocks are due to circular dependencies inside a session, but are also found in subtle entanglements between  different protocols/sessions. In particular, session type systems \cite{V12} cannot guarantee deadlock freedom of interleaved sessions.
To our knowledge, Dezani et al.~\cite{DLY07} were the first to address progress/deadlock freedom for session-typed processes. 
 Subsequent works on type-based analyses for deadlock freedom in structured communications include~\cite{BCDDDY08,CV09,CD10,CairesP10,P13,DBLP:conf/coordination/VieiraV13,P14}.

As discussed in detail in this paper, another approach to deadlock freedom is based on linear logic under a Curry-Howard correspondence~\cite{CairesP10,DBLP:conf/icfp/Wadler12}, where circular dependencies of communication in processes are eliminated by design, due to Rule~\Did{T-cut}. \nurev{As already mentioned, the type system PCP~\cite{DG18} allows to type safe cyclic process topologies. Following PCP, Kokke and Dardha~\cite{DBLP:conf/forte/KokkeD21} define Priority GV (PGV) and its implementation in Linear Haskell \cite{DBLP:conf/haskell/KokkeD21}, which is a core concurrent $\lambda$-calculus with priorities, whereas the work of Van den Heuvel and P{\'{e}}rez~\cite{DBLP:journals/corr/ICE21} extends PCP with asynchronous communication.} 
\rev{Balzer et al.~\cite{DBLP:conf/esop/BalzerTP19} develop a type system that enforces deadlock freedom for a language with \emph{manifest sharing}, which supports possibly recursive processes in which cycles arise  under an acquire-release discipline for shared channels.}
We note that none of the above works propose a formal comparison between different type systems for deadlock freedom, as we achieve in this paper.

Building upon a translation first suggested by Kobayashi~\cite{K07}, the work of Dardha et al.~\cite{DGS12,DBLP:journals/corr/Dardha14,DGS17} offered a formal relationship between usage types and session types. Such a relationship made it possible to use type-based analysis techniques for usage types in the analysis of session-typed process specifications; this insight was formalized by Carbone et at.~\cite{CDM14,Dardha16}. The effectiveness of the approach in~\cite{CDM14} is supported by informal comparisons with respect to different type systems for deadlock freedom (including those in~\cite{DLY07} and \cite{CD10}) using processes typable in one framework but not in another.

A comparison with the workshop version of our work~\cite{DBLP:journals/corr/DardhaP15} is relevant. In~\cite{DBLP:journals/corr/DardhaP15}, we introduced the notion of \emph{degree of sharing} to classify processes in \fullKoba. Intuitively, the degree of sharing quantifies the greatest number of channels shared by parallel components: processes with a degree of sharing equal to 0 do not share any channels (concurrency without interaction), processes with a degree of sharing equal to 1 share at most one session, and so on. The degree of sharing determines a strict hierarchy of processes (denoted $\Koba{0}, \Koba{1}, \ldots$ in \cite{DBLP:journals/corr/DardhaP15}). This hierarchy, however, is only static---it is not closed under reduction. 

A recent work by Van den Heuvel and P{\'{e}}rez~\cite{DBLP:journals/corr/abs-2004-01320} also compares session type systems based on Curry-Howard foundations. In~\cite{DBLP:journals/corr/abs-2004-01320}, the focus is on the classes of typable processes induced by classical and intuitionistic presentations of linear logic; the main result is that the class induced by intuitionistic linear logic is strictly contained in the class induced by  classical linear logic. That is, classical linear logic leads to more permissive session type systems. \rev{The developments in this paper involved several process languages, each with its own type system---note that the type system that defines \lcp (\Cref{fig:type-system-cll}) falls under the classical perspective}. In contrast, the comparisons in~\cite{DBLP:journals/corr/abs-2004-01320} concern different logically-motivated type systems for the same process language.

Loosely related to our work 
\ifarxiv
(in particular to the translations in \Cref{s:enco} and \Cref{ss:oprw}) 
\else 
(in particular to the translation in \Cref{s:enco}) 
\fi 
is 
previous work on deadlock resolution in the $\pi$-calculus by Giunti and Ravara~\cite{GR13} and on unlocking blocked processes by Francalanza et al.~\cite{DBLP:journals/corr/FrancalanzaGR15}. 
The approach in~\cite{GR13} 
relies on a typing algorithm that detects a particular class of deadlocks (so-called 
self-holding deadlocks), but instead of rejecting the code, fixes it by looking into the session types and producing new safe code that obeys the protocols and is deadlock-free.
Building upon~\cite{GR13}, the work in~\cite{DBLP:journals/corr/FrancalanzaGR15}
investigates 
methods for resolving circular-wait deadlocks across parallel compositions, with a focus on finite CCS processes.


\section{Concluding Remarks}\label{s:concl}

%

We have presented a formal comparison of fundamentally distinct type systems for deadlock-free, 
session typed $\pi$-calculus processes. To the best of our knowledge, ours is the first work to establish precise relationships of this kind. 
Indeed, prior comparisons 
between type systems for deadlock freedom
are informal, given in terms of specific processes typable in one type system but not in some other. 

An immediate difficulty in giving a unified account of different typed frameworks for deadlock freedom is the
variety of 
process languages, type structures, and typing rules that define each framework. 
Our comparisons involve: 
the framework of session processes put forward by Vasconcelos~\cite{V12}; 
the interpretation of linear logic propositions as session types by Caires~\cite{CairesCFest14};
the $\pi$-calculus with usage types by Kobayashi in~\cite{K02}.
Finding some common ground for comparing these three frameworks is not trivial---several translations/transformations
were required in our developments to account for numerous syntactic differences.
We made an effort to follow the exact definitions in each framework.
Overall, we believe that we managed to 
concentrate on essential semantic features of two 
salient classes of deadlock-free session processes, here defined as \lcp and \fullKoba (cf. \Cref{d:lang}).


One main technical contribution is the identification of the precise conditions under which \lcp and \fullKoba coincide. 
We introduced \muKoba, a strict sub-class of \fullKoba that coincides with \lcp. 
The class \muKoba arises as the specialization of Kobayashi's type system~\cite{K06} that results from internalizing the key aspects of the Curry-Howard interpretation of session types, most notably, the principle of ``composition plus hiding'', which allows to compose only processes that share exactly one session. 
This way, \muKoba represents a new characterization of \lcp, given in terms of a very different type system for deadlock freedom adopting usages, capabilities and obligations, put forward by Kobayashi. 
The identification of \muKoba and its coincidence with \lcp  thus provide an immediate way of separating the classes of processes in  \lcp and \fullKoba.
As another technical contribution, but on the opposite direction, we determined how to unify these two classes by developing an intuitive translation of processes in $\fullKoba$ into
processes in $\lcp$, which addresses the syntactic differences between different classes of deadlock-free processes in a type-respecting manner.
Although based on simple ideas,   our technical developments  substantially clarify our  understanding of  type systems for  
liveness properties (such as deadlock freedom) 
in the context of $\pi$-calculus processes.


\paragraph*{Acknowledgments}
We would like to thank the anonymous reviewers of previous versions of this paper for their useful suggestions, which led to substantial improvements.
We are also grateful to Lu\'{i}s Caires and Simon J. Gay
for their valuable comments 
and suggestions on prior versions of this manuscript.
We thank Bas van den Heuvel for his help in designing  \Cref{f:overview}.

This work was partially supported by the EU COST Action IC1201 ({Behavioural Types for Reliable Large-Scale Software Systems}).
 Dardha is supported by the UK EPSRC project
EP/K034413/1
 ({From Data Types to Session Types: A Basis  for Concurrency and Distribution}).
P\'{e}rez has been partially supported by the Dutch Research Council (NWO) under project No. 016.Vidi.189.046 (Unifying Correctness for Communicating Software).
 P\'{e}rez is  also affiliated to the NOVA  Laboratory for Computer Science and Informatics, Universidade Nova de Lisboa, Portugal.

\bibliographystyle{alpha}
\bibliography{my_biblio}

\newpage
\appendix 
\input{appendix}
\ifarxiv
\input{appendix_vp}
\else 
\fi 
\end{document}

%% file: packages.tex
\usepackage[dvipsnames]{xcolor}

\usepackage[utf8]{inputenc}
\usepackage{subfigure}
\usepackage[english]{babel}
\usepackage{amsmath,amsfonts,amssymb}
\usepackage{proof}
\usepackage{graphicx}
\usepackage{multicol}
\usepackage{listings}
\usepackage{url}
\usepackage{cmll}
\usepackage{algpseudocode}
\usepackage{algorithmicx}
\usepackage{enumitem}
\usepackage{textcomp}
\usepackage{latexsym}
\usepackage{epsfig} 
\usepackage{xcolor}
\usepackage{bbm}
\usepackage{etoolbox}

\usepackage{stmaryrd}
\usepackage{xspace}
\usepackage{amsthm}
\usepackage{mdframed}
\mdfdefinestyle{newtight}{innerleftmargin=0pt,innerrightmargin=6pt}
\usepackage{tikz}
\usetikzlibrary{fit,shapes,positioning}

\usepackage{packs/mathpartir}
\usepackage{packs/proof-dashed}
\input{packs/my_macros}

\input{packs/prooftree}

%% file: packs/my_macros.tex
\newif{\ifSHORT}
\newif{\iflong}
\newif{\ifLongVersion}
\newif{\ifWithRecords}
\newif{\ifWithProofs}

\newcommand{\procs}{\ensuremath{\mathbb{P}}\xspace}
\newcommand{\fred}{\hookrightarrow}
\newcommand{\restricted}[1]{[#1]}
\newcommand{\myplus}[1]{\ensuremath{#1^{+}}}

\newcommand{\muKoba}{\ensuremath{\mu\mathcal{K}}\xspace}
\newcommand{\Koba}[1]{\ensuremath{\mathcal{K}_{#1}}\xspace}
\newcommand{\fullKoba}{\ensuremath{\mathcal{K}}\xspace}
\newcommand{\lkoba}{\Koba{}}

\newcommand{\suchthat}{\,:\,}
\newcommand{\rev}[1]{#1}
\newcommand{\nurev}[1]{#1}
\newcommand{\lcp}{\ensuremath{\mathcal{L}}\xspace}

\newcommand{\m}[1]{\mathsf{#1}}

\newcommand{\wt}[1]{\widetilde{#1}}
\newcommand{\tl}[1]{\tilde{#1}}
\newcommand{\defeq}{\triangleq}

\newcommand{\emp}{\emptyset}

\newcommand{\fv}[1]{\fn{#1}}

\newcommand{\compsi}{\diamond}
\newcommand{\circc}{\star}

\def\midd{\; \; \mbox{\Large{$\mid$}}\;\;}

\newcommand{\p}{$\pi$-}
\newcommand{\nil}{{\mathbf{0}}}
\newcommand{\selection}[2]{#1\triangleleft {#2}}

\newcommand{\branching}[3]{#1\triangleright\{{#2}_i:#3_i\}_{i\in I}}
\newcommand{\parbranching}[3]{#1\triangleright\{{#2}:#3\}_{i\in I}}

\newcommand{\out}[1]{\langle #1\rangle}
\newcommand{\bout}[2]{\overline{#1}(#2)}
\newcommand{\inp}[1]{(#1)}
\newcommand{\res}[1]{({\boldsymbol \nu} #1)}
\newcommand{\pp}{\ \boldsymbol{|}\ }
\newcommand{\picase}[3]
{\mathbf{case} \, {#1} \, \mathbf{of}\, \{ \mathnormal{l}_i \_ {#2} \triangleright {#3} \}_{i\in I}}

\newcommand{\vv}[2]{\mathnormal{l}_{#1}\_ #2}

\newcommand{\uptok}{\doteqdot}
\newcommand{\catdep}[3]{\mathcal{#1}_{#2}^{#3}}

\newcommand{\core}[1]{#1^{\downarrow}}
\newcommand{\eend}[1]{\ensuremath{\mathsf{un}(#1)}}
\newcommand{\eende}{\ensuremath{\mathsf{un}}}
\newcommand{\variant}[2]{\langle {#1} : #2 \rangle_{i\in I}}

\newcommand{\select}[2]{\oplus\{{#1}_i:#2_i\}_{i\in I}}
\newcommand{\branch}[2]{\&\{{#1}_i:#2_i\}_{i\in I}}

\newcommand{\parselect}[3]{\oplus\{{#1} :#2 \}_{#3}}
\newcommand{\parbranch}[3]{\&\{{#1} :#2 \}_{#3}}
\newcommand{\nilT}{\ensuremath{{\bf end}}}

\newcommand{\unit}{{\mathbf{n}}}

\newcommand{\ct}{\ell}

\newcommand{\su}{\mathsf{u}}
\newcommand{\inpuse}[2]{\wn^{#1}_{#2}}
\newcommand{\outuse}[2]{\oc^{#1}_{#2}}
\newcommand{\ca}{\kappa}
\newcommand{\ob}{\mathsf{o}}
\newcommand{\obss}[2]{\mathsf{ob}_{#1}(#2)}
\newcommand{\caps}[2]{\mathsf{cap}_{#1}(#2)}
\newcommand{\conpar}[2]{\mathsf{con}_{#1}(#2)}
\newcommand{\rel}[1]{\mathsf{rel}( #1 )}
\newcommand{\usa}[1]{\mathit{u}( #1 )}
\newcommand{\con}[1]{\mathsf{con}(#1)}
\newcommand{\proj}[1]{\downarrow_{#1}}
\newcommand{\decomp}[1]{\asymp_{#1}}

\newcommand{\utype}[2]{\mathsf{chan}(#2\,;\,#1)}
\newcommand{\ctype}{\tau}

\newcommand{\Didtext}[1]{{\scriptsize{\textsc{#1}}}}
\newcommand{\Did}[1]{({\Didtext{#1}})}
\newcommand{\lf}{\vdash_{\prec}}
\newcommand{\mlf}{\vdash_{\prec}^\mu}
\newcommand{\mlfpar}[1]{\vdash^{\mu}_{#1}} 

\newcommand{\lfpar}[1]{\vdash_{#1}} 

\newcommand{\s}{\vdash_{\mathtt{ST}}}
\newcommand{\cp}{\vdash_{\mathtt{LL}}}

\newcommand{\zusage}{\mathsf{0}}
\newcommand{\temp}{-}
\newcommand{\fn}[1]{\mathsf{fn}(#1)}   
\newcommand{\ctx}[3]{\mathcal{#1}_{#2}\big[#3\big]}

\newcommand{\semi}{\ {\bf ; }_{\prec}\ }
\newcommand{\mto}{\to^*}

\newcommand{\dom}{\m{dom}}

\newcommand{\chrp}[2]{\ensuremath{\langle\!| #1 |\!\rangle^{#2}}}
\newcommand{\chrpk}[2]{\ensuremath{\langle\!| #1 |\!\rangle^{#2}}}
\newcommand{\chr}[1]{\llbracket #1 \rrbracket_{\ell}}

\newcommand{\encCP}[1]{\llparenthesis #1 \rrparenthesis}
\newcommand{\encCPalt}[1]{\llparenthesis #1 \rrparenthesis^{\mathtt{V}}}
\newcommand{\encoCP}[2]{\llparenthesis #1 \rrparenthesis^{#2}}
\newcommand{\f}[1]{{\mathnormal f}_{#1}}

\newcommand{\enc}[2]{\llbracket #1 \rrbracket^{\mathnormal{f}, \{x,y \mapsto {#2}\}}_{\su}}

\newcommand{\encob}[2]{\llbracket #1 \rrbracket_{#2}}
\newcommand{\encf}[1]{\llbracket #1 \rrbracket^{\mathnormal{f}}_{\su}}
\newcommand{\encfp}[2]{\llbracket #1 \rrbracket^{\mathnormal{#2}}_{\su}}
\newcommand{\encx}[2]{\llbracket #1 \rrbracket^{\mathnormal{f},\{x\mapsto {#2}\}}_{\su}}
\newcommand{\ency}[2]{\llbracket #1 \rrbracket^{\mathnormal{f},\{y\mapsto {#2}\}}_{\su}}

\newcommand{\vdepset}[1]{\llparenthesis #1 \rrparenthesis}
\newcommand{\vdep}[2]{(#1,#2)}

\newcommand{\chrpv}[3]{\ensuremath{\langle\!| #1 |\!\rangle^{#2}_{#3}}}

\newcommand{\mydefref}[1]{Definition~\ref{#1}}

\newcommand{\lolli}{\mathord{\multimap}}
\newcommand{\tensor}{\otimes}
\newcommand{\parl}{\mathbin{\bindnasrepma}}
\newcommand{\one}{{\bf 1}}
\newcommand{\CLL}{\texttt{CLL}\xspace}
\newcommand{\dual}[1]{\overline{#1}}

\newcommand{\zero}{{\bf 0}}
\newcommand{\nub}{{\boldsymbol{\nu}}}
\newcommand{\linkr}[2]{[#1\!\leftrightarrow\!#2]}
\newcommand{\jname}[1]{\mbox{{{(#1)}}}}
\newcommand{\cut}{\mathsf{cut}}
\newcommand{\mix}{\mathsf{mix}}

\newcommand{\cycle}{\mathsf{cycle}}
\newcommand{\tid}{\mathsf{id}}
\newcommand{\ov}[1]{\overline{#1}}
\newcommand{\para}{\mathord{\;\boldsymbol{|}\;}}
\def\substj#1#2{[\raisebox{.5ex}{\small$#1$}\! / \mbox{\small$#2$}]}

\newcommand{\D}{\Delta}

\newcommand{\tra}[1]{\xrightarrow{#1}}
\newcommand{\redd}{\tra{~~~}}

%% file: packs/prooftree.tex
\newdimen\proofrulebreadth \proofrulebreadth=.05em
\newdimen\proofdotseparation \proofdotseparation=1.25ex
\newdimen\proofrulebaseline \proofrulebaseline=2ex
\newcount\proofdotnumber \proofdotnumber=3
\let\then\relax
\def\hfi{\hskip0pt plus.0001fil}
\mathchardef\squigto="3A3B
\newif\ifinsideprooftree\insideprooftreefalse
\newif\ifonleftofproofrule\onleftofproofrulefalse
\newif\ifproofdots\proofdotsfalse
\newif\ifdoubleproof\doubleprooffalse
\let\wereinproofbit\relax
\newdimen\shortenproofleft
\newdimen\shortenproofright
\newdimen\proofbelowshift
\newbox\proofabove
\newbox\proofbelow
\newbox\proofrulename
\def\shiftproofbelow{\let\next\relax\afterassignment\setshiftproofbelow\dimen0 }
\def\shiftproofbelowneg{\def\next{\multiply\dimen0 by-1 }%
\afterassignment\setshiftproofbelow\dimen0 }
\def\setshiftproofbelow{\next\proofbelowshift=\dimen0 }
\def\setproofrulebreadth{\proofrulebreadth}

\def\prooftree{
\ifnum  \lastpenalty=1
\then   \unpenalty
\else   \onleftofproofrulefalse
\fi
\ifonleftofproofrule
\else   \ifinsideprooftree
        \then   \hskip.5em plus1fil
        \fi
\fi
\bgroup
\setbox\proofbelow=\hbox{}\setbox\proofrulename=\hbox{}%
\let\justifies\proofover\let\leadsto\proofoverdots\let\Justifies\proofoverdbl
\let\using\proofusing\let\[\prooftree
\ifinsideprooftree\let\]\endprooftree\fi
\proofdotsfalse\doubleprooffalse
\let\thickness\setproofrulebreadth
\let\shiftright\shiftproofbelow \let\shift\shiftproofbelow
\let\shiftleft\shiftproofbelowneg
\let\ifwasinsideprooftree\ifinsideprooftree
\insideprooftreetrue
\setbox\proofabove=\hbox\bgroup$\displaystyle 
\let\wereinproofbit\prooftree
\shortenproofleft=0pt \shortenproofright=0pt \proofbelowshift=0pt
\onleftofproofruletrue\penalty1
}

\def\eproofbit{
\ifx    \wereinproofbit\prooftree
\then   \ifcase \lastpenalty
        \then   \shortenproofright=0pt  
        \or     \unpenalty\hfil         
        \or     \unpenalty\unskip       
        \else   \shortenproofright=0pt  
        \fi
\fi
\global\dimen0=\shortenproofleft
\global\dimen1=\shortenproofright
\global\dimen2=\proofrulebreadth
\global\dimen3=\proofbelowshift
\global\dimen4=\proofdotseparation
\global\count255=\proofdotnumber
$\egroup  
\shortenproofleft=\dimen0
\shortenproofright=\dimen1
\proofrulebreadth=\dimen2
\proofbelowshift=\dimen3
\proofdotseparation=\dimen4
\proofdotnumber=\count255
}

\def\proofover{
\eproofbit 
\setbox\proofbelow=\hbox\bgroup 
\let\wereinproofbit\proofover
$\displaystyle
}%
\def\proofoverdbl{
\eproofbit 
\doubleprooftrue
\setbox\proofbelow=\hbox\bgroup 
\let\wereinproofbit\proofoverdbl
$\displaystyle
}%
\def\proofoverdots{
\eproofbit 
\proofdotstrue
\setbox\proofbelow=\hbox\bgroup 
\let\wereinproofbit\proofoverdots
$\displaystyle
}%
\def\proofusing{
\eproofbit 
\setbox\proofrulename=\hbox\bgroup 
\let\wereinproofbit\proofusing
\kern0.3em$
}

\def\endprooftree{
\eproofbit 
  \dimen5 =0pt
\dimen0=\wd\proofabove \advance\dimen0-\shortenproofleft
\advance\dimen0-\shortenproofright
\dimen1=.5\dimen0 \advance\dimen1-.5\wd\proofbelow
\dimen4=\dimen1
\advance\dimen1\proofbelowshift \advance\dimen4-\proofbelowshift
\ifdim  \dimen1<0pt
\then   \advance\shortenproofleft\dimen1
        \advance\dimen0-\dimen1
        \dimen1=0pt
        \ifdim  \shortenproofleft<0pt
        \then   \setbox\proofabove=\hbox{%
                        \kern-\shortenproofleft\unhbox\proofabove}%
                \shortenproofleft=0pt
        \fi
\fi
\ifdim  \dimen4<0pt
\then   \advance\shortenproofright\dimen4
        \advance\dimen0-\dimen4
        \dimen4=0pt
\fi
\ifdim  \shortenproofright<\wd\proofrulename
\then   \shortenproofright=\wd\proofrulename
\fi
\dimen2=\shortenproofleft \advance\dimen2 by\dimen1
\dimen3=\shortenproofright\advance\dimen3 by\dimen4
\ifproofdots
\then
        \dimen6=\shortenproofleft \advance\dimen6 .5\dimen0
        \setbox1=\vbox to\proofdotseparation{\vss\hbox{$\cdot$}\vss}%
        \setbox0=\hbox{%
                \advance\dimen6-.5\wd1
                \kern\dimen6
                $\vcenter to\proofdotnumber\proofdotseparation
                        {\leaders\box1\vfill}$%
                \unhbox\proofrulename}%
\else   \dimen6=\fontdimen22\the\textfont2 
        \dimen7=\dimen6
        \advance\dimen6by.5\proofrulebreadth
        \advance\dimen7by-.5\proofrulebreadth
        \setbox0=\hbox{%
                \kern\shortenproofleft
                \ifdoubleproof
                \then   \hbox to\dimen0{%
                        $\mathsurround0pt\mathord=\mkern-6mu%
                        \cleaders\hbox{$\mkern-2mu=\mkern-2mu$}\hfill
                        \mkern-6mu\mathord=$}%
                \else   \vrule height\dimen6 depth-\dimen7 width\dimen0
                \fi
                \unhbox\proofrulename}%
        \ht0=\dimen6 \dp0=-\dimen7
\fi
\let\doll\relax
\ifwasinsideprooftree
\then   \let\VBOX\vbox
\else   \ifmmode\else$\let\doll=$\fi
        \let\VBOX\vcenter
\fi
\VBOX   {\baselineskip\proofrulebaseline \lineskip.2ex
        \expandafter\lineskiplimit\ifproofdots0ex\else-0.6ex\fi
        \hbox   spread\dimen5   {\hfi\unhbox\proofabove\hfi}%
        \hbox{\box0}%
        \hbox   {\kern\dimen2 \box\proofbelow}}\doll%
\global\dimen2=\dimen2
\global\dimen3=\dimen3
\egroup 
\ifonleftofproofrule
\then   \shortenproofleft=\dimen2
\fi
\shortenproofright=\dimen3
\onleftofproofrulefalse
\ifinsideprooftree
\then   \hskip.5em plus 1fil \penalty2
\fi
}


%% file: figures/diagram.tex
\begin{tikzpicture}
    \begin{scope}[local bounding box=BBu]
        \node (u) [] at (-0.125, 0) {
            $\chr{P}   \cp  \encob{\Gamma}{\ct}$
        };
        \node (thm32) [below=.5pt of u] {{\small DF: \Cref{t:progress}}};
    \end{scope}
    \node [fit=(BBu), thick, draw=RoyalBlue, shape=ellipse, inner sep=0pt, label={above right:$\textcolor{RoyalBlue}{\lcp}$}] (BBuFit) {};

    \node (st) [semithick, draw=black, shape=ellipse, label={15:$\mathbb{P}$}] at (2.25, 1.5) {$\Gamma \s P$ {\small \cite{V12}}};

    \begin{scope}[local bounding box=BBltm]
        \node (lt) [right=95pt of thm32] {
            $\encf{\Gamma}  \mlf \encf{P} $
        };
        \begin{scope}[overlay]
            \node [fit=(lt), semithick, draw=RedOrange, shape=ellipse, inner sep=0pt, label={[xshift=4pt]190:$\textcolor{RedOrange}{\muKoba}$}] (BBltFit) {};
        \end{scope}
        \node [fit=(lt), shape=ellipse, inner sep=0pt] (BBltStretch) {};

        \node (ltm) [below=6pt of lt] {
            $  \encf{\Gamma}    \lf \encf{P}$
        };
        \node (cor32) [below=.5pt of ltm] {{\small DF: \Cref{c:df}}};
    \end{scope}
    \node [fit=(BBltm), thick, draw=RedOrange, shape=ellipse, inner sep=-1pt, label={above left:$\textcolor{RedOrange}{\fullKoba}$}] (BBltmFit) {};

    \draw [stealth reversed-stealth reversed, draw=ForestGreen, thick] (thm32.east)+(2pt,0) -- node [below, at start, xshift=29pt, yshift=2pt] {
        $\begin{array}{c} \text{{\small{Thm.~\ref{t:cppkoba}}}} \end{array}$
    } (lt);
    
    \draw [-to, in=285, out=215, draw=ForestGreen, very thick, densely dashed] (ltm.west)+(-15pt,-10pt) to node [below left] {$\begin{array}{c}\text{\Cref{def:typed_enc}:} \\ \text{Translation $\encCP{\cdot}{}$}\end{array}$} (BBuFit);
    
    \draw [-to, in=90, out=200] (st.west) to node [above left, xshift=-2pt, yshift=-10pt] {
        $\begin{array}{l} \text{{\small \Cref{f:sesstocp}:~$\chr{\cdot}$ (processes)}} \\[0.5mm] \text{{\small \Cref{def:enc_env_sc}:~$\encob{\cdot}{\ct}$ (contexts)}}\end{array}$
    }(BBuFit);
    \draw [-to, out=0, in=105,>=stealth] (st.east) to node [above right, xshift=2pt, yshift=-12pt] {
        $\begin{array}{l} \text{{\small \Cref{d:encdgs}:~$\encf{\cdot}$ (processes)}}
        \\[0.5mm] \text{{\small \Cref{def:enc_env_su}:~$\encf{\cdot}$~(contexts)}} \end{array}$
    }(BBltmFit);
\end{tikzpicture}

%% file: appendix.tex
\section{Omitted Proofs for \Cref{s:hier}}
\label{sec:appendix-4}

\subsection{Proof of \Cref{t:cppkoba} (Page~\pageref{t:cppkoba})}\label{appt:cppkoba}

We divide the proof into the following two lemmas: \Cref{lem: aux1} (see below) and \Cref{lem: aux2} (Page~\pageref{lem: aux2}).

\begin{lemma}[Substitution Lemma for CP]\label{lem:subst_cp}
If $P \cp \Gamma, x:T$ and $v \notin \fv P$, then
$P[v/x] \cp \Gamma, v:T$.
\end{lemma}

\begin{lemma}
\label{lem: aux1}
If $P \in \lcp$ then $P \in \muKoba$.
\end{lemma}

\begin{proof}
By structural induction on $P$.
By 
\Cref{d:lang} and \Cref{d:langmu}, we have that 
\begin{align*}
\lcp  & =    \big\{P \in \procs \suchthat \exists \Gamma.\ (\Gamma \s P \,\land\, \chr{P} \cp \encob{\Gamma}{\ct}) \big\} \\
\muKoba & =  \big\{  P \in \procs \suchthat \exists \Gamma,f.\ (\Gamma \s P \,\land\, \encf{\core{\Gamma}}  \mlf \encf P) \big\} 
\end{align*}
where $\chr{\cdot}$ is as in \Cref{d:trans}.
Following the syntax in \Cref{fig:sessionpi}, 
there are seven cases to consider:

\begin{enumerate}

\item $P = \nil$: 
By assumption we have both
$\Gamma \s \nil$, for some context $\Gamma$ (with $\eend{\Gamma}$), 
and 
$\nil \cp x_1{:}\bullet, \cdots, x_n{:}\bullet$, for some $x_1, \ldots, x_n$ such that $x_i:\nilT  \in \Gamma$, for all $i \in \{1, \ldots, n\}$. 
Because $\core{\Gamma} = \emp$, we must show that 
$\encf{\emp} \mlf \encf{\nil}$.
By \Cref{def:enc_env_su} and~\Cref{d:encdgs}, this is the same as showing that 
$\emp \mlf \nil$.
The thesis then follows immediately from Rule~{\Did{T$\pi$-Nil}} in \Cref{f:mupityping}.

\smallskip

\item $P =  x\inp y.P'$:
By assumption 
and \Cref{def:enc_env_su},~\Cref{d:trans}, and~\Cref{d:lang}, 
we have 
\begin{eqnarray}
\Gamma, x : \wn T.S & \s &  x\inp{y}.P' \label{d:mueq0} \\
x\inp{y}.\chr{P'} & \cp & \encob{\Gamma}{\ct}, x : \encob{T}{\ct}  \parl  \encob{S}{\ct} \nonumber
\end{eqnarray}
for some context $\Gamma$  and session types $S, T$.
By inversion on typing judgement on \eqref{d:mueq0}:
\begin{equation}
\label{d:mueqj0}
    \begin{prooftree}
      \Gamma, x:S, y:T \s P'
      \justifies
      \Gamma, x : \wn T.S \s x\inp{y}.P'
    \end{prooftree}
\end{equation}

By \Cref{def:enc_env_su} and~\Cref{d:langmu} and by \Cref{d:encdgs}, we then must show: 
$$ f_x:\utype{\inpuse{0}{0}}{\encob{T}{\su}, \encob{S}{\su}}  \semi \encf{\core{\Gamma}}  \mlf f_x\inp{y,c}.\encfp{P'}{f,\{x\mapsto c\}}$$
By induction hypothesis on the premise of \eqref{d:mueqj0} we have:
\begin{equation}
\label{eq:ip}
\encfp{\core{\Gamma}}{f'}, f'_x:\encob{S}{\su}, f'_y:\encob{T}{\su} \mlf \encfp{P'}{f'}
\end{equation}
for some renaming function $f'$. We let $f$ be such that $f' = f, \{x\mapsto c\}$.
We can then use $f$ to rewrite the above judgement \eqref{eq:ip}, as follows:
\begin{equation}
\encf{\core{\Gamma}}, c:\encob{S}{\su}, y:\encob{T}{\su} \mlf \encfp{P'}{f,\{x \mapsto c\}} \label{eq:mu1}
\end{equation}
Then, the thesis follows from 
\eqref{eq:mu1}
after applying Rule~{\Did{T$\pi$-In}} in \Cref{f:mupityping}.

\smallskip

\item $P = \overline x\out y.P'$: 
By assumption 
and \Cref{def:enc_env_su},~\Cref{d:trans}, and~\Cref{d:lang}, 
we have: 
\begin{eqnarray}
\Gamma,  x : \oc T.S, y:T & \s & \overline x\out{y}.P' \label{eq:mu2}\\
\bout{x}{z}.(\linkr{z}{y} \pp \chr{P'})  & \cp & \encob{\Gamma}{\ct}, x : \encob{\dual{T}}{\ct} \otimes \encob{S}{\ct}, y : \encob{T}{\ct}\nonumber
\end{eqnarray}
for some context $\Gamma$  and session types $S, T$.
By inversion on typing judgement on \eqref{eq:mu2} we have:
\begin{equation}
\label{eqj:mu1}
    \begin{prooftree}
      \Gamma, x:S \s P'
      \justifies
      \Gamma,  x : \oc T.S, y:T \s \overline x\out{y}.P'
    \end{prooftree}
\end{equation}
By \Cref{d:langmu} and \Cref{d:encdgs}, we must show: 
$$
{
{\f x}: \utype{\outuse{0}{0}}{\encob{T}{\su},  \, \encob{\ov S}{\su}}
 \semi 
\big(y:\encob{T}{\su}\ \pp \ \encf{\core{\Gamma}}\big)\ \mlf
\res c\ov{\f x}\out {y,c}.\encfp{P'}{f,\{x \mapsto c\}}
}
$$
By induction hypothesis on the premise of \eqref{eqj:mu1} we have:
\begin{equation*}
\encfp{\core{\Gamma}}{f'}, f'_x:\encob{S}{\su} \  \mlf \encfp{P'}{f'}
\end{equation*}
for some renaming function $f'$. We let $f$ be such that $f' = f, \{x\mapsto c\}$.
We can then rewrite the above judgement in terms of $f$ as follows:
\begin{equation}
\encf{\core{\Gamma}}, c:\encob{S}{\su} \  \mlf \encfp{P'}{f,\{x\mapsto c\}} \label{eq:mu3}
\end{equation}
By applying Rule~{\Did{T\p Var}} 
we can derive both
$$y:\encob{T}{\su}\ \mlf y:\encob{T}{\su}\ (*) \quad\text{and}\quad
c:\encob{\ov S}{\su}\ \mlf c:\encob{\ov S}{\su}\ (**)$$ 
Let $U_1 = \usa{\encob{\ov S}{\su}}$ and $U_2 = \usa{\encob{S}{\su}}$.
\Cref{p:relia} ensures $\rel {U_1 \pp  U_2}$. 
We can then apply Rule~{\Did{T\p BOut}} in \Cref{f:mupityping} on $(*)$, $(**)$, and the induction hypothesis in \eqref{eq:mu3}:
\begin{equation}
\inferrule{
y:\encob{T}{\su}, c:\encob{\ov S}{\su} 
\mlf 
y:\encob{T}{\su}, c:\encob{\ov S}{\su}
\and
\encf{\core{\Gamma}}, c:\encob{S}{\su}   \mlf \encfp{P'}{f,\{x \mapsto c\}}
\and
\rel {U_1 \pp  U_2}
}{
\f x: \utype{\outuse{0}{0}}{\encob{T}{\su},  \, \encob{\ov S}{\su}}
 \semi 
\big(y:\encob{T}{\su}\ \pp  \encf{\core{\Gamma}}  \big) 
\mlf
\res{c} \ov{\f x}\out {y,c}.\encfp{P'}{f,\{x \mapsto c\}}
}
\label{eq:mu4}
\end{equation}
which concludes this case.

\smallskip

\item
$P=\branching xlP$.
Then,
by assumption 
and \Cref{def:enc_env_su},~\Cref{d:trans}, and~\Cref{d:lang}
we have: 
\begin{eqnarray}
\Gamma, x:\branch lT &\s &\branching {x}lP \label{eq:mbra}\\
\parbranching x{l_i}{P_i} &\cp &\encob{\Gamma}{\ct}, x{:}\parbranch {l_i}{\encob{T_i}{\ct}}{i\in I}\nonumber
\end{eqnarray}
for some context $\Gamma$  and session types $T_i$ for $i\in I$.
By inversion on the typing judgement given in \eqref{eq:mbra} we have:
\begin{equation}
\label{eqj:m3}
\inferrule
	{
		\Gamma, x: T_i \s P_i \quad  \forall i\in I
	}
	{
		\Gamma, x:\branch lT \s \branching {x}lP
	}
\end{equation}
By \Cref{def:enc_env_su} and~\Cref{d:langmu}, and by \Cref{d:encdgs}, we must show:
\[
\f x: \utype{\inpuse{0}{0}}{\variant {l_i}{\encob{T_i}{\su}}} \semi  \encf{\core{\Gamma}}
\mlf   
\f x \inp y.\ \picase {y}{c}{\encfp{P_i}{f,\{x \mapsto c\}}}
\]
By induction hypothesis on the premise of \eqref{eqj:m3} we have:
\[
\encfp{\core{\Gamma}}{f'}, f'_x:\encob{T_i}{\su} 
\mlf 
\encfp{P_i}{f'}\quad \forall i\in I
\]
for some renaming function $f'$. We let $f$ be such that $f' = f, \{x\mapsto c\}$.
We can then rewrite the above judgement as follows:
\begin{equation*}
\encf{\core{\Gamma}}, c:\encob{T_i}{\su} 
\mlf 
\encfp{P_i}{f,\{x \mapsto c\}} \quad \forall i\in I\label{eq:menc_bra}
\end{equation*}
By applying Rule~{\Did{T\p Var}}, we obtain 
$y:\variant {l_i}{\encob{T_i}{\su}} \mlf y:\variant {l_i}{\encob{T_i}{\su}}$.
Then we apply in order Rules~{\Did{T\p Case}}
and {\Did{T\p Inp}} (cf.~\Cref{f:mupityping}) on the judgement above to conclude:
\[
\f x: \utype{\inpuse{0}{0}}{\variant {l_i}{\encob{T_i}{\su}}} \semi  \encf{\core{\Gamma}}
\mlf   
\f x \inp y.\ \picase {y}{c}{\encfp{P_i}{f,\{x \mapsto c\}}}
\]

\smallskip
\item
$P= \selection x{l_j}.{P_j}$.
Then,
by assumption 
and \Cref{def:enc_env_su},~\Cref{d:trans}, and~\Cref{d:lang}, 
we have: 
\begin{eqnarray}
\Gamma, x:  \select lT &\s& \selection {x} {l_j}.{P_j}  \label{eq:msel} \\
\selection {x} {l_j}.{\chr {P_j}}  &\cp& \encob{\Gamma}{\ct}, x{:}\parselect {l_i}{\encob{T_i}{\ct}}{i\in I} \nonumber
\end{eqnarray}
for some context $\Gamma$  and session types $T_i$, for $i\in I$.
By inversion on typing judgement on~\eqref{eq:msel} we have:
\begin{equation}
\label{eqj:m4}
\inferrule
	{
		\Gamma,  x: T_j \s {P_j} \quad   \exists j \in I
	}
	{ 	
		\Gamma, x:  \select lT \s \selection {x} {l_j}.{P_j}
	}
\end{equation}
By \Cref{d:langmu} and \Cref{d:encdgs}, we must show:
\[
{
\f x: \utype{\outuse{0}{0}}{\variant {l_i}{\encob{\dual {T_i}}{\su}}}
\semi
\encf{\core{\Gamma}}
\mlf \res c \ov{\f x}\out {\vv{j}c}.\encfp{P_j}{f, \{x \mapsto c\}}
}
\]

By induction hypothesis on the premise of \eqref{eqj:m4} we have:
\[
\encfp{\core{\Gamma}}{f'}, f'_x:\encob{T_j}{\su} 
\mlf \encfp{P_j}{f'}
\]
for some renaming function $f'$. We let $f$ be such that $f' = f, \{x\mapsto c\}$.
We can then rewrite the above judgement as follows:
\begin{equation}
\encf{\core{\Gamma}}, c: \encob{T_j}\su 
\mlf \encfp{P_j}{f, \{x \mapsto c\}}
\label{eq:msel_enc}
\end{equation}
By applying Rules {\Did{T\p Var}} and {\Did{T\p LVal}} we obtain:
\begin{equation}
c:\encob{\dual{T_j}}{\su} 
\mlf    {l_j}\_c: \variant {l_i}{\encob{\dual{T_i}}{\su}}
\label{eq:mlval}
\end{equation}

Let $U_1 = \usa{\encob{\ov{T_j}}{\su}}$ and $U_2 = \usa{\encob{T_j}{\su}}$
\Cref{p:relia} ensures $\rel {U_1 \pp  U_2}$. 
We can then apply Rule~{\Did{T\p BOut}} (cf. \Cref{f:mupityping}), without a free channel (only a bound one),
on \eqref{eq:msel_enc} and~\eqref{eq:mlval}. We have:
\[
\inferrule{
c:\encob{\dual{T_j}}{\su} \mlf  {l_j}\_c: \variant {l_i}{\encob{\dual{T_i}}\su}
\and
\encf{\core{\Gamma}}, c: \encob{T_j}\su 
\mlf  \encfp{P_j}{f, \{x \mapsto c\}}
\and
\rel {U_1 \pp  U_2}
}{
\f x: \utype{\outuse{0}{0}}{\variant {l_i}{\encob{\dual {T_i}}\su}}
\semi
\encf{\core{\Gamma}}
\mlf \res c \ov{\f x}\out {\vv{j}c}.\encfp{P_j}{f, \{x \mapsto c\}}
}
\]
which concludes this case.

\smallskip

\item
$P = \res {xy}P'$: 
Then, by assumption 
and \Cref{def:enc_env_su},~\ref{d:trans}, and~\ref{d:lang}, 
we have that there exist $P_1, P_2$ such that 
$P' = P_1 \pp P_2$ with
\begin{eqnarray}
\Gamma  & \s & \res {xy}(P_1 \pp P_2) \label{eq:m4a} \\ 
\res{w}(\chr{P_1}\substj{w}{x} \pp \chr{P_2}\substj{w}{y})  & \cp & \encob{\Gamma}{\ct} \nonumber
\end{eqnarray}
for some context $\Gamma$.
By  inversion on the typing judgements given in \eqref{eq:m4a} we have the following derivation,
for some session type $T$:
\begin{equation}
\label{eqj:m5}
    \begin{prooftree}
    \begin{prooftree}
      \Gamma_1,  x:T \s  P_1 
      \quad 
       \Gamma_2,  y: \dual{T} \s  P_2 
      \justifies
      \Gamma,  x:T ,   y: \dual{T} \s   P_1 \pp P_2
    \end{prooftree}
      \justifies
      \Gamma  \s  \res {xy}(P_1 \pp P_2)
    \end{prooftree}
\end{equation}
where $\Gamma = \Gamma_1, \Gamma_2$.
Because $P\in \lcp$, we have that $\Gamma_1$ and $\Gamma_2$ are disjoint and 
that channel endpoints $x$ and $y$ (or the single name $w$ in $\lcp$) form the only shared channel between processes $P_1$ and $P_2$.
By \Cref{d:langmu} and \Cref{d:encdgs}, we must show: 
\[
{
\encf{\core{\Gamma}}\ 
\mlf \res w ( \encfp{P_1}{f, \{x \mapsto w\}} \pp  \encfp{P_2}{f, \{y \mapsto w\}} )
}
\]
By induction hypothesis on the premises of \eqref{eqj:m5}, we have:
\begin{equation*}
\encfp{\core{\Gamma_1}}{f'}, f'_x:\encob{T}{\su} \  \mlfpar{\prec_1} \encfp{P_1}{f'}
\quad\text{and}\quad
\encfp{\core{\Gamma_2}}{f'}, f'_y:\encob{\dual{T}}{\su} \  \mlfpar{\prec_2} \encfp{P_2}{f'}
\end{equation*}
for some renaming function $f'$. We let $f$ be such that $f' = f,  \{x,y\mapsto w\}$, hence we have $f'_x=f'_y=w$.
We can now rewrite the above judgements in terms of $f$ as follows:
\begin{equation*}
\encf{\core{\Gamma_1}}, w:\encob{T}{\su} \  \mlfpar{\prec_1} \encfp{P_1}{f,\{x \mapsto w\}}
\quad\text{and}\quad
\encf{\core{\Gamma_2}}, w:\encob{\dual{T}}{\su} \  \mlfpar{\prec_2} \encfp{P_2}{f,\{y \mapsto w\}} 
\end{equation*}
Let $U_1 = \usa{\encob{{T}}{\su}}$ and $U_2 = \usa{\encob{\ov{T}}{\su}}$.
By \Cref{p:relia}, $\rel {U_1 \pp  U_2}$. 
Combining these facts, we can apply Rule~{\Did{T\p Par+Res}} in \Cref{f:mupityping}:
\[
\inferrule{
\encf{\core{\Gamma_1}}, w:\encob{T}{\su} \  \mlfpar{\prec_1} \encfp{P_1}{f,\{x \mapsto w\}}
\and
\encf{\core{\Gamma_2}}, w:\encob{\dual{T}}{\su} \  \mlfpar{\prec_2} \encfp{P_2}{f,\{y \mapsto w\}} 
\and 
\rel {U_1 \pp  U_2}
}
{
\encf{\core{\Gamma_1}}, \encf{\core{\Gamma_2}}\mlf \res w (\encfp{P_1}{f,\{x \mapsto w\}}\pp  \encfp{P_2}{f,\{y \mapsto w\}} )
}
\]
with
$\prec_i = \prec \cup \{(w,y) \,|\, y \in \fv{P_i}\setminus \{w\} \}$ ($i \in \{1,2\}$. This completes the proof for this case.

\smallskip

\item $P = P_1 \pp P_2$: 
This case is similar to the previous one. By assumption 
we have
\begin{eqnarray}
\Gamma  & \s &   P_1 \pp P_2 \label{meq:6} \\
\chr{P_1}\pp \chr{P_2}  & \cp &  \encob{\Gamma}{\ct}\nonumber
\end{eqnarray}
and inversion on typing on \eqref{meq:6} we infer:
\begin{equation}
\label{meq:7}
    \begin{prooftree}
      \Gamma_1 \s  P_1 
      \quad 
       \Gamma_2  \s  P_2 
      \justifies
      \Gamma_1 ,  \Gamma_2 \s   P_1 \pp P_2
    \end{prooftree}
\end{equation}
where $\Gamma = \Gamma_1 ,  \Gamma_2$.
As before, because $P\in \lcp$, we have that $\Gamma_1$ and $\Gamma_2$ are disjoint: they share no endpoints. 
By \Cref{d:langmu} and \Cref{d:encdgs}, we must show: 
$$  \encf{\core{\Gamma}}   \mlf   \encf {P_1} \pp  \encf {P_2}$$
which follows by induction hypothesis on the premises of \eqref{meq:7}
and by applying Rule~{\Did{\textsc{T$\pi$-IndPar}}} (cf.~\Cref{f:mupityping}).
\end{enumerate}
\end{proof}

\begin{lemma}
\label{lem: aux2}
If $P \in \muKoba$ then $P \in \lcp$.
\end{lemma}

\begin{proof}[Proof]
By structural induction on $P$. 
Recall that by 
\Cref{d:langmu} and \Cref{d:lang}, we have that 
\begin{align*}
\muKoba & \defeq  \Big\{  P \in \procs \suchthat \exists \Gamma,f.\ (\Gamma \s P \,\land\, \encf{\core{\Gamma}} \mlf \encf P) \Big\} 
\\
\lcp  & \defeq     \big\{P \in \procs \suchthat \exists \Gamma.\ (\Gamma \s P \,\land\, \chr{P} \cp \encob{\Gamma}{\ct}) \big\}
\end{align*}

\noindent
Following the syntax of processes in \Cref{fig:sessionpi}, 
there are seven cases to consider:
\begin{enumerate}
\item $P = \nil$: 
By assumption we have both
$\Gamma \s \nil$, for some   $\Gamma$ (with $\eend{\Gamma}$), 
and 
$\encf{\core{\Gamma}}  \mlf \encf \nil$.
Notice that $\core{\Gamma} = \emp$. 
We have to show
$\nil \cp x_1{:}\bullet, \cdots, x_n{:}\bullet$, for some $x_1, \ldots, x_n$ such that $x_i:\nilT  \in \Gamma$, for all $i \in \{1, \ldots, n\}$. 
The thesis follows by using Axiom~\Did{T-$\one$}, followed by $n-1$ applications of Rule~\Did{T-$\bot$} in \Cref{fig:type-system-cll}.

\smallskip

\item $P =  x\inp y.P'$:  
Then, by assumption, \Cref{d:encdgs} and \Cref{d:langmu}, 
we have both
\begin{eqnarray}
\Gamma, x : \wn T.S & \s & x\inp{y}.P' \label{meq:21} \\
 \encf{\core{\Gamma}}, f_x:\utype{\inpuse{0}{0}}{\encob{T}{\su}, \encob{S}{\su}} & \mlf & \f x\inp{y,c}.\encfp{P'}{f,\{x\mapsto c\}} \nonumber
\end{eqnarray}
for some context $\Gamma$, session types $S, T$, and renaming function $f$. 
By inversion on typing judgement on \eqref{meq:21} we have:
\begin{equation}
\label{meqj:7}
    \begin{prooftree}
      \Gamma, x:S, y:T \s P'
      \justifies
      \Gamma, x : \wn T.S \s x\inp{y}.P'
    \end{prooftree}
\end{equation}
By \Cref{d:trans} and \Cref{f:enctypesct}, we must show:
$$
 x\inp{y}.\chr{P'} \cp \encob{\Gamma}{\ct}, x : \encob{T}{\ct} \parl \encob{S}{\ct}
$$
By induction hypothesis 
on the premise of \eqref{meqj:7}
we have:
\begin{equation}
\chr{P'} \cp \encob{\Gamma}{\ct}, x : \encob{S}{\ct}, y : \encob{T}{\ct} \label{meq:22}
\end{equation}
and the thesis follows easily from \eqref{meq:22} using Rule~\jname{T-$\parl$} (cf. \Cref{fig:type-system-cll}).

\smallskip

\item $P = \overline x\out y.P'$: 
Then, by assumption, \Cref{d:encdgs} and \Cref{d:langmu}, 
we have both
\begin{eqnarray}
\Gamma,  x : \oc T.S, y:T & \s & \overline x\out{y}.P' \label{meq:23} \\
\encf{\core{\Gamma}}, f_x:\utype{\outuse{0}{\ca}}{\encob{T}{\su}, \encob{S}{\su}} & \mlf&  \res c {\f x}\out {y,{c}}.\encfp{P'}{f,\{x\mapsto c\}} \nonumber
\end{eqnarray}
for some typing context $\Gamma$, session types $S, T$ and renaming function $f$. 
By inversion on typing judgement on \eqref{meq:23} we have:
\begin{equation}
\label{meqj:8}
    \begin{prooftree}
      \Gamma, x:S \s P'
      \justifies
      \Gamma,  x : \oc T.S, y:T \s \overline x\out{y}.P'
    \end{prooftree}
\end{equation}
By \Cref{d:trans} and \Cref{f:enctypesct}, we must show:
$$
    \bout{x}{z}.(\linkr{z}{y} \pp \chr{P'})  \cp \encob{\Gamma}{\ct}, x : \encob{\dual{T}}{\ct} \otimes \encob{S}{\ct}, y : \encob{T}{\ct}
$$
By induction hypothesis
on the premise of \eqref{meqj:8}
we have:
$$
\chr{P'}  \cp \encob{\Gamma}{\ct}, x:\encob{S}{\ct} 
$$
and then we can conclude using Rule~\Did{T-$\otimes$} in \Cref{fig:type-system-cll}:
$$
    \begin{prooftree}
     \chr{P'}  \cp \encob{\Gamma}{\ct}, x:\encob{S}{\ct} \qquad \linkr{y}{z} \cp y:\encob{T}{\ct}, z:\encob{\dual{T}}{\ct} 
      \justifies
      \bout{x}{z}.(\linkr{z}{y} \pp \chr{P'})  \cp \encob{\Gamma}{\ct}, x : \encob{\dual{T}}{\ct} \otimes \encob{S}{\ct}, y : \encob{T}{\ct}
    \end{prooftree}
$$
where 
$\linkr{y}{z} \cp y:\encob{T}{\ct}, z:\encob{\dual{T}}{\ct}$
follows from Rule \Did{T-$\tid$}.
\smallskip

\item $P = \selection x {l_j}.P'$:
Then, by assumption, \Cref{d:encdgs}, and \Cref{d:langmu}, 
we have both 
\begin{eqnarray}
\Gamma, x:  \select lS & \s & {\selection x {l_j}.P'} \label{meq:sel} \\
\encf{\core{\Gamma}}, f_x:\utype{\outuse 0\ca}{\variant {l_i}{\encob{\dual{S_i}}\su}} 
& \mlf&  \res c \ov{\f x}\out {\vv{j}c}.\encfp{P'}{f,\{x\mapsto c\}}
\nonumber
\end{eqnarray}
for some typing context $\Gamma$, session types $S_i$ for $i\in I$, and renaming function $f$. 
By inversion on typing judgement on \eqref{meq:sel} we have:
\begin{equation}
\label{meq:sel-der}
    \inferrule
	{
		\Gamma,  x: S_j \s P' \quad   \exists j \in I
	}
	{ 	
		\Gamma, x:  \select lS \s \selection {x} {l_j}.P'
	}
\end{equation}
By \Cref{d:trans} and \Cref{f:enctypesct}, we must show:
$$
\selection {x} {l_j}.\chr{P'}  \cp \encob{\Gamma}{\ct}, x{:}\parselect {l_i}{\encob{S_i}{\ct}}{i\in I}
$$
By induction hypothesis
on the premise of \eqref{meq:sel-der}
we have:
$$
\chr{P'}  \cp \encob{\Gamma}{\ct}, x:\encob{S_j}{\ct} 
$$
and then we can conclude by using Rule~\Did{T-$\oplus$} in \Cref{fig:type-system-cll}:
$$
\inferrule
{\chr{P'}  \cp \encob{\Gamma}{\ct}, x:\encob{S_j}{\ct} \and j\in I}
{ \selection {x} {l_j}.\chr{P'}  \cp \encob{\Gamma}{\ct}, x{:}\parselect {l_i}{\encob{S_i}{\ct}}{i\in I}}
$$

\smallskip

\item $P = \branching xl{P'}$: Similar to the previous case.

\smallskip

\item $P = \res {xy}P'$: 
Then, by assumption, \Cref{d:encdgs}, and \Cref{d:langmu}
we have: 
\begin{eqnarray}
\Gamma  &\s&   \res {xy}P' \label{meq:23}\\
\encf{\core{\Gamma}} & \mlf & \res{w}\enc{P'}w  \label{meq:enc-res}
\end{eqnarray}
for some context $\Gamma$.
By \Cref{d:trans} and \Cref{f:enctypesct}, we must show:
$$
\res{w}(\chr{P_1}\substj{w}{x} \pp \chr{P_2}\substj{w}{y}) \cp \encob{\Gamma_1}{\ct},\encob{\Gamma_2}{\ct}
$$
for some $\Gamma = \Gamma_1, \Gamma_2$.

By inversion on \eqref{meq:enc-res}, using Rule~\Did{T$\pi$-Par+Res} in \Cref{f:mupityping}, we can infer the following:
\begin{itemize}
\item
Let $w: \utype{U}{\widetilde{T}}$, then $\rel{U}$.

\item
Since $\rel{U}$, then it must be the case that $U = U_1 \pp U_2$, which in turn implies that $\enc{P'}w = \encx{P'_1}w \pp \ency{P'_2}w$, for some $P'_1, P'_2$.

\item
Processes $\encx{P'_1}w$ and $\ency{P'_2}w$ share  \emph{exactly} one channel, namely $w$.
\end{itemize}

\noindent
By inversion on the encoding and renaming function $f$, and on the typing judgement on~\eqref{meq:23} we have:
\begin{equation}
\label{meqj:9}
    \begin{prooftree}
    \begin{prooftree}
      \Gamma_1,  x:S \s  P_1 
      \quad 
       \Gamma_2,  y: \dual{S} \s  P_2 
      \justifies
      \Gamma_1,  x:S ,  \Gamma_2, y: \dual{S} \s   P_1 \pp P_2
    \end{prooftree}
      \justifies
      \Gamma_1 ,  \Gamma_2  \s  \res {xy}(P_1 \pp P_2)
    \end{prooftree}
\end{equation}
for some session type  $S$ and $\Gamma = \Gamma_1, \Gamma_2$.
By applying the induction hypothesis on the premises of \eqref{meqj:9}, we have both
$$
\chr{P_1} \cp  \encob{\Gamma_1}{\ct},  x:\encob{S}{\ct}
\quad
\text{and}
\quad
\chr{P_2} \cp  \encob{\Gamma_2}{\ct},  y:\encob{\dual{S}}{\ct}
$$
By \Cref{lem:subst_cp} we obtain:
$$
\chr{P_1}\substj{w}{x} \cp  \encob{\Gamma_1}{\ct},  w:\encob{S}{\ct}
\quad
\text{and}
\quad
\chr{P_2}\substj{w}{y} \cp  \encob{\Gamma_2}{\ct},  w:\encob{\dual{S}}{\ct}
$$
We then conclude by using \Cref{lem:dualenc}
and Rule~\jname{T$\cut$} (cf. \Cref{fig:type-system-cll}):
$$
\begin{prooftree}
\chr{P_1}\substj{w}{x} \cp  \encob{\Gamma_1}{\ct},  w:\encob{S}{\ct} 
\quad
\chr{P_2}\substj{w}{y} \cp  \encob{\Gamma_2}{\ct},  w:\encob{\dual{S}}{\ct}
\justifies
\res{w}(\chr{P_1}\substj{w}{x} \pp \chr{P_2}\substj{w}{y}) \cp \encob{\Gamma_1}{\ct},\encob{\Gamma_2}{\ct}
\end{prooftree}
$$

\smallskip

\item $P = P_1 \pp P_2$: similar to the previous case, noticing that since there is no restriction to bind two endpoints together, membership of $P$ in \muKoba relies on Rule~\Did{T$\pi$-InDPar} in \Cref{f:mupityping} rather than on Rule~\Did{T$\pi$-Par+Res} (as in the previous case).
Then, we use Rule~\Did{T-$\mix$} (rather than Rule~\Did{T$\cut$}) to type the composition of $P_1$ and $P_2$ (cf. \Cref{fig:type-system-cll}). 
\end{enumerate}
\end{proof}

\section{Omitted Proofs for \Cref{s:enco}}
\label{sec:appendix-5}

\subsection{Proof of \Cref{lem: wt_characteristic}} \label{applem: wt_characteristic}
We repeat the statement in Page~\pageref{lem: wt_characteristic}:

\medskip
\noindent
\textbf{\Cref{lem: wt_characteristic}.}
Let $T$ be a session type and $\Gamma$ be a session context.
\begin{enumerate}
\item\label{5.4-1} For all $P \in \chrp{T}{x}$, we have  $P \cp x:\encob{T}{\ct}$.
\item\label{5.4-2} For all $P\in \chrpk{\Gamma}{}$, we have $P\cp \encob{\Gamma}{\ct}$.
\end{enumerate}
\begin{proof}
We consider both parts separately.

\bigskip

\noindent{\underline{\textbf{\Cref{5.4-1}.}}} 
The proof proceeds by induction on the structure of  $T$.
Thus, there are five cases to consider:
\begin{enumerate}
\item
Case 
$T = \nilT$. Then, by \Cref{def:charprocess}, we have 
$\chrp{\nilT}{x} =  \big\{ \nil \big\}$. We conclude by the fact that $\encob{\nilT}{\ct}= \bullet$ (cf. \Cref{f:enctypesct}) and by Rule~\Did{T-\one} (cf. \Cref{fig:type-system-cll}).
 
\smallskip

\item
Case 
$T = \wn T.S$. Then, by \Cref{def:charprocess},  $P$ is of the form  
$x(y).(P_1\pp P_2)$, with $P_1 \in \chrp{T}{y}$ and $P_2 \in \chrp{S}{x}$.
By applying the induction hypothesis twice, on $T$ and $S$, 
we obtain:
$$
P_1 \cp y:\encob{T}{\ct} \qquad 
P_2 \cp x:\encob{S}{\ct}
$$
Now, 
by applying Rules~\Did{T-$\mix$} and~\Did{{T-$\parl$}}  (cf. \Cref{fig:type-system-cll}) we have
$$
\begin{prooftree}
{P_1\cp y: \encob{T}{\ct} \qquad P_2\cp x:\encob{S}{\ct}}
\justifies
\begin{prooftree}
{P_1\pp P_2 \cp  y: \encob{T}{\ct}, x:\encob{S}{\ct}}
\justifies
{x(y).(P_1\pp P_2) \cp x{:}\encob{T}{\ct}\parl  \encob{S}{\ct}}
\end{prooftree}
\end{prooftree}
$$
By encoding of types in \Cref{f:enctypesct}, we have
$\encob{\wn T.S}{\ct}  =   \encob{T}{\ct} \parl \encob{S}{\ct}$, which concludes this case.
 
\smallskip

\item
Case 
$T = \oc T.S$.
 Then, by \Cref{def:charprocess},   $P$ is of the form
$\bout{x}{y}.(P_1 \para P_2)$, with $P_1 \in \chrp{{\dual T}}{y}$ and $P_2 \in \chrp{S}{x}$.
By applying the induction hypothesis twice, on $T$ and $S$, 
we obtain
$$
P_1 \cp y:\encob{\dual{T}}{\ct} \qquad 
P_2 \cp x:\encob{S}{\ct}
$$
Now, 
by applying Rule~\Did{{T-$\otimes$}} (cf. \Cref{fig:type-system-cll}) we have
$$
\inferrule
{P\cp y: \encob{\dual{T}}{\ct} \\ Q\cp x:\encob{S}{\ct}}
{\bout{x}{y}.(P \para Q)\cp x: \encob{\dual T}{\ct}\otimes \encob{S}{\ct}}
$$
By encoding of types in \Cref{f:enctypesct}, we have
$\encob{\oc T.S}{\ct} =  \encob{\dual T}{\ct} \otimes \encob{S}{\ct}$ which concludes this case.
 
\smallskip

\item
Case 
$T = \branch lS$.
 Then, by \Cref{def:charprocess},   $P$ is of the form
$\branching xlP$, with $P_i\in \chrp{S_i}x$, for all $i\in I$.
By induction hypothesis on those $S_i$, we obtain 
$P_i\cp x:\encob{S_i}{\ct}$, for all $i\in I$. Then, by 
Rule~\Did{T$\with$} (cf. \Cref{fig:type-system-cll}) we have:
$$
\inferrule
{P_i\cp x:\encob{S_i}{\ct}\and  \forall i\in I}
{\branching xlP \cp x: \parbranch {l_i}{\encob{S_i}{\ct}}{i\in I}}
$$
By encoding of types in \Cref{f:enctypesct}, 
$\encob{\branch lS}{\ct} = \parbranch {l_i}{\encob{S_i}{\ct}}{i\in I}$,
which concludes this case.
 
\smallskip

\item
Case $T= \select lS$. 
 Then, by \Cref{def:charprocess},  $P$ is of the form
$\selection x{l _j}.{P_j}$, with $P_j\in \chrp{S_j}x$ and $j \in I$.
By induction hypothesis on $S_j$, we obtain 
$P_j\cp x:\encob{S_j}{\ct}$.
Then, 
by Rule~\Did{T$\oplus$} (cf. \Cref{fig:type-system-cll}) we have
$$
\inferrule
{P_j\cp x: \encob{S_j}{\ct}}
{\selection {x} {l_j}.{P_j}  \cp x{:}\parselect {l_i}{\encob{S_i}\ct}{i\in I}}
$$
By encoding of types in \Cref{f:enctypesct}, we have
$\encob{\select lS}{\ct} = \parselect {l_i}{\encob{S_i}\ct}{i\in I}$, which concludes this case.
\end{enumerate}

\bigskip

\noindent{\underline{\textbf{\Cref{5.4-2}.}}} 
Given  $\Gamma = w_1:T_1, \ldots, w_n:T_n$, the proof is by induction on $n$, the size of $\Gamma$.
The base case is when $n=1$: then, by \Cref{5.4-1},
$
\chrpk{\Gamma}{}
=
 \chrp{T_1}{w_1}
$,  
and the thesis follows immediately.
The inductive step ($n > 1$) proceeds by using the inductive hypothesis and Rule~\Did{T-$\mix$}.
Let $\Gamma = \Gamma', w_n:T_n$. Then, by inductive hypothesis, 
$P_i\in \chrpk{\Gamma'}{}$ implies $P_i\cp \encob{\Gamma'}{\ct}$.
Recall that  
$\encob{\Gamma'}{\ct} = w_1{:}\encob{T_1}{\ct}, \ldots, w_{n-1}{:}\encob{T_{n-1}}{\ct}$
by \Cref{def:enc_env_sc}.
Also, by \Cref{5.4-1}, 
$P' \in \chrp{T_n}{w_n}$ implies  $P' \cp w_n:\encob{T_n}{\ct}$.
Since by \Cref{def:charenv}, $P_i \para P' \in \chrpk{\Gamma}{}$, 
  the thesis follows by composing $P_i$ and $P'$ using Rule~\Did{T-$\mix$}:
$$
\inferrule
{P_i \cp w_1:\encob{T_1}{\ct}, \cdots, w_{n-1}:\encob{T_{n-1}}{\ct} \qquad P' \cp w_n:\encob{T_n}{\ct}}
{P_i \para  P' \cp w_1:\encob{T_1}{\ct}, \cdots, w_n:\encob{T_n}{\ct}}
$$
This concludes the proof. 
\end{proof}

\subsection{Proof of \Cref{thm:L0-L2}}\label{appthm:L0-L2}
%
We repeat the statement in Page~\pageref{thm:L0-L2}:

\medskip
\noindent
\textbf{\Cref{thm:L0-L2}.} ($\encCP{\cdot}$ is Type Preserving)
\label{appthm:L0-L2}
Let $\Gamma\s P$.
Then, for all $Q\in \encCP{{\Gamma}\s P}$, we have that $Q\cp \encob{\Gamma}{\ct}$.
\begin{proof}
The proof proceeds by cases on the judgement $\Gamma\s P$ used in the translation given in \Cref{def:typed_enc}, and by inversion on the last typing rule applied (given in \Cref{fig:sess_typing}).
There are seven cases to consider:

\begin{enumerate}
\item
${\Gamma \s \nil}$. Then, by inversion the last rule applied is \Did{T-Nil}, with $\eend{\Gamma}$.
%
By \Cref{def:typed_enc} we have
$\encCP{\Gamma \s \nil} = \{ \nil \}$.
By applying Rule~{\Did{T-$\one$}} (cf. \Cref{fig:type-system-cll}) we have:
$$
\inferrule*[Right = {\Did{T-$\one$}}]
{ }
{\nil\cp  x: \bullet}
$$
The thesis follows immediately by the encoding of types $\encob{\cdot}{\ct}$ in \Cref{f:enctypesct}
and \Cref{def:enc_env_sc}, which ensure that
$\encob{x:\nilT}{\ct}=x:\encob{\nilT}{\ct}=x: \bullet$.

\smallskip

\item
${{\Gamma'}, x:\oc T.S, v:T  \s\dual{x}\out v.P'}$, where 
$\Gamma = \Gamma', x:\oc T.S, v:T$ and
$P = \dual{x}\out v.P'$.
By inversion the last rule applied is \Did{T-Out}:
$$
\inferrule*[Right = \Did{T-Out}]
{\Gamma', x:S \s P'}
{\Gamma', x:\oc T.S, v:T \s {\overline{x}\out v.P'}}
$$
By \Cref{def:typed_enc},
$$\encCP{{ \Gamma', x:\oc T.S, v:T} \s\dual{x}\out v.P'} = \{\dual{x}(z). \big(\linkr{v}{z}\para Q \big) \suchthat Q\in \encCP{{{\Gamma'}, x:{S}}\s P'}\}$$
By Rule~\Did{T-$\tid$} and by \Cref{lem:dualenc} we have:
\begin{equation}
\inferrule*[Right =  \Did{T-$\tid$}]
{ }
{\linkr{v}{z}\cp v:\encob{T}{\ct}, z:\encob{\overline T}{\ct}}
\label{eq:tid}
\end{equation}
By induction hypothesis, for all
$Q\in \encCP{{{\Gamma'}, x:{S}}\s P'}$ we have that $Q\cp\encob{\Gamma'}{\ct}, x:\encob{S}{\ct}$.
Let $Q'$ be a process in this set.
By applying Rule~\Did{T-$\otimes$} on $Q'$ and on \eqref{eq:tid} we have
$$
\inferrule*[Right=\Did{T-$\otimes$}]
{\linkr{v}{z}\cp v:\encob{T}{\ct}, z:\encob{\overline T}{\ct}  \\ 
Q' \cp\encob{\Gamma'}{\ct}, x:\encob{S}{\ct}}
{\dual{x}(z). \big(\linkr{v}{z}\para Q' \big)
\cp \encob{\Gamma'}{\ct},x: \encob{\dual T}{\ct}\otimes \encob{S}{\ct}, v:\encob{T}{\ct} }
$$
By encoding of types $\encob{\cdot}{\ct}$ in \Cref{f:enctypesct} we have
$\encob{\oc T.S}{\ct} =  \encob{\dual T}{\ct}\otimes \encob{S}{\ct}$,
and
by \Cref{def:enc_env_sc}, we have
$\encob{ \Gamma', x:\oc T.S, v:T}{\ct} = \encob{\Gamma'}{\ct},x: \encob{\dual T}{\ct}\otimes \encob{S}{\ct}, v:\encob{T}{\ct}$, which concludes this case.

%

\smallskip

\item
${{\Gamma_1,\Gamma_2}, x:\oc T.S \s\res{zy}\dual{x}\out y.(P_1\para P_2)}$, where $\Gamma = \Gamma_1,\Gamma_2, x:\oc T.S$.
By inversion, this judgement is derived by a sequence of applications of rules,
the last rule applied is \Did{T-Res}, and before that \Did{T-Out} and \Did{T-Par} as follows:
$$
\inferrule*[Right = \Did{T-Res}]
{
\inferrule*[Right = \Did{T-Out}]
{\inferrule*[Right = \Did{T-Par}]
{\Gamma_1, z: \overline T \s P_1 \and \Gamma_2, x:S \s P_2}
{\Gamma_1, z: \overline T , \Gamma_2, x:S\s P_1\pp P_2}}
{\Gamma_1, \Gamma_2, x:\oc T.S, y:T, z: \overline T \s {\overline{x}\out y.(P_1\para P_2)}}
}
{{\Gamma_1,\Gamma_2}, x:\oc T.S \s \res{zy}\dual{x}\out y.(P_1\para P_2)}
$$
By \Cref{def:typed_enc}, we have
\begin{align*}
\ \
\encCP{{\Gamma_1,\Gamma_2}, x:\oc T.S&\s\res{zy}\dual{x}\out y.(P_1\para P_2)} =\\
&\big\{ \dual{x}(z).(Q_1\para Q_2)  \suchthat   Q_1 \in \encCP{{{\Gamma_1},  z:\dual{T}\s P_1}}
\ \wedge\  Q_2 \in \encCP{{{\Gamma_2},x:{S}\s P_2}}\big\}
\end{align*}
By induction hypothesis, for all processes
$$Q_1\in \encCP{{{\Gamma_1},  z:\dual{T}\s P_1}} \mbox{ we have }
Q_1\cp\encob{\Gamma_1}{\ct}, z:\encob{\dual T}{\ct}$$
and 
$$Q_2\in \encCP{{{\Gamma_2},  x:S\s P_2}} \mbox{ we have }
Q_2\cp\encob{\Gamma_2}{\ct}, x:\encob{S}{\ct}$$
Let $Q'_1$ and $Q'_2$ be processes in the first and second set, respectively.
By applying Rule~\Did{T-$\otimes$} on $Q'_1$ and $Q'_2$ we have:
$$
\inferrule*[Right=\Did{T-$\otimes$}]
{Q'_1\cp\encob{\Gamma_1}{\ct}, z:\encob{\dual T}{\ct}  \\ 
Q'_2\cp\encob{\Gamma_2}{\ct}, x:\encob{S}{\ct}}
{\dual{x}(z). \big(Q'_1\para Q'_2 \big)
\cp \encob{\Gamma_1}{\ct},\encob{\Gamma_2}{\ct},x: \encob{\dual T}{\ct}\otimes \encob{S}{\ct} }
$$
By the encoding of types $\encob{\cdot}{\ct}$ in \Cref{f:enctypesct} we have
$\encob{\oc T.S}{\ct} =  \encob{\dual T}{\ct}\otimes \encob{S}{\ct}$, and
by \Cref{def:enc_env_sc}, we have
$\encob{ \Gamma_1, \Gamma_2, x:\oc T.S}{\ct} = \encob{\Gamma_1}{\ct}, \encob{\Gamma_2}{\ct},x: \encob{\dual T}{\ct}\otimes \encob{S}{\ct}$, which concludes this case.

\smallskip

\item
$\Gamma', x:\wn T.S \s {x\inp {y:T}.P'}$, where $\Gamma = \Gamma', x:\wn T.S$.
By inversion, the last typing rule applied is \Did{T-In}:
$$
\inferrule*[Right = \Did{T-In}]
{\Gamma', x:S, y:T \s P'}
{\Gamma',x:\wn T.S \s {x\inp {y:T}.P'}}
$$
By \Cref{def:typed_enc} we have
$$\encCP{{\Gamma', x:\wn T.S}\s x\inp {y:T}.P'}=\{x(y).Q \suchthat Q\in \encCP{{\Gamma',x:S,y:T}\s P'}\}$$
By induction hypothesis,  
$Q \in \encCP{{\Gamma',x:S,y:T}\s P'}$ implies
$Q\cp \encob{\Gamma'}{\ct}, x:\encob{S}{\ct}, y:\encob{T}{\ct}$.
Let $Q'$ be a process in this set.
By applying Rule~\Did{T-$\parl$} on $Q'$:
$$
\inferrule*[Right=\Did{T-$\parl$}]
{Q'\cp \encob{\Gamma'}{\ct}, x:\encob{S}{\ct}, y:\encob{T}{\ct}}
{x(y).Q' \cp \encob{\Gamma'}{\ct}, x:\encob{T}{\ct}\parl \encob{S}{\ct}}
$$
where by the encoding of types $\encob{\cdot}{\ct}$ in \Cref{f:enctypesct} we have
$\encob{\wn T.S}{\ct} =  \encob{T}{\ct}\parl \encob{S}{\ct}$, and
by \Cref{def:enc_env_sc}, we have
$\encob{ \Gamma', x:\wn T.S}{\ct} = \encob{\Gamma'}{\ct}, x:\encob{T}{\ct}\parl \encob{S}{\ct}$, which concludes this case.

\smallskip

\item
$\Gamma', x:\select lS \s {\selection x{l_j}.P'}$, where $\Gamma = \Gamma', x:\select lS$.
By inversion, the last typing rule applied is \Did{T-Sel}:
$$
\inferrule*[Right = \Did{T-Sel}]
{\Gamma', x:S_j \s P' \and j\in I}
{\Gamma', x:\select lS \s {\selection x{l_j}.P'}}
$$
By \Cref{def:typed_enc} we have
$$\encCP{{\Gamma', x:\select lS} \s\selection x{l_j}.P'}= \{\selection x{l_j}.Q \suchthat Q\in {\encCP{{\Gamma',x:S_j}\s P'}}\}$$
By induction hypothesis, for all
$Q \in \encCP{{\Gamma',x:S_j}\s P'}$ we have $Q\cp \encob{{\Gamma'}}{\ct}, x:\encob{S_j}{\ct}$.
Let $Q'$ be a process in this set.
By applying Rule \Did{T-$\oplus$} on $Q'$ we have:
$$
\inferrule*[Right = \Did{T-$\oplus$}]
{Q'\cp \encob{{\Gamma'}}{\ct}, x:\encob{S_j}{\ct} \and j\in I}
{\selection x{l_j}.Q' \cp  \encob{{\Gamma'}}{\ct}, x: \oplus \{l_i: \encob{S_i}{\ct}\}_{i\in I}}
$$
By the encoding of types $\encob{\cdot}{\ct}$ in \Cref{f:enctypesct},
and \Cref{def:enc_env_sc}, we have
$$\encob{\Gamma', x:\select lS}{\ct} =  \encob{{\Gamma'}}{\ct}, x:\oplus \{l_i: \encob{S_i}{\ct}\}_{i\in I}$$
which concludes this case.

\smallskip

\item
$\Gamma' ,x:\branch lS\s {\branching xlP}$, where $\Gamma = \Gamma' ,x:\branch lS$.
By inversion, the last typing rule applied is \Did{T-Bra}:
$$
\inferrule*[Right = \Did{T-Bra}]
{\Gamma', x: S_i\s P_i \and \forall i\in I}
{\Gamma' ,x:\branch lS \s {\branching xlP}}
$$
By \Cref{def:typed_enc} we have that
$$\encCP{{\Gamma' ,x:\branch lS}\s  \branching xlP}= \{\parbranching x{l_i}Q \suchthat Q\in {\encCP{{\Gamma',x:S_i}\s P_i}}\}$$
By induction hypothesis,
$Q_i \in\encCP{{\Gamma',x:S_i}\s P_i}$ implies $Q_i \cp \encob{\Gamma'}{\ct}, x:\encob{S_i}{\ct}$,  for all $i\in I$.
Let $Q'_i$ be a process in the corresponding set,  for all $i\in I$.
By applying Rule~\Did{T-$\&$} on each of $Q'_i$ we have:
$$
\inferrule*[Right=\Did{T-$\&$}]
{Q'_i\cp \encob{\Gamma'}{\ct}, x:\encob{S_i}{\ct}}
{\parbranching x{l_i}{Q'_i}\cp \encob{\Gamma'}{\ct}, x: \&\{l_i:\encob{S_i}{\ct}\}_{i\in I}}
$$
By the encoding of types $\encob{\cdot}{\ct}$ in \Cref{f:enctypesct} we have
$\encob{\branch lS}{\ct} = \&\{l_i:\encob{S_i}{\ct}\}_{i\in I}$, and
by \Cref{def:enc_env_sc}, we have
$\encob{\Gamma' ,x:\branch lS}{\ct} = \encob{\Gamma'}{\ct}, x: \&\{l_i:\encob{S_i}{\ct}\}_{i\in I}$, which concludes this case.

\smallskip

\item
${\Gamma_1, \restricted{\wt{x:S}},  \Gamma_2, \restricted{\wt{y:T}}\s \res {\wt{x}\wt{y}:\wt{S}}(P_1\para P_2)}$, where by inversion we have
$\Gamma_1,  \wt{x:S} \s P_1$ and $\Gamma_2,  \wt{y:T}\s P_2$,
and the last typing rule applied is \Did{T-Res}, and before that Rule~\Did{T-Par} is used, as follows:
$$
\inferrule*[Right = \Did{T-Res}]
{
\inferrule*[Right = \Did{T-Par}]
{\Gamma_1,  \wt{x:S} \s P_1\ (1) \\ \Gamma_2,  \wt{y:T}\s P_2 \ (2)}
{\Gamma_1 ,  \Gamma_2,  \wt{x:S},\wt{y:T} \s P_1\para P_2}
}
{\Gamma_1, \restricted{\wt{x:S}},  \Gamma_2, \restricted{\wt{y:T}}\s \res {\wt{x}\wt{y}:\wt{S}}(P_1\para P_2)}
$$
Notice that since restriction is the only means of creating dual session channel endpoints (co-variables) and  the only restricted names in $P$ are $\wt{xy}$,
it then follows that $\Gamma_1\cap\Gamma_2 =\emp$.
Hence, by the definition of the `$, $' operator 
we have that
$\Gamma_1 ,  \Gamma_2 = \Gamma_1,\Gamma_2$.

By \Cref{def:typed_enc}, we have that $\encCP{{\Gamma_1, \restricted{\wt{x:S}} ,  \Gamma_2, \restricted{\wt{y:T}}}\s\res {\wt {xy}:\wt {S}}(P_1\para P_2)}$ is the following set of processes:
\begin{align}
\label{eq:old_union_par}
\big\{
C_1[Q_1] \para G_2  \suchthat  &C_1 \in \mathcal{C}_{\wt{x:T}},\, Q_1 \in \encCP{{\Gamma_1, {\wt{x:S}}}\s P_1}, \, G_2 \in  \chrpk{\Gamma_2}{}  \big\}
\\[1mm]									
\cup ~~\big\{G_1 \para C_2[Q_2]   \suchthat  & C_2 \in \mathcal{C}_{\wt{y:S}},\,  Q_2 \in \encCP{\Gamma_2,\wt{y:T}\s P_2}, \,
G_1 \in \chrpk{\Gamma_1}{}\label{eq:old_union_par2}
\big\}
\end{align}

We start by inspecting the set of processes in \eqref{eq:old_union_par}.
By induction hypothesis on the left-hand side premise of Rule~\Did{T-Par}, marked (1), we have:
$$\text{for all processes } Q\in \encCP{{\Gamma_1, {\wt{x:S}}}\s P_1} \text{ we have that }
Q\cp \encob{\Gamma_1}{\ct},  \wt{x: {\encob{S}{\ct}}}$$
Let $Q'$ be an arbitrary process in this set.
By \Cref{lem:typability_catal} we have that $C_1[Q'] \cp \encob{\Gamma_1}{\ct}$.
By \Cref{cor:context_char}(b)
since $G_2\in \chrpk{\Gamma_2}{}$, we have that $G_2 \cp \encob{\Gamma_2}{\ct}$.
Since $\Gamma_1$ and $\Gamma_2$ are disjoint, by Rule~\Did{T-$\mix$} we have the following derivation, which concludes the inspection of (\ref{eq:old_union_par}):
$$
\inferrule*[Right=\Did{T-$\mix$}]
{C_1[Q'] \cp \encob{\Gamma_1}{\ct} \\ G_2\cp  \encob{\Gamma_2}{\ct}}
{C_1[Q']\para G_2\cp \encob{\Gamma_1}{\ct}, \encob{\Gamma_2}{\ct} }
$$

We inspect now the the set of processes in \eqref{eq:old_union_par2}.
By induction hypothesis on the right-hand side premise of Rule~\Did{T-Par}, marked (2),
we have:
$$\text{for all processes }	R \in \encCP{{\Gamma_2,{\wt{y:T}}}\s R} \text{ we have that }
R\cp \encob{\Gamma_2}{\ct}, \wt {y:\encob{T}{\ct}}$$

Let $R'$ be an arbitrary process in this set.
By \Cref{lem:typability_catal}, $C_2[R'] \cp \encob{\Gamma_2}{\ct}$.
By \Cref{cor:context_char}(b)
since $G_1\in \chrpk{\Gamma_1}{}$, we have $G_1 \cp \encob{\Gamma_1}{\ct}$.
Since $\Gamma_1$ and $\Gamma_2$ are disjoint, by Rule~\Did{T-$\mix$} we have the following derivation:
$$
\inferrule*[Right=\Did{T-$\mix$}]
{C_2[R'] \cp \encob{\Gamma_2}{\ct} \\ G_1\cp  \encob{\Gamma_1}{\ct}}
{C_2[R']\para G_1\cp \encob{\Gamma_1}{\ct}, \encob{\Gamma_2}{\ct} }
$$
We thus conclude that every process belonging to the set in \eqref{eq:old_union_par} or \eqref{eq:old_union_par2} is typed under the typing context
$\encob{\Gamma_1}{\ct}, \encob{\Gamma_2}{\ct}$, concluding this case (and the proof).
\end{enumerate}
\end{proof}

\subsection{Proof of \Cref{thm:oc} (Page~\pageref{thm:oc})}\label{appthm:oc}
We repeat \Cref{def:uptok} (cf. Page~\pageref{def:uptok}):
\begin{definition}
Let $P, Q$ be processes such that $P,Q\cp \Gamma$.
We write $P\uptok Q$ if and only if there exist
$P_1, P_2, Q_1, Q_2$ and  $\Gamma_1, \Gamma_2$ such that the following hold:
\begin{align*}
P = P_1\pp P_2
\qquad
Q = Q_1\pp Q_2
\qquad
P_1, Q_1\cp \Gamma_1
\qquad
P_2, Q_2\cp \Gamma_2
\qquad
\Gamma = \Gamma_1,\Gamma_2
\end{align*}
\end{definition}
\begin{lemma}[Substitution Lemma for Sessions~\cite{V12}]\label{lem:subst_sess}
If $\Gamma_1\s v:T$ and $\Gamma_2,x:T\s P$ and
$\Gamma=\Gamma_1, \Gamma_2$, then
$\Gamma\s P[v/x]$.
\end{lemma}
\begin{lemma}[Substitution Lemma for $\encCP\cdot$]
\label{lem:enc_subst}
If
$P\in \encCP{\Gamma, x:T\s Q}$ and $v\notin \fv{P,Q}$,
then
$P\substj vx \in \encCP{\Gamma, v:T \s Q[v/x]}$.
\end{lemma}
\begin{proof}
Immediate from \Cref{def:typed_enc} and \Cref{lem:subst_sess}.
\end{proof}

\begin{proposition}[Composing Characteristic Processes]
\label{prop:charp}
Let $\Gamma$ and  $T$ be a typing context and a type, respectively.
\begin{itemize}
\item 
If $P_1 \in \chrpk{\Gamma}{}$ and $P_2 \in \chrpk{T}{x}$ then $P_1 \para P_2 \in \chrpk{\Gamma, x:T}{}$.
\item If $Q \in \chrpk{\Gamma, x:T}{}$ then there are $Q_1, Q_2$ such that $Q = Q_1 \para Q_2$ with $Q_1 \in \chrpk{\Gamma}{}$ and $Q_2 \in \chrpk{T}{x}$.
\end{itemize}
\end{proposition}
\begin{proof}
Immediate from \Cref{def:charenv}.
\end{proof}

We repeat the statement of the operational correspondence given in Page~\pageref{thm:oc}:

\medskip
\noindent
\textbf{\Cref{thm:oc}.} (Operational Correspondence for $\encCP{\cdot}$)
\label{appthm:oc}
Let $P$ be such that $\Gamma\s P$ for some typing context $\Gamma$. Then, we have:
\begin{enumerate}
\item\label{5.2-1}
\nurev{If $P\to P'$, then for all $Q \in \encCP{{\Gamma}\s P}$ there exist $Q', R$ such that
$Q\to\fred Q'$, $Q' \uptok R$, and $R \in \encCP{{\Gamma}\s P'}$.}
\item\label{5.2-2}
\nurev{If $Q \in \encCP{{\Gamma}\s P}$, such that $P\in \lkoba$, and $Q \to\fred Q'$, then there exist
$P', R$ such that
$P\to P'$,
$Q' \uptok R$, and $R \in \encCP{{\Gamma}\s P'}$.}
\end{enumerate}
%

\begin{proof}
We consider both parts separately. 

\bigskip

\noindent{\underline{\textbf{\Cref{5.2-1}.}}} The proof  is by induction on the height of the derivation $P\to P'$ (cf. \Cref{fig:sessionpi}).
There are two main cases to consider, which are reductions inferred using Rules \Did{R-Com} and \Did{R-Case}; 
there are also cases corresponding to Rules \Did{R-Par}, \Did{R-Res}, and \Did{R-Str}, which are straightforward via the induction hypothesis.
For convenience, below we annotate bound names with their types: this way, e.g.,
$\res {xy:S}$ means that $x:S$ and  $y:\ov S$.

\begin{enumerate}
\item
Case \Did{R-Com}: Then we have:
$$P\defeq \res {xy:S'}(\overline x\out v.P_1\pp y\inp {t:T}.P_2) 
		\quad\to \quad
		\res {xy:S''}(P_1\pp P_2[v/t]) \defeq P'$$
Since $\Gamma\s P$, then by inversion we get
$S'=\oc T.S$ for some session types $S, T$. Then, $S''=S$.
Again by inversion we have the following derivation:
$$
\inferrule*[Right = \Did{T-Res}]
{
\inferrule*[Right=\Did{T-Par}]
{
\inferrule*[left = \Did{T-Out}]
{\Gamma_1,x:S\s P_1}
{\Gamma_1,v:T,x:\oc T.S\s \dual x\out v.P_1}
\and
\inferrule*[right=\Did{T-Inp}]
{\Gamma_2,y:\dual S, t:T\s P_2}
{\Gamma_2,y:\wn T.\dual S\s y\inp {t:T}.P_2}
}
{\Gamma_1,v:T,x:\oc T.S ,  \Gamma_2,y:\wn T.\dual S \s
\dual x\out v.P_1 \para y\inp {t:T}.P_2
}
}
{(\Gamma_1,v:T),  \Gamma_2 \s \res {xy:S'}(\overline x\out v.P_1\pp y\inp {t:T}.P_2)}
$$

By \Cref{def:typed_enc}, the translation of $P$ is as follows:
\begin{align*}
\encCP {\Gamma\s P}= &
\encCP{(\Gamma_1,v:T),  \Gamma_2\s\res {xy:\oc T.S}(\overline x\out v.P_1\pp y\inp {t:T}.P_2)}\\[2mm]
=\ & \underbrace{\big\{
C_1[Q_1] \para G_2  \suchthat  C_1 \in \mathcal{C}_{x:\wn T.\dual S},\,
Q_1 \in \encCP{{\Gamma_1, {x:\oc T.S}},v:T\s \overline x\out v.P_1}, \, G_2 \in  \chrpk{\Gamma_2}{}  \big\}}_{A_1} 
\\							
& ~\cup
\\						
& \underbrace{\big\{G_1 \para C_2[Q_2]   \suchthat  C_2 \in \mathcal{C}_{y:\oc T. S},\,
Q_2 \in \encCP{{\Gamma_2,{y:{\wn T.\dual S}}}\s y\inp {t:T}.P_2}, \,
G_1 \in \chrp{\Gamma_1, v:T}{}
\big\}}_{A_2}
\end{align*}
where:
\begin{align}
A_1 =\ &\big\{
\res x\big({\dual{x}(w). \big(\linkr{v}{w}\para P^*_1\big)} \pp P_x\big) 
\para G_2 
\suchthat \nonumber \\
& \qquad \qquad   
P^*_1 \in \encCP{{{\Gamma_1}, x:{S}}\s P_1},\,
P_x \in \chrp{\wn T.\dual S}{x},\,
G_2 \in  \chrpk{\Gamma_2}{}  \big\}\quad 
\label{eq:52a1}
\\[1mm]							
A_2 =\ & \big\{
G_1 \para
\res y \big(y\inp {t}.P^*_2 \pp P_y\big) \suchthat \nonumber \\
&  \qquad \qquad    
P^*_2 \in \encCP{\Gamma_2,y:\dual S, t:T\s P_2},\,
P_y \in\chrp{\oc T.S}y,\,
G_1 \in \chrp{\Gamma_1, v:T}{}
\big\} 
\label{eq:52a2}
\end{align}

\noindent 
Before spelling out the translation of $P'$, we record some considerations.
By \Cref{thm:L0-L2}, the translation preserves types:
  $Q\in \encCP{{\Gamma}\s P}$  implies $Q\cp \encob{\Gamma}{\ct}$.
Since $P\to P'$,  \Cref{thm:subj_red} ensures $\Gamma\s P'$.
Again, by \Cref{thm:L0-L2},
$O\in \encCP{{\Gamma}\s P'}$ implies $O\cp \encob{\Gamma}{\ct}$.
Also, since ${\Gamma_2,y:\dual S, t:T\s P_2}$, then by \Cref{lem:subst_sess} we have
$\Gamma_2,y:\dual S, v:T\s P_2 \substj vt$ and by well-typedness $v\notin \fn {P_2} $.
By \Cref{def:typed_enc}, the translation of $P'$ is as follows:
\begin{align*}
\encCP {\Gamma\s P'} 
=\ & \encCP{\Gamma_1,  (\Gamma_2,v:T)\s\res {xy:S}(P_1\pp P_2\substj vt)}\\[2mm]
=\ & \underbrace{\big\{
C'_1[Q_1] \para G'_2  \suchthat  C'_1 \in \mathcal{C}_{x:\dual S},\,
Q_1 \in \encCP{{\Gamma_1, {x:S}}\s P_1}, \, G'_2 \in  \chrpk{\Gamma_2, v:T}{}  \big\}}_{B_1}
\\
& ~ \cup 
\\								
&\underbrace{\big\{G'_1 \para C'_2[Q_2]   \suchthat  C'_2 \in \mathcal{C}_{y:S},\,
Q_2 \in \encCP{\Gamma_2,y:{\dual S}, v:T\s P_2\substj vt}, \,
G'_1 \in \chrp{\Gamma_1}{}
\big\}}_{B_2}
\end{align*}
where:
\begin{align}
B_1 =\ & \big\{
G'_2 \para \res x\big(Q_1 \pp P'_x \big) \suchthat \nonumber \\
& \qquad \qquad 
   P'_x \in \chrp{\dual S}{x},\,
Q_1 \in \encCP{{{\Gamma_1}, x:{S}}\s P_1}, \, G'_2 \in  \chrpk{\Gamma_2, v:T}{}  \big\}
\label{eq:52b1}
\\[1mm]									
B_2 = &\big\{G'_1 \para \res y\big(Q_2 \pp P'_y \big) \suchthat  \nonumber \\
& \qquad \qquad 
   P'_y \in\chrp{S}y,\,
Q_2 \in \encCP{\Gamma_2,y:{\dual S}, v:T\s P_2\substj vt}, \,
G'_1 \in \chrp{\Gamma_1}{}
\big\} 
\label{eq:52b2}
\end{align}
%


{We now show} that every process in $\encCP {\Gamma\s P}$ reduces into a process in $\encCP {\Gamma\s P'}$. 
We address two distinct sub-cases:
\begin{enumerate}
\item[(i)] We show that every  $Q \in A_1$ (cf. \Cref{eq:52a1}) reduces to a $Q' \in B_1$ (cf. \Cref{eq:52b1});
\item[(ii)] We show that every  $Q \in A_2$ (cf. \Cref{eq:52a2}) reduces to a 
$Q'$ such that $Q'' \uptok R$, with $R \in B_2$ (cf. \Cref{eq:52b2}).
\end{enumerate}


\medskip
\noindent \textbf{\emph{Sub-case (i)}}. 
Let 
$Q = G_2 \para \res x\big({\dual{x}(w). \big(\linkr{v}{w}\para P^*_1\big)} \pp P_x\big)$
be an arbitrary process in $A_1$,
with $P_x \in \chrp{\wn T.\dual S}{x}$.
By \Cref{def:charprocess}: 
$$\chrp{\wn T.\dual S}{x} = {\big\{x(t).(Q_t \para Q_x) \suchthat Q_t \in \chrp{T}{t}\ \land\ Q_x \in \chrp{\dual S}{x}\big\}}$$
We may then let
$P_x = x(t).(Q_t \para Q_x)$ where $Q_t \in \chrp{T}{t}$ and 
$Q_x \in \chrp{\dual S}{x}$.
Considering this, and 
by applying  Rules~\Did{R-ChCom} and \Did{R-Fwd}
(cf. \Cref{fig:redlcp}) we have: 
$$
\begin{array}{rllll}
Q =  &
G_2 \para
\res x\big({\dual{x}(w). \big(\linkr{v}{w}\para P^*_1\big)} \pp P_x\big)\\[2mm]
= &G_2 \pp \res x\big({\dual{x}(w). \big(\linkr{v}{w}\para P^*_1 \big)} \pp x(t).(Q_t \para Q_x)\big)\\[2mm]
\to& G_2 \pp \res x \big( \res w \big( \linkr{v}{w}\para P^*_1 \pp Q_w \pp Q_x  \big) \big) \\[2mm]
\fred& G_2   \pp Q_v \pp \res x \big( P^*_1 \pp Q_x \big) \ \defeq\ Q'
\end{array}
$$
(Recall that $\fred$ is structural congruence extended with a reduction by Rule~\Did{R-Fwd}.)
We shall show that $Q' \in B_1$. Let us consider/recall the provenance of its different components:
\begin{enumerate}
\item[(a)] $G_2 \in  \chrpk{\Gamma_2}{}$
\item[(b)] Since $Q_v$ stands for $Q_t[w/t][v/w]$, we have $Q_v \in \chrp{T}{v}$.
\item[(c)] $P^*_1\in \encCP{{{\Gamma_1}, x:{S}}\s P_1}$.
\item[(d)] $Q_x \in \chrp{\dual S}{x}$
\end{enumerate}
By \Cref{prop:charp}, Items (a) and (b) above entail:
\begin{enumerate}
\item[(e)] $G_2 \para Q_v  \in  \chrpk{\Gamma_2, v:T}{}$
\end{enumerate}
In turn, by considering Items (c), (d), and (e), together with \Cref{eq:52b1}, it is immediate to see that $Q' \in B_1$. Therefore, $Q' \in \encCP {\Gamma\s P'}$, as desired.


\medskip
\noindent \textbf{\emph{Sub-case (ii).}}
Let $Q = G_1 \para
\res y \big(y\inp {t}.P^*_2 \pp P_y\big)$ 
be an arbitrary process in $A_2$. 
Since $P_y \in \chrp{\oc T.S}{y}$, 
by \Cref{def:charprocess} we have 
 $P_y = \bout{y}{k}.(Q_k \para Q_y)$, where $Q_k \in \chrp{\dual{T}}{k}$ and $Q_y \in \chrp{S}{y}$.
 Considering this, and by 
applying Rule~\Did{R-ChCom} (cf. \Cref{fig:redlcp}),
we have: 
$$
\begin{array}{rllll}
Q\ =
&G_1 \para
\res y \big(y\inp {t}.P^*_2 \pp P_y\big)\\[2mm]
=& G_1 \para
\res y \big(y\inp {t}.P^*_2 \pp \bout{y}{k}.(Q_k \para Q_y)\big)\\[2mm]
\to &G_1 \pp \res y  \big(\res k (P^*_2[k/t] \pp Q_k) \para Q_y \big) \ \defeq\ Q''
\end{array}
$$
We shall show that $Q'' \uptok R$, for some $R \in B_2$. Let us consider/recall the provenance of its different components:
\begin{enumerate}
\item[(a)] $G_1 \in  \chrpk{\Gamma_1, v:T}{}$
\item[(b)] By \Cref{lem:enc_subst}, $P^*_2[k/t]\in \encCP{\Gamma_2,y:\dual S, k:T\s P_2\substj kt}$.
\item[(c)] $Q_k \in \chrp{\dual{T}}{k}$ 
\item[(d)] $Q_y \in \chrp{S}{y}$
\end{enumerate}
Furthermore, we can infer:
\begin{enumerate}
\item[(e)] From (a) and \Cref{prop:charp}, there must exist $G^*_1$ and $G_v$ such that $G_1 = G^*_1 \pp G_v$,  $G^*_1 \in \chrp{\Gamma_1}{}$, and $G_v\in \chrp{T}{v}$.  
\item[(f)] By combining (b), (c) and (d), together with  \Cref{prop:charp}, we have that $ \res y  \big(\res k (P^*_2[k/t] \pp Q_k) \pp Q_y \big) \in \chrp{\Gamma_2}{}$.
\end{enumerate}

%
\noindent
Given this, we can rewrite $Q''$ as follows:
\begin{equation}
Q'' = G^*_1  \pp G_v \pp \res y  \big(\res k (P^*_2[k/t] \pp Q_k) \para Q_y \big) 
\label{eq:522}
\end{equation}
We now consider an arbitrary  $R \in B_2$. 
By \Cref{eq:52b2}, we have that 
\begin{equation}
	R\defeq G'_1 \para \res y\big(Q_2 \pp P'_y \big)
	\label{eq:523}	
\end{equation}
with
\begin{enumerate}
	\item[(a$'$)] $G'_1 \in \chrp{\Gamma_1}{}$
	\item[(b$'$)] $Q_2 \in \encCP{\Gamma_2,y:{\dual S}, v:T\s P_2\substj vt}$
	\item[(c$'$)] $P'_y \in\chrp{S}y$
\end{enumerate}
By \Cref{lem: wt_characteristic} and (c$'$), we can infer:
$P'_y \cp y: \encob{S}{\ct}$;
and by \Cref{thm:L0-L2} and (b$'$), we can infer:
$Q_2 \cp \encob{\Gamma_2,y:{\dual S}, v:T}{\ct}$.
Finally, we have:
\begin{enumerate}
	\item[(d$'$)] $\res y\big(Q_2 \pp P'_y \big) \cp \encob{\Gamma_2}{\ct}, \encob{v:T}{\ct}$
\end{enumerate}
We now compare $Q''$ and $R$ (as in \Cref{eq:522} and \Cref{eq:523}, respectively).
By \Cref{lem: wt_characteristic} and (e) and (f) above, 
it is easy to see that
$Q'', R\cp \encob{\Gamma_1}{\ct},\encob{\Gamma_2}{\ct}, \encob{v:T}{\ct}$.
Then, by \Cref{def:uptok}, we have  that $Q'' \uptok R$.
Therefore, there is an $R$ 
such that $Q'' \uptok R$, with 
$R \in \encCP {\Gamma\s P'}$, as desired. This concludes the analysis for Case \Did{R-Com}.

 \bigskip
 
  
\item
Case \Did{R-Case}:
$$P\defeq \res {xy:S'}( \selection x{l_j}.Q \pp \branching ylR) 
		\quad \to\quad 
		\res {xy:S''}(Q\pp R_j) \defeq P'$$
Since $\Gamma\s P$, then by inversion 
$S'= \select lS$ for some  $S_i$, with $i\in I$.
For simplicity,  let us write $T_i$ to denote the dual of any ${S_i}$.
As a result of the reduction, we have $S''=S_j$ for some $j\in I$.
Again by inversion we have the following derivation:
$$
\inferrule*[Right = \Did{T-Res}]
{
\inferrule*[Right=\Did{T-Par}]
{
\inferrule[\Did{T-Sel}]
	{
		\Gamma_1,  x: S_j \s Q \quad   \exists j \in I
	}
	{ 	
		\Gamma_1, x :  \select lS \s \selection x {l_j}.Q   
	}
\and
\inferrule[\Did{T-Brch}]
	{
		\Gamma_2, y: T_i \s R_i \quad  \forall i\in I
	}
	{
		\Gamma_2, y :\branch l{T} \s \branching ylR
	}
}
{
(\Gamma_1, x :  \select lS) ,  (\Gamma_2, y :\branch lT) \s \selection x {l_j}.Q \pp \branching ylR
}
}
{\Gamma_1,  \Gamma_2 \s \res {xy:S'}( \selection x{l_j}.Q \pp \branching ylR) }
$$
By \Cref{def:typed_enc} the translation of $P$ is as follows:
\begin{align*}
\encCP{\Gamma \s P} = \ &
 \encCP{\Gamma\s \res {xy:S'}(\selection x{l_j}.Q \pp \branching ylR)}\\[2mm]
= \ & \underbrace{\big\{
C_1[Q_1] \para G_2  \suchthat  
C_1 \in \mathcal{C}_{x:\branch l{T}},\,
Q_1 \in \encCP{\Gamma_1, x :  \select lS \s \selection x {l_j}.Q}, \, G_2 \in  \chrpk{\Gamma_2}{}  \big\}}_{A_1}
\\
& ~\cup
\\	
&\underbrace{\big\{G_1 \para C_2[Q_2]   \suchthat 
C_2 \in \mathcal{C}_{y:\select lS},\,
Q_2 \in \encCP{{\Gamma_2,{y:\branch lT}}\s\branching ylR}, \,
G_1 \in \chrp{\Gamma_1}{}
\big\}}_{A_2}
\end{align*}
where:
\begin{align}
A_1 =\ &\big\{
G_2\pp   \res x \big(\selection x{l_j}.Q^* \pp P_x\big) 
  \suchthat  \nonumber
\\
& \qquad \qquad  
Q^* \in {\encCP{{\Gamma_1,x:S_j}\s Q}},\,
P_x \in \chrp{\branch lT}{x}, \,
G_2 \in  \chrpk{\Gamma_2}{}  \big\} 
\label{eq:52a3}
\\[1mm]									
A_2 = & \big\{G_1 \para   \res y \big( \parbranching y{l_i}{R^*_i} \pp P_y \big)  
 \suchthat \nonumber \\ 
& \qquad \qquad
R^*_i \in {\encCP{{\Gamma_2,y:T_i}\s R_i}},\,
P_y \in\chrp{\select lS}y,\,
G_1 \in \chrp{\Gamma_1}{}
\big\} 
\label{eq:52a4}
\end{align}

\noindent
Before spelling out the translation of $P'$, we record some considerations.
By \Cref{thm:L0-L2},
 $M\in \encCP{{\Gamma}\s P}$  implies $M\cp \encob{\Gamma}{\ct}$.
Since $P\to P'$, then by \Cref{thm:subj_red} we have $\Gamma\s P'$.
Again, by \Cref{thm:L0-L2}, 
$O\in \encCP{{\Gamma}\s P'}$ implies $O\cp \encob{\Gamma}{\ct}$.
By \Cref{def:typed_enc} the encoding of $P'$ is as follows:
\begin{align*}
\encCP{\Gamma \s P'} = \ &
\encCP{\Gamma\s \res {xy:S_j}(Q\pp R_j)}\\[2mm]
 = \  &  
\underbrace{\big\{G'_2 \pp C'_1 \big[Q^*\big]  \suchthat 
C'_1 \in \mathcal{C}_{x: {T_j}},\,
Q^* \in {\encCP{{\Gamma_1, {x:S_j}}\s Q}},\,
G'_2\in  \chrp{\Gamma_2}{}
\big\}}_{B_1}
\\
& ~\cup
\\
& \underbrace{\{ G'_1 \pp C'_2\big[R^*_j\big]  \suchthat 
C'_2 \in \mathcal{C}_{y: S_j},\,
R^*_j \in {\encCP{{\Gamma_2,{y:{T_j}}}\s R_j}},\,
G'_1 \in \chrp{\Gamma_1}{} \}}_{B_2}
\end{align*}
where:
\begin{align}
B_1 =\ &
\big\{G'_2 \pp \res x(Q^* \pp P'_x)  \suchthat 
Q^* \in {\encCP{{\Gamma_1, {x:S_j}}\s Q}},\,
P'_x \in \chrp{T_j}{x},\,
G'_2\in  \chrp{\Gamma_2}{}
\big\} 
\label{eq:52b3}
\\[1mm]
B_2 =\ &\{ G'_1 \pp \res y(R^*_j \pp P'_y)  \suchthat 
R^*_j \in {\encCP{{\Gamma_2,{y:{T_j}}}\s R_j}},\,
P'_y \in\chrp{S_j}y,\,
G'_1 \in \chrp{\Gamma_1}{}
\} 
\label{eq:52b4}
\end{align}
%
%

{We now show} that every process in $\encCP {\Gamma\s P}$ reduces into a process in $\encCP {\Gamma\s P'}$. 
 We address two distinct sub-cases:
\begin{enumerate}
\item[(i)] We show that every $Q \in A_1$ (cf. \Cref{eq:52a3}) reduces to a $Q' \in B_1$ (cf. \Cref{eq:52b3});
\item[(ii)] Similarly, we show that every   $Q \in A_2$ (cf. \Cref{eq:52a4}) reduces to a  $Q'' \in B_2$ (cf. \Cref{eq:52b4}).
\end{enumerate}

\medskip
\noindent \textbf{\emph{Sub-case (i)}}. 
Let $Q = G_2\pp   \res x \big(\selection x{l_j}.Q^* \pp P_x\big)$ be an arbitrary process in $A_1$,
with
 $Q^* \in {\encCP{{\Gamma_1, {x:S_j}}\s Q}}$
and
$P_x \in \chrp{\branch lT}{x}$.
By \Cref{def:charprocess}, we may then
let $P_x = \branching xlP$, such that  
$P_i\in \chrp{T_i}x$, for all $i\in I$.
By applying Rule \Did{R-ChCase} (cf. \Cref{fig:redlcp}),
we have: 
$$
\begin{array}{rllll}
Q\ \defeq\
&G_2\pp   \res x \big(\selection x{l_j}.Q^* \pp P_x\big) \\[2mm]
\to
&G_2\pp   \res x \big( Q^* \pp P_j \big) \defeq Q'
\end{array}
$$
We shall show that $Q' \in B_1$. Let us consider/recall the provenance of its different components:
\begin{enumerate}
\item[(a)] $G_2 \in  \chrpk{\Gamma_2}{}$
\item[(b)] $Q^* \in {\encCP{{\Gamma_1, {x:S_j}}\s Q}}$.
\item[(c)] $P_j\in \chrp{T_j}x$.
\end{enumerate}
By considering Items (a) -- (c), together with \Cref{eq:52b3}, it is immediate to see that $Q' \in B_1$. Therefore, $Q' \in \encCP {\Gamma\s P'}$, as desired.


\medskip
\noindent \textbf{\emph{Sub-case (ii)}}.
Let $Q = G_1 \para   \res y \big( \parbranching y{l_i}{R^*_i} \pp   P_y \big)$ be an arbitrary process in $A_2$, with 
$R^*_i \in {\encCP{{\Gamma_2,y:T_i}\s R_i}}$
and
$P_y \in \chrp{\select lS}{y}$.
By \Cref{def:charprocess}, 
  $P_y$ is one of the processes in the union $\bigcup_{i\in I}\big\{ \selection y{l _i}.{P_i} \suchthat P_i\in \chrp{S_i}y\big\}$. 
We then choose
$P_y = \selection y{l_j}.{P_j}$ 
such that 
$j\in I$ and $P_j\in \chrp{S_j}y$.
By applying Rule \Did{R-ChCase} (cf. \Cref{fig:redlcp}), 
we have:
$$
\begin{array}{rllll}
Q\ \defeq\ &
G_1 \para   \res y \big( \parbranching y{l_i}{R^*_i} \pp  \selection y{l_j}.{P_j} \big)
\\[2mm]
\to &
G_1 \para   \res y \big({R^*_j} \pp {P_j} \big) \defeq\ Q''
\end{array}
$$
We shall show that $Q'' \in B_2$. Let us consider/recall the provenance of its different components:
\begin{enumerate}
\item[(a)] $G_1 \in  \chrpk{\Gamma_1}{}$
\item[(b)] $R^*_j \in {\encCP{{\Gamma_2,{y:{T_j}}}\s R_j}}$.
\item[(c)] $P_j\in \chrp{S_j}y$.
\end{enumerate}
By considering Items (a) -- (c), together with \Cref{eq:52b4}, it is immediate to see that $Q'' \in B_2$. Therefore, $Q'' \in \encCP {\Gamma\s P'}$, as desired.
This concludes the analysis for Case~\Did{R-Case} (and for \Cref{5.2-1}).

\end{enumerate}

\bigskip

\noindent{\underline{\textbf{\Cref{5.2-2}.}}}
Let $P\in \lkoba$
and 
$Q \in \encCP{{\Gamma}\s P}$. Suppose  that $Q \to\fred Q'$; we now show that 
  (i)~there is a
$P'$ such that 
 $P\to P'$ and
 (ii) 
$Q' \uptok R$, for some $R \in \encCP{{\Gamma}\s P'}$.

We first argue for (i), i.e., the existence of a reduction $P \to P'$.
Notice that for the reduction(s) $Q \to\fred Q'$ to occur, 
there must exist two complementary prefixes occurring at top level in $Q$.
By \Cref{def:typed_enc} and, crucially, by the assumption $P\in \lkoba$,
the same two prefixes occur at top-level also in $P$ and can reduce; 
indeed, without the assumption $P\in \lkoba$, it could be that $P$ cannot reduce due to a deadlock.
Hence, 
 $P$ can mimic the reduction from $Q$, up to structural congruence: $P \to P'$.
 It then suffices to prove the theorem for $M \equiv P$, in which the two prefixes involved occur in contiguous positions and can reduce. 
  
To address (ii), we now relate $P'$ and $Q'$ by considering two main cases for the reduction originating from $Q$ (cf. \Cref{fig:redlcp}): (1) it 
corresponds to an input-output communication via Rule~\Did{R-ChCom};
and (2) it corresponds to a selection-branching interaction via Rule~\Did{R-ChCase}.
(The third case, corresponding to Rule~\Did{R-ChRes}, is straightforward using the induction hypothesis.) 
We detail these two cases; the analysis largely mirrors the one given in \Cref{5.2-1}:

%

\begin{enumerate}
\item
Case
$M= \ctx{E}{}{\res {xy:S'}(\overline x\out v.P_1\pp y\inp {t:T}.P_2)}\to 
\ctx E{}{\res {xy:S}(P_1\pp P_2[v/t])}\defeq P'$.
\\

\noindent
Since $\Gamma\s P$, by \Cref{thm:subj_cong} $\Gamma\s M$.
By inversion, 
$S'=\oc T.S$ for some $S$.
Then, again by inversion
$\Gamma = (\Gamma_1,v:T),  \Gamma_2$.
By \Cref{def:typed_enc}, 
$ \encCP {\Gamma\s M}= A_1 \cup A_2$, 
where:
\begin{align*}
A_1 =\ & \big\{
\ctx{F}{}{
G_2 \para
\res x\big({\dual{x}(w). \big(\linkr{v}{w}\para P^*_1\big)} \pp P_x\big)}  \suchthat  
P^*_1 \in \encCP{{{\Gamma_1}, x:{S}}\s P_1},\,
P_x \in \chrp{\wn T.\dual S}{x},\,
G_2 \in  \chrpk{\Gamma_2}{}  \big\} 
\\[1mm]									
A_2 =\ & \big\{
\ctx{H}{}{G_1 \para
\res y \big(y\inp {t}.P^*_2 \pp P_y\big)}   \suchthat 
P^*_2 \in \encCP{\Gamma_2,y:\dual S, t:T\s P_2},\,
P_y \in\chrp{\oc T.S}y,\,
G_1 \in \chrp{\Gamma_1, v:T}{}
\big\} 
\end{align*}
Also by \Cref{def:typed_enc}, we have that $\encCP {\Gamma\s P'} = B_1 \cup B_2$, where:
\begin{align*}
B_1 =\ & \big\{
\ctx {F}{}{G'_2 \para \res x\big(Q_1 \pp P'_x \big)}  \suchthat  P'_x \in \chrp{\dual S}{x},\,
Q_1 \in \encCP{{{\Gamma_1}, x:{S}}\s P_1}, \, G'_2 \in  \chrpk{\Gamma_2, v:T}{}\big\} 
\\[1mm]									
B_2 =\ & \big\{
\ctx H{}{G'_1 \para \res y\big(Q_2 \pp P'_y \big)}   \suchthat  P'_y \in\chrp{S}y,\,
Q_2 \in \encCP{\Gamma_2,y:{\dual S}, v:T\s P_2\substj vt}, \,
G'_1 \in \chrp{\Gamma_1}{}
\big\}  
\end{align*}
%

\noindent We now address two distinct sub-cases:
\begin{enumerate}
\item[(i)] We show that every   $Q \in A_1$  reduces to a   $Q'  \in B_1$;
\item[(ii)] Similarly, we show that every  $Q \in A_2$   reduces to a   $Q'$ such that $Q' \uptok R$ and $R \in B_2$.
\end{enumerate}

 \medskip
\noindent \textbf{\emph{Sub-case (i)}}.
Let $Q =
\ctx F{}{G_2 \para
\res x\big({\dual{x}(w). \big(\linkr{v}{w}\para P^*_1\big)} \pp P_x\big)} $ be an arbitrary process in $A_1$, with $P_x \in \chrp{\wn T.\dual S}{x}$.
By \Cref{def:charprocess} we have that 
$$\chrp{\wn T.\dual S}{x} = {\big\{x(t).(Q_t \para Q_x) \suchthat Q_t \in \chrp{T}{t} \land Q_x \in \chrp{\dual S}{x}\big\}}$$
We may then let $P_x = x(t).(Q_t \para Q_x)$ where $Q_t \in \chrp{T}{t}$ and 
$Q_x \in \chrp{\dual S}{x}$.
By applying Rules \Did{R-ChCom} and \Did{R-Fwd} (cf.~\Cref{fig:redlcp}) we have: 
$$
\begin{array}{rllll}
Q\ \defeq &
\ctx F{}{G_2 \para
\res x\big({\dual{x}(w). \big(\linkr{v}{w}\para P^*_1\big)} \pp P_x\big)}\\[2mm]
= &
\ctx F{}{G_2 \pp \res x\big({\dual{x}(w). \big(\linkr{v}{w}\para P^*_1 \big)} \pp x(t).(Q_t \para Q_x)\big)}\\[2mm]
\to& \ctx F{}{G_2 \pp \res x \big( \res w \big( \linkr{v}{w}\para P^*_1 \pp Q_w \pp Q_x  \big) \big)}\\[2mm]
\fred&
\ctx F{}{G_2   \pp Q_v \pp \res x \big( P^*_1 \pp Q_x \big)} \defeq Q'
\end{array}
$$
It is then easy to see that $Q' \in \encCP{\Gamma \s P'}$.
 
 \medskip
\noindent \textbf{\emph{Sub-case (ii)}}.
Let $Q = \ctx H{}{
G_1 \para
\res y \big(y\inp {t}.P^*_2 \pp P_y\big) }$ be an arbitrary process in $A_2$.
By \Cref{def:charprocess},
since $P_y \in \chrp{\oc T.S}{y}$, we may then  let $P_y = \bout{y}{k}.(Q_k \para Q_y)$ where $Q_k \in \chrp{\dual{T}}{k}$ and $Q_y \in \chrp{S}{y}$.
By applying Rule \Did{R-ChCom} (cf.~\Cref{fig:redlcp}),
we have: 
$$
\begin{array}{rllll}
Q\ \defeq
& \ctx H{}{
G_1 \para
\res y \big(y\inp {t}.P^*_2 \pp P_y\big) }\\[2mm]
=& \ctx H{}{
G_1 \para
\res y \big(y\inp {t}.P^*_2 \pp \bout{y}{k}.(Q_k \para Q_y)\big)}\\[2mm]
\to &
\ctx H{}{G_1 \pp \res y  \big(\res k (P^*_2[k/t] \pp Q_k) \para Q_y \big)} \defeq Q'
\end{array}
$$
By \Cref{lem:enc_subst},  $P^*_2[k/t]\in \encCP{\Gamma_2,y:\dual S, k:T\s P_2\substj kt}$.
Also, by \Cref{prop:charp} 
we can rewrite $G_1$ as
$G_1 = G^*_1 \pp G_v$, such that $G^*_1 \in \encob{\Gamma_1}{\ct}$ and $G_v\in \encob{v:T}{\ct}$.
Then, we can rewrite $Q'$ as follows:
$$
Q' = \ctx H{}{G^*_1 \pp G_v \pp \res y  \big(\res k (P^*_2[k/t] \pp Q_k) \para Q_y \big)}
$$
It is then easy to see that 
 $Q' \uptok R$, with 
 $R \in \encCP{{\Gamma}\s P'}$, as wanted.

\medskip

\item
Case $M= \ctx{E}{}{\res {xy:S'}( \selection x{l_j}.Q \pp \branching ylR)} \to 
		\ctx E{}{\res {xy:S''}(Q\pp R_j)} \defeq P'$. \\

\noindent Since $\Gamma\s P$, by \Cref{thm:subj_cong} also $\Gamma\s M$.
By inversion let
$S'= \select lS$ for some session types $S_i$ ($i\in I$).
By \Cref{def:typed_enc}, $\encCP{\Gamma \s M} = A_1 \cup A_2$, where:
\begin{align*}
A_1 =\ & \big\{
\ctx{F}{}{
G_2\pp   \res x \big(\selection x{l_j}.Q^* \pp P_x\big)}   \suchthat  
Q^* \in {\encCP{{\Gamma_1,x:S_j}\s Q}},\,
P_x \in \chrp{\branch lT}{x}, \,
G_2 \in  \chrpk{\Gamma_2}{}  \big\} 
\\[1mm]									
A_2 =\ & \big\{
\ctx{H}{}{
G_1 \para   \res y \big( \parbranching y{l_i}{R^*_i} \pp P_y \big)}   \suchthat 
R^*_i \in {\encCP{{\Gamma_2,y:T_i}\s R_i}},\,
P_y \in\chrp{\select lS}y,\,
G_1 \in \chrp{\Gamma_1}{} 
\big\} \ 
\end{align*}
Also, by \Cref{def:typed_enc}, we have that $\encCP{\Gamma \s P'} = B_1 \cup B_2$, where:
\begin{align*}
B_1 =\ &  \big\{
\ctx F{}{
G'_2 \pp \res x(Q^* \pp P'_x)}  \suchthat 
Q^* \in {\encCP{{\Gamma_1, {x:S_j}}\s Q}},\,
P'_x \in \chrp{T_j}{x},\,
G'_2\in  \chrp{\Gamma_2}{} 
\big\} 
\\[1mm]
B_2 = \ &\big\{
\ctx H{}{
G'_1 \pp \res y(R^*_j \pp P'_y)}  \suchthat 
R^*_j \in {\encCP{{\Gamma_2,{y:{T_j}}}\s R_j}},\,
P'_y \in\chrp{S_j}y,\,
G'_1 \in \chrp{\Gamma_1}{}
\big\} 
\end{align*}

\noindent As before, we now address two distinct sub-cases:
\begin{enumerate}
\item[(i)] We show that every  $Q \in A_1$  reduces to a  $Q'  \in B_1$;
\item[(ii)] Similarly, we show that every  $Q \in A_2$   reduces to a  $Q'  \in B_2$.
\end{enumerate}

 \medskip
\noindent \textbf{\emph{Sub-case (i)}}.
Let $Q =\ctx F{}{G_2\pp   \res x \big(\selection x{l_j}.Q^* \pp P_x\big)}$ be an arbitrary process in $A_1$.
By \Cref{def:charprocess}, since
$P_x \in \chrp{\branch lT}{x}$, then  $P_x = \branching xlP$ with 
$P_i\in \chrp{T_i}x$, for all $i\in I$.
By applying Rule~\Did{R-ChCase} (cf.~\Cref{fig:redlcp}), and letting $j\in I$
we have: 
$$
\begin{array}{rllll}
Q\ \defeq\
&
\ctx F{}{G_2\pp   \res x \big(\selection x{l_j}.Q^* \pp P_x\big)} \\[2mm]
\to
&
\ctx F{}{G_2\pp   \res x \big( Q^* \pp P_j \big)} \defeq Q'
\end{array}
$$
It is then easy to see that $Q' \in B_1$, and therefore $Q' \in \encCP{\Gamma \s P'}$, as desired.

 \medskip
\noindent \textbf{\emph{Sub-case (ii)}}.
Let $Q= \ctx H{}{G_1 \para   \res y \big( \parbranching y{l_i}{R^*_i} \pp  \selection y{l_j}.{P_j} \big)}$ be an arbitrary process in $A_2$, with
$R^*_i \in {\encCP{{\Gamma_2,y:T_i}\s R_i}}$.
By \Cref{def:charprocess},
since $P_y \in \chrp{\select lS}{y} =  \bigcup_{i\in I}\big\{ \selection y{l _i}.{P_i} \suchthat P_i\in \chrp{S_i}y\big\}$,
it means that $P_y$ is one of the processes in the union.
We may then choose
$P_y = \selection y{l_j}.{P_j}$ 
such that 
$j\in I$ and $P_j\in \chrp{S_j}y$.
By applying Rule~\Did{R-ChCase} (cf.~\Cref{fig:redlcp})
we have:
$$
\begin{array}{rllll}
Q\ \defeq &
\ctx H{}{G_1 \para   \res y \big( \parbranching y{l_i}{R^*_i} \pp  \selection y{l_j}.{P_j} \big)}
\\[2mm]
\to &
\ctx H{}{G_1 \para   \res y \big({R^*_j} \pp {P_j} \big)} \defeq Q'
\end{array}
$$
Since $P_j\in \chrp{S_j}y$ and $P'_y \in\chrp{S_j}y$,
it is then easy to see that $Q' \in B_2$, and therefore $Q' \in \encCP{\Gamma \s P'}$, as desired.
This concludes the analysis for  this case (and for \Cref{5.2-2}).
\end{enumerate}
\end{proof}

%% file: appendix_vp.tex

\section{Details on the Second Translation (\Cref{ss:oprw})}
\label{sec:appendix-6}

\subsection{A Session Type System for Value Dependencies}
\label{ss:vd}
We define an extension of the session type system in \Cref{s:sessions}, in which  typing judgments are of the form $$\Gamma \compsi\Psi \s P$$
where the typing 
context $\Gamma$ is as before and the new 
\emph{dependency context}  $\Psi$  contains triples $(a, x, n)$ and $\langle a, x, n\rangle$, with $n \geq 0$: while the first triple
denotes an input of object $x$ along subject $a$, the latter denotes output prefixes.
As a result,  
the dependency context $\Psi$
   describes all communication prefixes in $P$ and their distance to top-level.
Indeed, the $n$ in $(a, x, n)$ and $\langle a, x, n\rangle$  helps us to track the position of the process prefix in relation to its session type, and to describe several value dependencies along a certain name.
To this end, we consider \emph{annotated} output and input session types, written $\oc^n S.T$ and $\wn^n S.T$ (with $n \geq 0$), respectively; other types remain unchanged, for they do not involve passing of values.

In the following, we use $\Psi, \Phi, \ldots$ to denote dependency contexts;
we sometimes write $\Psi(\wt{x})$ to make explicit that $\Psi$ concerns names in the sequence $\wt{x}$ (and similarly for $\Gamma(\wt{x})$).

\Cref{fig:sess_typingmod} gives the typing rules required for supporting $\Psi$ and annotated session types; they make use of the following auxiliary notions.
\begin{definition}
\label{def:myplus}
Given an annotated session type $S$ and contexts $\Psi$ and $\Gamma$, 
operations
\myplus{S} and \myplus{\Psi} and \myplus{\Gamma}   increment annotations in types and in triples:
\begin{align*}
\myplus{(\oc^n S.T)} &\defeq  \oc^{n+1} S.\myplus{T} & \myplus{(\Psi, (a, x, n))} &\defeq  \myplus{\Psi}, (a, x, n+1)
\\
\myplus{(\wn^n S.T)} &\defeq  \wn^{n+1} S.\myplus{T} & \myplus{(\Psi, \langle a, x, n\rangle)} &\defeq  \myplus{\Psi}, \langle a, x, n+1\rangle
\\
\myplus{{\branch lS}} &\defeq  \branch l{\myplus{S}} & \myplus{(\Gamma, x:T)} &\defeq  \myplus{\Gamma}, x:\myplus{T}
\\
 \myplus{{\select lS}} &\defeq   \select l{\myplus{S}} &\myplus{(\,\cdot\,)} &\defeq  \cdot
 \\ \myplus{\nilT} &\defeq  \nilT
\end{align*}
\end{definition}


\begin{figure}[t]
\begin{mdframed}
$$\begin{array}{c}
%
\inferrule*[right=\Did{T-NilD}]
{\eend{\Gamma}}
{\Gamma \compsi \emptyset \s \nil}
\qquad
\inferrule*[right=\Did{T-ParD}]
{\Gamma_1 \compsi\Psi \s P\quad
\Gamma_2 \compsi\Phi \s Q}
{\Gamma_1,\Gamma_2 \compsi\Psi, \Phi \s P\pp Q}
\\\\
%
\inferrule[{\Did{T-InD}}]
	{
 	  \Gamma, x:S, y:T \compsi\Psi \s P
	}
	{
	  \myplus{\Gamma}, x : \wn^0 T.\myplus{S} \compsi \myplus{\Psi}, (x,y,0)\s x\inp{y}.P
	}
\qquad
\inferrule[{\Did{T-OutD}}]
	{
	   \Gamma, x:S \compsi\Psi \s P
	}
	{
	   \myplus{\Gamma},  x : \oc^0 T.\myplus{S}, y:T \compsi \myplus{\Psi},  \langle x,y,0\rangle  \s \ov x\out{y}.P
	}
\\\\
\inferrule[{\Did{T-BrchD}}]
	{
		\Gamma, x: S_i \compsi\Psi \s P_i \quad  \forall i\in I
	}
	{
		\myplus{\Gamma}, x:\branch lS \compsi\Psi \s \branching {x}lP
	}
\qquad
\inferrule[{\Did{T-SelD}}]
	{
		\Gamma,  x: S_j \compsi\Psi \s P \quad   \exists j \in I
	}
	{ 	
		\myplus{\Gamma}, x:  \select lS \compsi\Psi \s \selection {x} {l_j}.P
	}
	\\\\
%
\inferrule[{\Did{T-ResD}}]
	{
	    \Gamma, x:T, y:\ov{T} \compsi\Psi\s P
	    \and
	    \Psi\decomp{x,y} \Psi\proj x,\ \Psi\proj y,\ \Psi'
	}
	{ 
	     \Gamma \compsi\Psi' \s \res{xy} P
	}
  \end{array}$$
  \end{mdframed}
\caption{Typing rules for the \p calculus with session types with value dependency tracking.}
\label{fig:sess_typingmod}
\end{figure}
We discuss the new typing rules for sessions with value dependencies, given in \Cref{fig:sess_typingmod}.
The most interesting cases
are \Did{T-InD} and \Did{T-OutD}, which record communication subjects and objects as they occur in the input and output processes, respectively.
The number annotation  tracks the ``distance'' from top-level, which in the case of top-level input or output is simply $0$. The annotations related to the continuation process $P$ are incremented, as captured by $\myplus{S}$ and $\myplus{\Psi}$ (cf. \Cref{def:myplus}).
Rule $\Did{T-ParD}$ simply combines the prefix information in the two parallel sub-processes.

To formalize Rule \Did{T-ResD}, we need a couple of auxiliary definitions.
We   write $\Psi\proj x$ to denote the result of ``projecting'' $\Psi$ onto name $x$,  giving the set containing \emph{all} tuples in which 
$x$ appears free, namely as input subject, output subject or output object.
Formally:
\begin{definition}[Projection of $\Psi$]
\label{def: proj_psi}
Let $\Psi$ be a context, and $x$ a name.
$$
\Psi\proj x \ \defeq\
\{
(x,w,n) \in \Psi
\ \vee\
\langle x, w, n\rangle \in \Psi
\ \vee\
\langle w, x, n\rangle \in \Psi,
\text{ for some name } w \text{ and } n\geq 0
\}
$$
\end{definition}
With this notion, we now define the decomposition of context $\Psi$ under names $x$ and $y$:
\begin{definition}[Decomposition of $\Psi$]
\label{def: decomp_psi}
Let $\Psi$ be a context, and $x$ and $y$ names.
The \emph{decomposition of $\Psi$ under $x$ and $y$}  is defined as follows:
$$
\Psi\decomp{x,y} \Psi\proj x,\ \Psi\proj y,\ \Psi'
$$
such that
$
\Psi' = (\Psi\setminus \Psi\proj x) \cup (\Psi\setminus\Psi\proj y)
$.
\end{definition}

Rule \Did{T-ResD} decomposes $\Psi$ using \Cref{def: decomp_psi}: context $\Psi$, used to type process $P$, is decomposed under names $x$ and $y$; the remaining context $\Psi'$ is used to type process $\res{xy}P$. Thus, intuitively, since $\Psi$ is used to track \emph{free} names in $P$, then triples where $x$ and $y$ occur free must be removed from $\Psi$ in order to type $\res{xy}P$.
Remaining rules are self-explanatory.

It is immediate to see
that main results for the type system in \Cref{s:sessions}
(Theorems~\ref{thm:subj_red} and~\ref{thm:safety}) extend  to the type system with $\Psi$ and annotations in input and output session types.

\begin{example}[The Extended Type System at Work]
\label{ex:track}
Using the rules in \Cref{fig:sess_typingmod},
we have the following typing derivation for 
process 
$\ov{a}\out{\unit}.b\inp u.\ov{c}\out{u}.\nil$ 
(similar to sub-process $\ov{a_0}\out{\unit}.a_1\inp u.\ov{a_2}\out{u}.\nil$
from \Cref{ex:vd1}):
$$
\infer[\Did{T-OutD}]
{\underbrace{a: \oc^{0}\nilT.\nilT, b: \wn^1 U.\nilT, c: \oc^2 U.\nilT}_{\Gamma} \compsi \underbrace{\langle a, \unit, 0\rangle, (b,u, 1), \langle c, u, 2\rangle}_{\Psi}  \s  \ov{a}\out{\unit}.b\inp u.\ov{c}\out{u}.\nil }
{
\infer[\Did{T-InD}]
	{
	a: \nilT, b: \wn^0 U.\nilT, c: \oc^1 U.\nilT \compsi (b,u, 0), \langle c, u, 1\rangle \s b\inp u.\ov{c}\out{u}.\nil
	}
	{
	\inferrule*[Right= \Did{T-OutD}]{
	\inferrule*[Right= \Did{T-NilD}]{
	}
	{
	a: \nilT, b: \nilT, c: \nilT \compsi  \emptyset \s \nil
	}
	}{
 	a: \nilT,  b: \nilT, c: \oc^0 U.\nilT, u:U \compsi \langle c, u, 0 \rangle \s \ov{c}\out{u}.\nil}
	}
}
$$
where $U \defeq \wn \nilT.\nilT$.
Using $\Psi$, the dependency between sessions $c$ and $b$ can be read as: 
``the output on $c$, which is two prefixes away from top-level, depends on the value received on $b$, which is one prefix away''.
\end{example}

\subsection{Value Dependencies}
Given a typed process $\Gamma \compsi\Psi \s P$, we now define the \emph{value dependencies} 
induced by  prefixes in $P$, but not captured by $\Gamma$:


\begin{definition}[Set of Value Dependencies]
\label{def:veep}
Let $P$ be a process such that $\Gamma \compsi\Psi \s P$. 
The \emph{value dependencies} of $P$ are given by the set $\vdepset{\Psi}$, which is defined as 
$$
\vdepset{\Psi} \defeq \big\{\vdep{x^n}{z^m} ~\suchthat  \exists y. (x,y,n) \in \Psi \land \langle z, y, m\rangle \in \Psi \big\}
$$
We sometimes write 
$\vdep{x^n}{z^m}:S$ 
when the dependency involves a value of type $S$.
We write 
$\dom(\Psi)$ 
to stand for the set 
$\{x  \suchthat  \vdep{x^n}{z^m} \in \vdepset{\Psi} \lor \vdep{z^m}{x^n} \in \vdepset{\Psi}\}$.
\end{definition}

Thus, the pair $\vdep{a^n}{b^m}$ denotes a 
  value dependency of an (output) prefix along session $b$ on an (input) prefix along session $a$. 

\begin{example}  
Consider process 
${\Gamma} \compsi {\Psi}  \s  \ov{a}\out{\unit}.b\inp u.\ov{c}\out{u}.\nil$
from 
\Cref{ex:track}. 
We can now track the value dependency between $c$ and $b$, because $\vdepset{\Psi} = \{\vdep{b^1}{c^2}\}$. 
Similarly, sub-process 
$\ov{a_0}\out{\unit}.a_1\inp u.\ov{a_2}\out{u}.\nil$
from 
\Cref{ex:vd1}
has a value dependency $\vdep{a_{1}^1}{a_{2}^2}$.
\end{example}

%

We now develop some auxiliary definitions, which are required to 
define characteristic processes and catalyzers that exploit value dependencies.

We may interpret $\vdepset{\Psi}$ as a binary relation 
on typing assignments in $\Gamma$,
in which 
each  
$\vdep{w_j^n}{w_k^m} \in \vdepset{\Psi}$, 
with $\Gamma(w_j) = T_j$ and $\Gamma(w_k) = T_k$, 
defines a pair $(w_j:T_j, w_k:T_k)$.
In turn, this relation can be seen as a directed graph, 
with edges $w_j:T_j \to w_k:T_k$.
In the following, we shall consider only typed processes whose associated $\vdepset{\Psi}$:
\begin{itemize} 
\item is simple, i.e., for any $x, z, S$ there is at most one dependency 
$\vdep{x^n}{z^m}:S \in \vdepset{\Psi}$.
\item 
determines a tree-like structure, without cycles, on $\Gamma$.
\end{itemize}
While considering simple $\vdepset{\Psi}$ simplifies some definitions, 
considering tree-like structures (rather than arbitrary relations) is necessary to ensure that the characteristic processes obtained from $\Gamma$ and $\Psi$ are typable in \lcp; see \Cref{lem: wt_characteristic_new} and 
\Cref{lem:typability_catal_rev}
below.

A tree-like interpretation of $\Gamma$ via $\vdepset{\Psi}$ does not necessarily define a fully connected structure. For instance, $\Gamma$ can have assignments not involved in any value dependency.
Actually, $\vdepset{\Psi}$ defines  
 a \emph{forest} (a set of trees).  
 This is formalized in the definition below, where  the index $i$ relates 
 $w_i:T_i$ (a type  assignment)
 and
 $\mathbf{T}_i$ (its  corresponding tree):

 \begin{definition}[Forest of a  Context]
 \label{d:forest}
Let
$\Gamma = w_1{:}T_1, \ldots, w_n{:}T_n$ be a typing context
and $\Psi$ its associated dependency context. 
The \emph{forest of $\Gamma$ (under $\Psi$)} is defined 
as $\mathcal{F}_\Psi(\Gamma) = \{\mathbf{T}_j\}_{j \in J}$,
where each $\mathbf{T}_j$ is interpreted as a tree with root $w_j:T_j \in \Gamma$.
Each dependency   
$\vdep{w_j^n}{w_k^m} \in \vdepset{\Psi}$ (if any) defines a  relation between
parent $\mathbf{T_j}$ and 
child $\mathbf{T}_k$, which is a tree inductively defined with root 
$w_k:T_k$.
We write $child(\mathbf{T}_j)$ to denote the  set of offspring of $\mathbf{T}_j$.
\end{definition}

We illustrate the above notions by means of examples.

\begin{example}[No Dependencies]
If  $\Gamma = w_1{:}T_1, \ldots, w_n{:}T_n$
and $\vdepset{\Psi} = \emptyset$ 
then the forest of $\Gamma$ will have as many trees as there are assignments in it, i.e.,
$\mathcal{F}_\Psi(\Gamma) = \{\mathbf{T}_1, \ldots, \mathbf{T}_n\}$, with $child(\mathbf{T}_i) = \emptyset$, for all $\mathbf{T}_i$.
\end{example}

\begin{example}[A Simple Tree]
\label{ex:vd}
Let 
$\Gamma$ and $\Psi$ be 
such that
\begin{align*}
\Gamma & = x_1: \wn S.\nilT, ~
x_2: \oc S.\wn T_1.\wn T_2.\nilT,~
x_3: \oc T_1.\nilT,~
x_4: \oc T_2.\nilT
\\
\vdepset{\Psi} & = \{ \vdep{x_1}{x_2}:S,~ \vdep{x_2}{x_3}:T_1,~ \vdep{x_2}{x_4}:T_2\}
\end{align*}
where we have omitted annotations.
Note that $\vdepset{\Psi}$ is simple \emph{and} induces a tree-like structure. 
Indeed we have: $child(\mathbf{T}_1) = \{\mathbf{T}_2\}$, $child(\mathbf{T}_2) = \{\mathbf{T}_3, \mathbf{T}_4\}$, 
$child(\mathbf{T}_3) = child(\mathbf{T}_4) = \emptyset$, and so 
 the forest $\mathcal{F}_\Psi(\Gamma)$ contains a single tree with two levels, with root $x_1:T_1$.

As a non-example of a simple, tree-like structure, consider: 
\begin{align*}
\Gamma' & = x_1: \wn S.\oc T.\nilT, ~
x_2: \oc S.\wn T.\nilT
\\
\vdepset{\Psi'} & = \{ \vdep{x_1}{x_2}:S,~ \vdep{x_2}{x_1}:T\}
\end{align*}
Clearly, $\vdepset{\Psi'}$ induces a graph-like structure, with a simple cycle, on $\Gamma'$.
\end{example}

We now move on to exploit both the value dependencies $\vdepset{\Psi}$ and the (annotated) context $\Gamma$ to refine the definitions of  characteristic processes (\mydefref{def:charprocess}) and  catalyzer contexts (\Cref{def:catalyser}).

\subsection{Characteristic Processes Refined with Value Dependencies}

\label{ss:cat}
To refine the definition of characteristic processes, 
the key idea is simple: given a dependency $\vdep{a^n}{b^m}$ between two sessions $a:\wn^n T.S_1$ and $b:\oc^m T.S_2$, 
rather than implementing two isolated characteristic processes for the two sessions, we  
use
a so-called \emph{bridging session} (denoted $c_{ab}$) 
to implement the dependency between  the characteristic processes. 
To this end, we first revisit \mydefref{def:charprocess}, as follows:

\begin{definition}[Characteristic Processes of a Session Type, with Value Dependencies]\label{def:charprocessv}
Let $\Psi$ be a dependency context. 
Given a name $x$, 
the set of \emph{characteristic processes} of 
the (annotated) session type
$T$ under $\Psi$, denoted $\chrpv{T}{x}{\Psi}$, is inductively defined as follows:
\begin{align*}
\chrpv{\wn^n T'.S}{x}{\Psi} 		&\defeq    
\begin{cases}
\big\{x(y).\bout{c_{xz}}{w}.(\linkr{y}{w} \para Q) \suchthat Q \in \chrpv{S}{x}{\Psi}\big\} 
& \text{if $\vdep{x^n}{z^m} \in \vdepset{\Psi}$}
\\
\big\{x(y).(P \para Q) \suchthat P \in \chrpv{T'}{y}{\Psi} \land Q \in \chrpv{S}{x}{\Psi}\big\} 
& \text{otherwise} 
\end{cases} 
\\
\chrpv{\oc^n T'.S}{x}{\Psi} 		&\defeq  
\begin{cases}
\big\{c_{zx}(y).\bout{x}{w}.(\linkr{y}{w} \para Q) \suchthat Q \in \chrpv{S}{x}{\Psi}\big\} 
& \text{if $\vdep{z^m}{x^n} \in \vdepset{\Psi}$}
\\
\big\{\bout{x}{y}.(P \para Q) \suchthat P \in \chrpv{\dual{T'}}{y}{\Psi} \land Q \in \chrpv{S}{x}{\Psi}\big\} 
& \text{otherwise} 
\end{cases}
\\
\chrpv{\branch lS}{x}{\Psi}				&\defeq    
\big\{\branching xlP \suchthat \forall i\in I.\ P_i\in \chrpv{S_i}{x}{\Psi}\big\}
\\
\chrpv{\select lS}{x}{\Psi} 				&\defeq    
\bigcup_{i\in I}\big\{ \selection x{l _i}.{P_i} \suchthat P_i\in \chrpv{S_i}{x}{\Psi}\big\}
\\
\chrpv{\nilT}{x}{\Psi} &\defeq   \{\nil\} 
\end{align*}
\end{definition}
Changes with respect to \mydefref{def:charprocess} are in the cases of input and output types, which now
implement an additional session if, according to $\vdepset{\Psi}$, the respective sessions are involved in a value dependency.
Indeed, the definition of $\chrpv{\wn^n T.S}{x}{\Psi}$ forwards the value received on $x$ along a channel $c_{xz}$ whenever 
sessions along $x$ and $z$ are related via a value dependency.
Accordingly, the first action in the definition of 
$\chrpv{\oc^n T.S}{x}{\Psi}$ is an input along $c_{xz}$: the received value will be emitted along $x$, thus completing the intended forwarding.


Using this refined definition of characteristic processes, we may now revisit the definition of the characteristic processes of a   typing context (\Cref{def:charenv}).
In order to ensure type preservation (cf. \Cref{lem: wt_characteristic_new} below), the main difference is the need for outermost restrictions that hide the behavior along the bridging sessions created in \mydefref{def:charprocessv}.
We use the tree-like interpretation of $\Psi$ (\Cref{d:forest}) to appropriately close these bridging sessions:



\begin{definition}[Characteristic Processes of a  Typing Context, with Value Dependencies]
\label{def:charenvv}
Assume $\Gamma = w_1{:}T_1, \ldots, w_n{:}T_n$ and $\Psi$.
Recall that  $\mathbf{T}_i$ denotes the tree rooted by $w_i:{T}_i \in \Gamma$.
\begin{enumerate}
\item 
The set of processes $N_i$ is defined 
for each 
${T}_i$ (with $1 \leq i \leq n$)   as:
\begin{equation}
\begin{cases}
\chrpv{{T_i}}{w_i}{\Psi} & \text{if $child(\mathbf{T}_i) = \emptyset$}
\\
\{ \res{\wt{c_{w_iw_j}}}(Q \para 
\prod_{j \in \{1,\ldots, m\}} Q_j)  \suchthat  Q \in \chrpv{{T_i}}{w_i}{\Psi} \land Q_j \in N_j\}
&
\text{if $child(\mathbf{T}_i) = \{\mathbf{T}_1, \ldots, \mathbf{T}_m\}$}.
\end{cases}
\label{eq:charp}
\end{equation}
\item 
Suppose that
$\mathcal{F}_\Psi(\Gamma) = \{\mathbf{T}_1, \ldots, \mathbf{T}_k\}$, with $1 \leq k \leq n$.
The \emph{characteristic processes} of $\Gamma$ is defined as the set
$
\chrpv{\Gamma}{}{\Psi}
\defeq
\{
P_1 \para \cdots \para P_k  \suchthat  P_i \in N_i
\}
$
with $N_i$ as in  (1).
\end{enumerate}
\end{definition}

Thus, the set of characteristic processes of $\Gamma$ is built by exploiting the hierarchy of type assignments 
defined by $\mathcal{F}_\Psi(\Gamma)$. This is required, because we need to connect each individual characteristic process using the bridging sessions $c_{w_iw_j}$.
This definition is a strict generalization of \Cref{def:charenv}: 
with $\vdepset{\Psi}= \emptyset $, every set $N_i$ coincides with the set $\chrpv{{T_i}}{w_i}{\Psi}$ (cf. \Cref{def:charprocessv}), which in turn coincides 
with $\chrpv{{T_i}}{w_i}{}$ (cf. \Cref{def:charprocess}) if $\vdepset{\Psi} = \emptyset$.

The following examples serve to illustrate this definition: 

\begin{example}
\label{ex:charprocgamma}
Consider $\Gamma_1$ and $\Psi_1$ as in \Cref{ex:track} for sub-process $\ov{a_0}\out{\unit}.a_1\inp u.\ov{a_2}\out{u}.\nil$:
\begin{align*}
\Gamma_1 & \defeq a_0: \oc^{0}\nilT.\nilT, a_1: \wn^1 U.\nilT, a_2: \oc^2 U.\nilT \\
\Psi_1 & \defeq \langle a_0, \unit, 0\rangle, (a_1,u, 1), \langle a_2, u, 2\rangle 
\end{align*}
By \Cref{def:veep}, $\vdepset{\Psi_1} = \{\vdep{a_1^1}{a_2^2}\}$.
Also, by \Cref{d:forest}, 
  $\mathcal{F}_{\Psi_{1}}(\Gamma_1) = \{\mathbf{T}_0, \mathbf{T}_1\}$,
with 
$child(\mathbf{T}_0) = child(\mathbf{T}_2) =\emptyset$
and 
$child(\mathbf{T}_1) = \{\mathbf{T}_2\}$.
Recalling that $U \defeq \wn \nilT.\nilT$, 
we have:
\begin{align*}
N_0 &= \chrpv{\oc^{0}\nilT.\nilT}{a_0}{\Psi_1}
\\
N_2 &= \chrpv{\oc^2 U.\nilT}{a_0}{\Psi_1}
\\
N_1 &= \{\res{c_{a_1a_2}}(Q_1 \para Q_2)  \suchthat  Q_1 \in \chrpv{\wn^1 U.\nilT}{a_1}{\Psi_1}, Q_2 \in N_2\}
\end{align*}
Notice that by 
\Cref{def:charprocessv}
processes in $N_2$ are of the form 
$c_{a_1a_2}(w).\bout{a_2}{v}.(\linkr{w}{v} \para Q')$
with $Q' \in  \chrpv{\nilT}{a_0}{\Psi_!}$.
Given these sets, 
we have $
\chrpv{\Gamma_1}{}{\Psi_1}
\defeq
\{
P_0 \para P_1   \suchthat  P_0 \in N_0, P_1 \in N_1\}
$.
%
\end{example}

\begin{example}
Let $\Gamma$ and $\Psi$ be as in \Cref{ex:vd}. 
Expanding \Cref{def:charenvv}, we have:
\begin{align*}
N_4 &= \chrpv{\oc T_2.\nilT}{x_4}{\Psi}
\\
N_3 &= \chrpv{\oc T_1.\nilT}{x_3}{\Psi}
\\
N_2 &= \{\res{c_{x_2x_4}}(\res{c_{x_2x_3}}(Q_2 \para Q_3) \para Q_4)  \suchthat  Q_2 \in \chrpv{\oc S.\wn T_1.\wn T_2.\nilT}{x_2}{\Psi}, Q_3 \in N_3, Q_4 \in N_4\}
\\
N_1 &= \{\res{c_{x_1x_2}}(Q_1 \para Q)  \suchthat  Q_1 \in \chrpv{\wn S.\nilT}{x_1}{\Psi}, Q \in N_2\}
\end{align*}
Recall that by \Cref{def:charprocessv} all processes in $N_4$ will implement an input on $c_{x_2x_4}$.
Similarly, all processes in $N_3$ will implement an input on $c_{x_2x_3}$.
Then, processes in $\chrpv{{T_2}}{x_2}{\Psi}$ implement the dependency by providing corresponding outputs.
With these sets in place, 
and given that 
$\mathcal{F}_\Psi(\Gamma) = \{\mathbf{T}_1\}$, 
we have that $
\chrpv{\Gamma}{}{\Psi}
\defeq
\{
P_1   \suchthat  P_1 \in N_1\}
$.
\end{example}

\noindent
Characteristic processes with value dependencies 
are well-typed in the system of \Cref{ss:lf}:

\label{app: new-char-process}
We repeat the statement of this lemma:

\begin{lemma}%
\label{lem: wt_characteristic_new}
Let $T$, $\Gamma$, and $\Psi$ be an annotated session type, a session typing context, and a dependency context, respectively.
\begin{enumerate}
\item \label{6.11-1}
For all $P \in \chrpv{T}{x}{\Psi}$, we have  $P \cp x:\encob{T}{\ct}, \Delta$, for some  $\Delta$ (possibly empty).
\item  \label{6.11-2}
For all $P \in \chrpv{\Gamma}{}{\Psi}$, we have $P \cp \encob{\Gamma}{\ct}$.
\end{enumerate}
\end{lemma}


\begin{proof}
We consider both parts separately, assuming $\vdepset{\Psi} \neq \emptyset$: the case $\vdepset{\Psi} = \emptyset$ corresponds to \Cref{lem: wt_characteristic}.

\bigskip

\noindent{\underline{\textbf{\Cref{6.11-1}.}}} 
The proof is by induction on the structure of  $T$. 
Thus, there are five cases to consider. The most interesting cases are those for input and output session types with value dependencies; by \Cref{def:charprocessv}, the other cases (input/output without value dependencies, selection, branching) are as in the proof of \Cref{lem: wt_characteristic} (\Cref{applem: wt_characteristic}).

\begin{enumerate}
\item
$\chrpv{\nilT}{x}{\Psi} =  \big\{ \nil \big\}$.
This case follows easily by the fact that $\encob{\nilT}{\ct}= \bullet$ and by Rule~\Did{T-\one}.

\smallskip

\item
$\chrpv{\wn^n T.S}{x}{\Psi} = \big\{x(y).\bout{c_{xz}}{w}.(\linkr{y}{w} \para Q) \suchthat Q \in \chrpv{S}{x}{\Psi}\big\}$,
if $\vdep{x^n}{z^m} \in \vdepset{\Psi}$.
\\
\noindent
By induction hypothesis, 
$Q\cp \Theta, x:\encob{S}{\ct}$, for some context $\Theta$.
Here we distinguish the following two sub-cases, depending on whether name $c_{xz}$ already occurs in $Q$.

\noindent
\begin{enumerate}
\item ${c_{xz}} \notin \dom(\Theta)$.
By Rules~\Did{{T-$\otimes$}} and \Did{{T-$\parl$}}  applied in sequence we have:
\[
\inferrule*[right=\Did{$\parl$}]
{
	\inferrule*[right=\Did{$\tensor$}]
	{\inferrule*[left=\Did{T-$\tid$}]{ }{\linkr{y}{w}\cp y: \encob{T}{\ct}, w:\encob{\dual T}{\ct}}
	\\
	\inferrule*[right=\Did{$\bot$}]{Q\cp x:\encob{S}{\ct},\Theta}{Q\cp x:\encob{S}{\ct},\Theta, {c_{xz}}: \bullet}}
	{\bout{c_{xz}}{w}.(\linkr{y}{w} \para Q) \cp y: \encob{T}{\ct}, x:\encob{S}{\ct}, \Theta, {c_{xz}} : \encob{\dual T}{\ct} \tensor \bullet}
}
{x(y).\bout{c_{xz}}{w}.(\linkr{y}{w} \para Q)\cp \Theta,x:\encob{T}{\ct}\parl \encob{S}{\ct}, {c_{xz}} : \encob{\dual T}{\ct} \tensor \bullet}
\]
By encoding of types in \Cref{f:enctypesct}, we have
$\encob{\wn T.S}{\ct}  = \encob{T}{\ct} \parl \encob{S}{\ct}$. 
Hence, in this case $\Delta = \Theta, {c_{xz}} : \encob{\dual T}{\ct} \tensor \bullet$.

\item $({c_{xz}}: \encob{U}{\ct})\in \Theta$. Let $\Theta = \Theta', {c_{xz}}: \encob{U}{\ct}$. Then, we have the following derivation:
\[
\inferrule*[right=\Did{$\parl$}]
{
	\inferrule*[right=\Did{$\tensor$}]
	{\inferrule*[left=\Did{T-$\tid$}]{ }{\linkr{y}{w}\cp y: \encob{T}{\ct}, w:\encob{\dual T}{\ct}}
	\\ Q\cp x:\encob{S}{\ct},\Theta', {c_{xz}}: \encob{U}{\ct}}
	{\bout{c_{xz}}{w}.(\linkr{y}{w} \para Q) \cp
	y: \encob{T}{\ct}, x:\encob{S}{\ct}, \Theta', {c_{xz}} : \encob{\dual T}{\ct} \tensor \encob{U}{\ct}}
}
{x(y).\bout{c_{xz}}{w}.(\linkr{y}{w} \para Q)\cp
\Theta',x:\encob{T}{\ct}\parl \encob{S}{\ct}, {c_{xz}} : \encob{\dual T}{\ct} \tensor \encob{U}{\ct}}
\]
By encoding of types in \Cref{f:enctypesct}, we have
$\encob{\wn T.S}{\ct}  = \encob{T}{\ct} \parl \encob{S}{\ct}$.
Hence, in this case $\Delta = \Theta', {c_{xz}} : \encob{\dual T}{\ct} \tensor \encob{U}{\ct}$.
\end{enumerate}

\smallskip

\item
$\chrpv{\oc^n T.S}{x}{\Psi} = \big\{c_{zx}(y).\bout{x}{w}.(\linkr{y}{w} \para Q) \suchthat Q \in \chrpv{S}{x}{\Psi}\big\}$
if $\vdep{z^m}{x^n} \in \vdepset{\Psi}$.
\\
\noindent By induction hypothesis,  $Q\cp x:\encob{S}{\ct},\Theta$, for some typing $\Theta$.
Here we distinguish the following two sub-cases, depending on whether name $c_{xz}$ already occurs in $Q$.

\begin{enumerate}
\item ${c_{xz}} \notin \dom(\Theta)$.
By Rules \Did{{T-$\otimes$}} and \Did{{T-$\parl$}}  applied in sequence we have:
\[
\inferrule*[right=\Did{$\parl$}]
{
 {
	\inferrule*[right=\Did{$\tensor$}]
	{\inferrule*[left=\Did{T-$\tid$}]{ }{\linkr{y}{w}\cp y: \encob{T}{\ct}, w:\encob{\dual T}{\ct}}
	\\
	\inferrule*[right=\Did{$\bot$}]{Q\cp x:\encob{S}{\ct},\Theta}{Q\cp x:\encob{S}{\ct},\Theta, {c_{xz}}:\bullet}}
	{\bout{x}{w}.(\linkr{y}{w} \para Q) \cp y: \encob{T}{\ct}, \Theta, x:\encob{\dual T}{\ct}\tensor \encob{S}{\ct}, {c_{xz}}:\bullet}
}
}
{c_{zx}(y).\bout{x}{w}.(\linkr{y}{w} \para Q) \cp \Theta, x:\encob{\dual T}{\ct}\tensor \encob{S}{\ct},  {c_{xz}}:\encob{T}{\ct}\parl\bullet}
\]
By encoding of types in \Cref{f:enctypesct}, we have
$\encob{\oc T.S}{\ct} =  \encob{\dual T}{\ct} \otimes \encob{S}{\ct}$.
Hence, in this case $\Delta = \Theta, {c_{xz}} : \encob{T}{\ct}\parl\bullet$. 

\item $({c_{xz}}: \encob{U}{\ct})\in \Theta$. Let $\Theta = \Theta', {c_{xz}}: \encob{U}{\ct}$. Then, we have the following derivation:
\[
\inferrule*[right=\Did{$\parl$}]
{
 {
	\inferrule*[right=\Did{$\tensor$}]
	{\inferrule*[left=\Did{T-$\tid$}]{ }{\linkr{y}{w}\cp y: \encob{T}{\ct}, w:\encob{\dual T}{\ct}}
	\\
	{Q\cp x:\encob{S}{\ct},\Theta', {c_{xz}}:\encob{U}{\ct}} }
	{\bout{x}{w}.(\linkr{y}{w} \para Q) \cp y: \encob{T}{\ct}, \Theta', x:\encob{\dual T}{\ct}\tensor \encob{S}{\ct}, {c_{xz}}:\encob{U}{\ct}}
}
}
{c_{zx}(y).\bout{x}{w}.(\linkr{y}{w} \para Q) \cp \Theta', x:\encob{\dual T}{\ct}\tensor \encob{S}{\ct},  {c_{xz}}:\encob{T}{\ct}\parl\encob{U}{\ct}}
\]
By encoding of types in \Cref{f:enctypesct}, we have
$\encob{\oc T.S}{\ct} =  \encob{\dual T}{\ct} \otimes \encob{S}{\ct}$.
Hence, in this case $\Delta = \Theta', {c_{xz}}:\encob{T}{\ct}\parl\encob{U}{\ct}$.
\end{enumerate}

\smallskip

\item
$\chrpv{\branch lS}{x}{\Psi} = \big\{\branching xlP \suchthat \forall i\in I.\ P_i\in \chrpv{S_i}{x}{\Psi}\big\}$. This case follows the same lines as the corresponding case in \Cref{lem: wt_characteristic}.

\smallskip

\item
$\chrpv{\select lS}{x}{\Psi} = \bigcup_{i\in I}\big\{ \selection x{l _i}.{P_i} \suchthat P_i\in \chrpv{S_i}{x}{\Psi}\big\}$. This case also follows the same lines as the corresponding one in \Cref{lem: wt_characteristic}.
\end{enumerate}

\bigskip

\noindent{\underline{\textbf{\Cref{6.11-2}.}}} 
By \Cref{def:charenvv}, every $P \in \chrpv{\Gamma}{}{\Psi}$ is of the form
$
\prod_{i \in \{1,\ldots, k\}} P_i  
$,
with $P_i \in N_i$ (cf.~\eqref{eq:charp}).
We first prove that for every 
$P \in N_i$, we have  $P \cp \encob{\Gamma_i}{\ct}, 
\Delta$, where 
$\Gamma_i \subseteq \Gamma$
and
$\Delta$ is possibly empty.
$\Gamma_i$ contains $w_i:T_i$ and  assignments connected to $w_i:T_i$ via dependencies.
Recall that $N_i$ is defined for each $w_i:T_i$.
We then proceed by induction on $m$, the height of $\mathbf{T}_i$ (the tree associated to $w_i:T_i$):

\begin{itemize}
\item Base case, $m = 0$: Then, 
since $\mathbf{T}_i$ does not have children, 
by definition of $N_i$ (cf.~\eqref{eq:charp}), the sets $N_i$  and $\chrpv{{T_i}}{w_i}{\Psi}$ coincide.
Hence, \Cref{lem: wt_characteristic_new}(1) ensures 
that  $P \cp w_i:\encob{T_i}{\ct}$, for all $P \in N_i$.

\item Inductive case, $m > 0$. 
Then, $\mathbf{T}_i$ has $m$ levels; by \Cref{def:charenvv}, every $P_m \in N_i$ 
is of the form
$$
P_m = \res{\wt{c_{w_iw_j}}}(Q_0 \para \prod_{j \in \{1,\ldots, l\}} Q_j)
$$
where 
$Q_0 \in \chrpv{{T_i}}{w_i}{\Psi}$ 
and 
$Q_j \in N_j$, for some $\mathbf{T}_1, \ldots, \mathbf{T}_l \in child(\mathbf{T}_i)$.
Thus,  
$w_i$ has a dependency with each of $w_1, \ldots, w_l$, implemented along $c_{w_iw_1}, \ldots, c_{w_iw_l}$. 
In fact, by \Cref{lem: wt_characteristic_new}(1):
$$
Q_0 
\cp
w_i:\encob{T_i}{\ct}, \Delta_0
$$
where 
$\Delta_0 = c_{w_iw_1}:T_1, \ldots, c_{w_iw_l}:T_l$.

We must prove that $P_m \cp \encob{\Gamma_i}{\ct}$. 
Recall that 
processes $Q_1, \ldots, Q_l$ are associated to trees $\mathbf{T}_1, \ldots, \mathbf{T}_l$;
now, because $child(\mathbf{T}_i) = \{\mathbf{T}_1, \ldots, \mathbf{T}_l\}$, the
height of these trees is \emph{strictly} less than that of $\mathbf{T}_i$.
Therefore, by IH we have
$$Q_j \cp \encob{\Gamma_j}{\ct}, \Delta_j$$
for each $j \in \{1,\ldots, l\}$, 
where $\Delta_j = c_{w_iw_j}: S_j$, for some $S_j$. By construction, $S_j$ is dual to $T_j$.
(There are no other assignments in $\Delta_j$, because $\mathbf{T}_j$ is a proper tree whose only visible connection is with 
its parent $\mathbf{T}_i$.)
Also, $ \encob{\Gamma_i}{\ct} = w_i:\encob{T_i}{\ct} \cup \bigcup_{j \in \{1,\ldots, l\}} 
 \encob{\Gamma_j}{\ct}$.
Therefore,  $P_m$ can be typed by  composing $Q_0$ with each of $Q_1, \ldots, Q_l$, using $l$ successive applications of Rule~$\Did{T-$\cut$}$:
$$
\res{c_{w_iw_l}}( \cdots \res{c_{w_iw_2}}(\res{c_{w_iw_1}}(Q_0 \para Q_1) \para Q_2) \para \cdots \para Q_l)
\cp 
\encob{\Gamma_i}{\ct}
$$
\end{itemize}
This proves that $P_i \in N_i$ implies $P_i \cp \encob{\Gamma_i}{\ct}$.
The thesis, 
$
\prod_{i \in \{1,\ldots, k\}} P_i  \cp \encob{\Gamma}{\ct}
$, 
then follows easily by composing $P_1, \ldots, P_k$: since they are all typable under disjoint 
contexts ($\encob{\Gamma_1}{\ct}, \ldots, \encob{\Gamma_k}{\ct}$, respectively), it suffices to  use Rule~\Did{T-$\mix$}.
\end{proof}

The main difference between \Cref{lem: wt_characteristic_new} (\Cref{6.11-1})
and 
the corresponding result in \Cref{s:enco} (i.e., \Cref{lem: wt_characteristic})
is the typing context $\Delta$, which describes the bridging sessions: if there are no dependencies then no bridging sessions occur and $\Delta$ will be empty.
\Cref{lem: wt_characteristic_new} (\Cref{6.11-2}) is exactly as 
in \Cref{lem: wt_characteristic}: even though characteristic processes for a typing context (\Cref{def:charenvv}) are obtained by composing characteristic processes with bridging sessions, these 
sessions are closed under restriction and so no additional typing context is needed. 

\subsection{Catalyzers Refined with Value Dependencies}
\label{app:typability_catal_rev}

To revise \Cref{def:catalyser}, we use some 
 convenient notations:

\begin{notation}
\label{not:catal}
We shall use the following notations:
\begin{itemize}
\item Let $y^n$ denote the $n$-th occurrence of name $y$ with respect to top-level.

\item Let $\alpha, \alpha', \ldots$ denote session prefixes of the form $\wn^n S$ or $\oc^n S$.

\item Given an annotated session type $T$, we 
use a context on types $K$ to highlight a specific prefix $\alpha$ in $T$, i.e., we 
write $T = K[\alpha]$.

\item Given an annotated session type $T$, we write 
$T \setminus (\alpha_1, \cdots, \alpha_k)$
to denote the type obtained from $T$ by removing  
 prefixes $\alpha_1, \cdots, \alpha_k$ ($k \geq 1$).
\end{itemize}
\end{notation}

The following auxiliary definition explains how to ``shorten'' session types based on value dependencies.
Its precise purpose will become apparent shortly.

%

\begin{definition}[Refined Characteristic Processes]
\label{d:rcp}
 Let  $\Gamma = \{x_1:T_1, \ldots, x_n:T_n\}$
and $\Psi$ be a typing context and a dependency context, respectively.
\begin{enumerate}
\item Given  $\Psi$ and a name $x_i$, we define the sets of types
\begin{align*} 
\mathtt{out}_\Psi(x_i) & = \{S^m \,|\, \vdep{x_j^n}{x_i^m}:S \in \vdepset{\Psi}\}
\\
\mathtt{in}_\Psi(x_i)  & = \{T^n \,|\, \vdep{x_i^n}{x_j^m}:T \in \vdepset{\Psi}\}
\end{align*} 

\item 
We write $\prod_{i \in \{1, \ldots, n\}} \mathsf{Q}_i$ to denote 
the \emph{refined characteristic processes under $\Psi$} where, 
each $\mathsf{Q}_i$ satisfies:
$$
\begin{cases}
\mathsf{Q}_i \in \chrpv{T_i}{x_i}{\Psi}
& 
\text{If $\mathtt{out}_\Psi(x_i) = \mathtt{in}_\Psi(x_i) = \emptyset$}
\\
\mathsf{Q}_i \in \chrpv{T_i\setminus(\wn^{l_1} T_1, \cdots, \wn^{l_m} T_m) }{x_i}{\Psi}
& 
\text{If $\mathtt{out}_\Psi(x_i) = \emptyset \land 
\mathtt{in}_\Psi(x_i) = \{T_1^{l_1}, \ldots, T_m^{l_m}\}$}
\\
\mathsf{Q}_i \in \chrpv{T_i\setminus(\oc^{h_1} S_1) }{x_i}{\Psi}
&
\text{If $\mathtt{out}_\Psi(x_i) = \{S_1^{h_1}\} \land \mathtt{in}_\Psi(x_i) = \emptyset$}
\\
\mathsf{Q}_i \in \chrpv{T_i\setminus(\oc^{h_1} S_1, \wn^{l_1} T_1, \cdots, \wn^{l_m} T_m) }{x_i}{\Psi}
&
\text{If $\mathtt{out}_\Psi(x_i) = \{S_1^{h_1}\} \land \mathtt{in}_\Psi(x_i) = \{T_1^{l_1}, \ldots, T_m^{l_m}\}$}
\end{cases}
$$
\end{enumerate}
\end{definition}

Notice that because we assume $\vdepset{\Psi}$ is simple and tree-like, $\mathtt{out}_\Psi(x_i)$ is either empty or a singleton---a node of the tree is either the root or has exactly one parent.
We illustrate \Cref{d:rcp} with an example:

\begin{example}
Considering $\Gamma$ and $\vdepset{\Psi}$ as in \Cref{ex:vd}, we have:
\begin{align*}
\mathtt{out}_\Psi(x_1)  & =  \emptyset & \mathtt{in}_\Psi(x_1) &  = \{S\}
\\
\mathtt{out}_\Psi(x_2) & =  \{S\} & \mathtt{in}_\Psi(x_2) &  = \{T_1,T_2\}
\\
\mathtt{out}_\Psi(x_3) & =  \{T_1\} & \mathtt{in}_\Psi(x_2) &  = \emptyset
\\
\mathtt{out}_\Psi(x_4) &  =  \{T_2\} & \mathtt{in}_\Psi(x_2) &  = \emptyset
\end{align*}
Because all types in $\Gamma$ are involved in value dependencies, 
$\mathsf{Q}_i \in \chrpv{\nilT}{x_i}{\Psi}$, for all $i \in \{1, \ldots, 4\}$.
\end{example}


We now finally ready to revise \Cref{def:catalyser}:

\begin{definition}[Catalyzers, with Value Dependencies] 
\label{def:catalyserv}
Consider typing contexts $\Gamma,\Gamma'$ and a simple dependency context $\Psi$.
Suppose 
 $\Gamma = \{x_1:T_1, \ldots, x_n:T_n\}$
 and
$|\vdepset{\Psi}| = k$, with $\dom{(\Psi)} \subseteq \dom{(\Gamma)}$.
Assume that 
$x_i:T_i \in \Gamma$ implies  $y_i:\dual{T_i} \in \Gamma'$, for all 
$x_i \in \dom{(\Psi)}$.

The  set of  \emph{catalyzers of $\Gamma, \Psi$ with respect to $\Gamma'$}, noted 
$\mathcal{C}_{\Gamma | \Gamma'}^\Psi$, 
is defined as  follows:
$$
\{
C[\cdot]  \suchthat  C[\cdot] = 
(\nub \wt{x})((\nub \wt{u})([\cdot]\sigma_1 \cdots \sigma_k \para F_1 \para \cdots \para F_k)  
\para \mathsf{Q}_1 \para \cdots \para \mathsf{Q}_n)
\}
$$
where, for the $h$-th element in $\vdepset{\Psi}$, denoted 
  $\vdep{x_i^{l_1}}{x_j^{l_2}}:S$ with $x_i:T_i$ and $x_j:T_j$, we have:

\begin{itemize}
\item  
$\sigma_h = \substj{u_h}{{y_i}^{n_1}}\substj{u_{h}}{{y_j}^{n_2}}$, 
where $n_1$ and $n_2$ are such that
$\dual{T_i} = K_i[\oc^{n_1} S]$
and 
$\dual{T_j} = K_j[\wn^{n_2} S]$, respectively, for some type contexts $K_i, K_j$, 
and $y_i:\dual{T_i}, y_j:\dual{T_j} \in \Gamma'$.
\item  $F_h = u_h(y).\bout{u_{h}}{w}.(\linkr{w}{y} \para \nil)$.
\end{itemize}
In $\sigma_h$ and $F_h$, name $u_h$ is fresh.
Also, $\mathsf{Q}_1, \ldots, \mathsf{Q}_n$
are refined under $\Psi$ as in
\Cref{d:rcp}.
\end{definition}

Intuitively, the above definition isolates the value dependencies from the process in the hole by first creating a separate session $u_k$ using the substitution $\sigma_k$; the dependency is then preserved using the forwarder process $F_k$. Isolating dependencies in this way modifies the session types implemented by the process in the hole; to account for these modifications, we consider characteristic processes based on \emph{shortened} session types---session types in which the prefixes involved in a dependency have been removed (cf. \Cref{d:rcp}). 


If there are no dependencies (i.e., if $\vdepset{\Psi}$ is empty) then we recover the  definition from \Cref{s:enco} (\Cref{def:catalyser}): indeed, without dependencies there are no $\sigma_h$ and $F_h$, and 
 the condition $\mathtt{out}_\Psi(x_i) = \mathtt{in}_\Psi(x_i) = \emptyset$
follows trivially, for all $x_i$ such that $x_i : T_i \in \Gamma$.





\begin{example}[\Cref{def:catalyserv} at Work]
Consider process
$\res{x_1y_1}\res{x_2y_2}(P_1 \pp P_2)$, 
where $P_1$ and $P_2$ are defined as:
\begin{align*}
\underbrace{y_1:\oc^0 \mathsf{Int}.\wn^1 \mathsf{Bool}.\nilT, y_2:\wn^2 \mathsf{Int}.\wn^3 \mathsf{Bool}.\nilT}_{\Gamma_2}\compsi \Psi_2 
& \s \underbrace{\bout{y_1}{2}.y_1(z).y_2(x).y_2(w).\nil}_{P_2}
\\
\underbrace{x_1:T_1, x_2:T_2}_{\Gamma_1}
\compsi \Psi_1 
& \s \underbrace{x_1(u).\bout{x_1}{true}. \bout{x_2}{u}.\bout{x_2}{false}. \nil}_{P_1}
\end{align*}
where 
$T_1 = \wn^0 \mathsf{Int}.\oc^1 \mathsf{Bool}.\nilT$ and $T_2:\oc^2 \mathsf{Int}.\oc^3 \mathsf{Bool}.\nilT$.
$P_1$ implements the dependency $\vdepset{\Psi_1} = \{\vdep{x_1^0}{x_2^2}:\mathsf{Int}\}$:
 the received value by $y_2(x)$ in $P_2$ should be $2$, not some arbitrary $\mathsf{Int}$.

We define context $C[\cdot] \in \mathcal{C}_{\Gamma_1(y_1,y_2) | \Gamma_2}^{\Psi_1(y_1,y_2)}$ 
to host 
the translation of $P_2$ in \lcp, denoted $Q_2$.
We write
$\Psi_1(y_1,y_2)$
and
$\Gamma_1(y_1,y_2)$ to denote a name renaming on $\Psi_1$ and $\Gamma_1$, needed 
for consistency with the single restriction in \lcp. 
Types in $\Gamma_1, \Gamma_2$ are annotated with the positions of their associated prefixes wrt the top-level.
Here, relevant positions in $P_2$ are 0 (for $y_1$) and 2 (for $y_2$): these are the prefixes that interact with the dependency in $P_1$.
Note that these positions in $P_2$ are not stored in $\Gamma_1$, but in $\Gamma_2$.
(In this simple example, the positions for $x_1, x_2$ in $P_1$ happen to coincide with those for $y_1, y_2$ in $P_2$; in the general case, they do not coincide.)
Using \Cref{not:catal}, we have 
$\dual{T_1} = K_1[!^0 \mathsf{Int}]$ and $\dual{T_2} = K_2[?^2 \mathsf{Int}]$, 
with type contexts $K_1  = [\cdot].\wn^1 \mathsf{Bool}.\nilT, K_2 = [\cdot].\wn^3 \mathsf{Bool}.\nilT$.
Thus, $n_1 = 0, n_2 = 2$ and the required substitution is
$\sigma_1 = \substj{u_1}{y^{0}_1}\substj{u_1}{y^{2}_2}$.

We can now define $C[\cdot]$ by considering the value dependence from $P_1$ as captured by $\Psi_1$. 
Expanding \Cref{def:catalyserv} we have:
$$C[\cdot] = 
(\nub x_1,x_2)((\nub u_1)([\cdot]\sigma_1 \para F_1 ) \para \mathsf{Q}_1 \para  \mathsf{Q}_2)
$$
with $\mathsf{Q}_1, \mathsf{Q}_2$ as in \Cref{d:rcp}.
As we will see, \Cref{def:typed_encr} translates $P_2$ into $Q_2$:
$$
y_1: \mathsf{Int} \otimes \mathsf{Bool} \parl \bullet,
y_2: \mathsf{Int} \parl \mathsf{Bool} \parl \bullet
\cp
 \dual{y_1}(z).(\linkr{2}{z}\para y_1(z).y_2(w).y_2(x).\nil)
$$
We now illustrate how $Q_2$ ``fits into'' $C[\cdot]$. We first apply $\sigma_1$ to $Q_2$:
$$
y_1: \mathsf{Bool} \parl \bullet,
y_2: \mathsf{Bool} \parl \bullet,
u_1: \mathsf{Int} \otimes \mathsf{Int} \parl \bullet
\cp
 \dual{u_1}(z).(\linkr{2}{z}\para y_1(z).u_1(w).y_2(x).\nil)
$$
By renaming the two prefixes involved in the dependency, we have isolated it using a new session $u_1$.
The types for $y_1, y_2$ have been shortened.
The only forwarder needed is $F_1$:
$$u_1(y).\bout{u_{1}}{w}.(\linkr{w}{y} \para \nil) \cp u_1:  \mathsf{Int} \parl \mathsf{Int} \otimes \bullet$$
Clearly, $Q_2\sigma_1$ and $F_1$ can be composed with a cut:
$$\res{u_1}(Q_2\sigma_1 \para F_1) \cp y_1:\wn \mathsf{Bool}.\nilT, y_2:\wn \mathsf{Bool}.\nilT
$$
We close the sessions on $y_1$ and $y_2$ using 
$R_1 \in \chrpv{\wn \mathsf{Bool}.\nilT}{y_1}{\Psi}$
and 
$R_2 \in \chrpv{\wn \mathsf{Bool}.\nilT}{y_2}{\Psi}$:
$$
C[Q_2] = \res{y_2}(\res{y_1}(\res{u_1}([Q_2]\sigma_1  \para F_1) \para R_1) \para R_2)
$$
Hence, $C[\cdot]$ closes $Q_2$ by preserving the dependency implemented by $P_1$.
\end{example}

Using this new definition, we have 
the following lemma, which is the analog of \Cref{lem:typability_catal} (Page~\pageref{lem:typability_catal}) extended with value dependencies:

\begin{lemma}[Catalyzers with Value Dependencies Preserve Typing]
\label{lem:typability_catal_rev}
Let $P\cp \encob{\Gamma}{\ct}, \encob{\Gamma'}{\ct}$, such that $\Gamma_1, \Gamma_2$ are typing contexts
with $\Gamma_1 = \dual\Gamma$.
Let $\Psi$ be a simple dependency context 
such that $\dom{(\Psi)} \subseteq \dom{(\Gamma_1)}$.
Assume that $x_i:T_i \in \Gamma_1$ implies  $y_i:\dual{T_i} \in \Gamma_2$, for all 
$x_i \in \dom{(\Psi)}$.
Then, by letting $C[\cdot] \in \catdep C{\Gamma_1 | \Gamma_2}\Psi$ (as in \Cref{def:catalyserv}), we have 
$C[P]\cp  \encob{\Gamma'}{\ct}$.
\end{lemma}

%

\begin{proof}[Proof (Sketch)]
We proceed by induction on $k$, the size of $\Gamma$.
\begin{itemize}
\item Case $k = 1$: Trivial, because no dependencies are possible and a single characteristic process suffices to close the single session in $\encob{\Gamma}{\ct}$ and establish the thesis.
\item Case $k = 2$: There are two sub-cases, depending on whether the two assignments are related by a value dependency.
\begin{enumerate}
\item
If there is no value dependency then two independent characteristic processes suffice 
close the two sessions in $\encob{\Gamma}{\ct}$ and establish the thesis. 

\item
Suppose there is a dependency  
$\vdep{x_i^{l_1}}{x_j^{l_2}}:S$, with $x_i:\dual{T_i}$ and $x_j:\dual{T_j}$ in $\Gamma_1$.
This dependency should be implemented by any $C_2 \in \catdep C{\Gamma_1 | \Gamma_2}\Psi$.
By \Cref{def:catalyserv}, any $C_2[\cdot] $ is of the form
$$
(\nub \wt{x})((\nub u_1)([\cdot]\sigma_1  \para F_1) \para \mathsf{Q}_i  \para \mathsf{Q}_j)
$$
where 
$\sigma_1$ is a substitution, $F_1$ is a forwarder process, and 
$\mathsf{Q}_i$ and $\mathsf{Q}_j$ are characteristic processes defined shortened versions of $\dual{T_i}$ and  $\dual{T_j}$.
Observe that 
$\mathtt{out}_\Psi(x_i) = \emptyset$,
$\mathtt{in}_\Psi(x_i) = \{S^{l_1}\}$,
$\mathtt{out}_\Psi(x_j) = \{S^{l_2}\}$,
and
$\mathtt{in}_\Psi(x_j) = \emptyset$.

We must prove that $C_2[P] \cp \encob{\Gamma'}{\ct}$.
We show this by successively incorporating all ingredients in $C_2$.
First, we spell out the typing for $P$:
\begin{equation*}
P \cp x_i:\encob{T_i}{\ct}, x_j:\encob{T_j}{\ct}, \encob{\Gamma'}{\ct} 
\end{equation*}
Then, applying the substitution, we have:
\begin{equation*}
P\sigma_1 \cp 
\encob{\Gamma_0}{\ct}, x_i:\encob{T'_i}{\ct}, x_j:\encob{T'_j}{\ct}, \encob{\Gamma'}{\ct}, 
u_1: {S} \otimes \dual{S} \parl  \bullet
\end{equation*}
where 
$T'_i = T_i \setminus (\oc^{l_2} S)$
and 
$T'_j = T_j \setminus (\wn^{l_2} S)$.
We now compose with the forwarder 
$F_1 = u_1(y).\bout{u_{1}}{w}.(\linkr{w}{y} \para \nil)$, which can be typed 
as
$ F_1 \cp u_1:  \dual{S} \parl  {S} \otimes \bullet$:
\begin{equation*}
\res{u_1}(P\sigma_1 \para F_1) 
\cp 
  x_i:\encob{T'_i}{\ct}, x_j:\encob{T'_j}{\ct}, \encob{\Gamma'}{\ct}
\end{equation*}
At this point, we  compose $\res{u_1}(P\sigma_1 \para F_1)$ with  the characteristic processes $\mathsf{Q}_i \in \chrpv{\dual{T_i}\setminus(\wn^{l_1} S)}{x_i}{\Psi}$
and 
$\mathsf{Q}_j \in \chrpv{\dual{T_j}\setminus(\oc^{l_1} S)}{x_i}{\Psi}$:
\begin{equation*}
\res{x_2}(\res{x_1}(\res{u_1}(P\sigma_1 \para F_1) \para \mathsf{Q}_i) \para \mathsf{Q}_j) 
\cp 
\encob{\Gamma'}{\ct}
\end{equation*}
and so the thesis  $C_2[P] \cp \encob{\Gamma'}{\ct}$ holds.
\end{enumerate}
\item Case $k > 2$ (inductive step): We assume that 
 $C_k[P] \cp \encob{\Gamma'}{\ct}$ holds, with $C_k[\cdot]$ closing sessions $
 \Gamma' = x_{1}:T_{1}, \ldots, x_{k}:T_{k}$. We  must show that 
  $C_{k+1}[P] \cp \encob{\Gamma'\setminus x_{k+1}:T_{k+1}}{\ct}$ holds, with $C_{k+1}[\cdot]$ closing sessions 
  $\Gamma'' = \Gamma', x_{k+1}:T_{k+1}$. 
  
  Considering $x_{k+1}:T_{k+1}$ entails adding  $0 \leq n \leq k$ additional value dependencies.
  We proceed by induction on $n$. The base case $n=0$ is straightforward: $C_{k+1}[\cdot]$ can be obtained by composing $C_{k}[\cdot]$ with the additional characteristic process $\mathsf{Q}_{k+1} \in \chrpv{\dual{T_{k+1}}}{x_{k+1}}{\Psi}$, using Rule~$\Did{T-$\cut$}$. 
  The inductive step ($n > 0$) proceeds by 
  observing that any context $C_{k+1}[\cdot]$ can be built out of $C_{k}[\cdot]$. 
  Exact modifications depend on $n$:
  adding $n$ new dependencies concern $n+1$ characteristic processes  in $C_{k}[\cdot]$. 
  Therefore, $C_{k+1}[\cdot]$ modifies $C_{k}[\cdot]$ by 
  (i) adding $n$ new $\sigma_i$ and $F_i$;
  (ii) adding a new $\mathsf{Q}_{k+1}$;
(ii)  modifying $n$  processes among $\mathsf{Q}_1, \ldots, \mathsf{Q}_k$ (already in $C_{k}[\cdot]$) to account for the fact that they need to be (further) shortened. 
With these modifications, the process $C_{k+1}[P]$ is typable following the lines of case $k = 2$.
\end{itemize}
\end{proof}

\subsection{Translating Session Processes into \lcp Exploiting Value Dependencies}
\label{app: type-preservation-values}

We may now give the revised translation with value dependencies, denoted $\encCPalt{\cdot}{}$.
We rely on \Cref{n:bn} as well as on
{the following extension of \Cref{not:para}, which annotates type judgments with ``discarded'' portions of contexts:}

\begin{notation}[Hidden/Bracketed Sessions, Extended]
\label{not:paraV}
Consider the following notations:
\begin{enumerate}[label=$\bullet$]
\item We shall write:
$$
{\Gamma_1,\Gamma_2}, x:\oc^0 T.S\compsi\Psi_1, \Phi_1  \restricted{\Psi_2, \Phi_2} \s\res{zy}\dual{x}\out y.(P_1\para P_2)
$$
rather than
$${\Gamma_1,\Gamma_2}, x{:}\oc^0 T.S\compsi\Psi_1, \Phi_1 \s\res{zy}\dual{x}\out y.(P_1\para P_2)$$
assuming both 
${{\Gamma_1},  z:\ov T\compsi\Psi_1, \Psi_2\s P_1}$
and
${{\Gamma_2},x:{S}\compsi \Phi_1,\Phi_2 \s P_2}$.

\item 
%
Similarly, we shall write
$${\Gamma_1\restricted{\wt{x:S}} \compsi \Psi_1 \restricted{\Psi_2} \circc 
\Gamma_2\restricted{\wt{y:T}} \compsi \Phi_1 \restricted{\Phi_2}}\s \res {\wt{x}\wt{y}:\wt{S}}(P_1\para P_2)$$
rather than
$${\Gamma_1,\Gamma_2} \compsi \Psi_1, \Phi_1 \s \res {\wt{x}\wt{y}:\wt{S}}(P_1\para P_2)$$
assuming 
${\Gamma_1, {\wt{x:S}}} \compsi\Psi_1,\Psi_2 \s P_1$, 
and
$\Gamma_2, \wt{y:T} \compsi \Phi_1, \Phi_2 \s P_2$,
and 
$S_i = \dual{T_i}$.

\end{enumerate}
\end{notation}



\begin{definition}[Translating into $\lcp$ with Value Dependencies]\label{def:typed_encr}
Let $P$ be such that $\Gamma \compsi\Psi \s P$ and $P \in \lkoba$.
The set of \lcp processes $\encCPalt{\Gamma \compsi\Psi \s P}{}$ is  defined inductively in \Cref{fig:typed_encr}.
\begin{figure}[!t]
\begin{mdframed}
\small
\begin{align*}
\encCPalt{\Gamma^\eende \compsi \emptyset \s \nil}					& \defeq  \big\{\nil \big\}
\\[2mm]
\encCPalt{\Gamma, x:\oc^0 T.S, v:T \compsi\Psi, \langle x, v, 0 \rangle \s\dual{x}\out v.P'}		& \defeq  \big\{ \dual{x}(z).(\linkr{v}{z}\para Q) \suchthat
Q \in \encCPalt{\Gamma, x:{S} \compsi\Psi\s P'}\big\}
\\[2mm]
\encCPalt{{\Gamma_1,\Gamma_2}, x:\oc^0 T.S\compsi\Psi_1, \Phi_1  [\Psi_2, \Phi_2] \s\res{zy}\dual{x}\out y.(P_1\para P_2)}&	\defeq\\
\big\{ \dual{x}(y).(Q_1\para Q_2) \ \suchthat\
Q_1 \in \encCPalt{{{\Gamma_1},  z:\ov T\compsi\Psi_1, \Psi_2\s P_1}}& \wedge\
Q_2 \in \encCPalt{{{\Gamma_2},x:{S}\compsi \Phi_1,\Phi_2\s P_2}}
\big\}
\\[2mm]
\encCPalt{\Gamma, x:\wn^0 T.S \compsi\Psi, (x, y, 0) \s x\inp {y:T}.P'}		& \defeq  \big\{x(y).Q \suchthat Q \in 
\encCPalt{{\Gamma,x:S,y:T \compsi\Psi }\s P'}	\big\}		
\\[2mm]
\encCPalt{\Gamma, x:\select lS \s\selection x{l_j}.P'}		& \defeq  \big\{\selection x{l_j}.Q \suchthat Q \in 
\encCPalt{{\Gamma,x:S_j \compsi\Psi }\s P'}\big\}								
\\[2mm]  			
\encCPalt{\Gamma, x:\branch lS \compsi\Psi\s \branching xlP}&  \defeq
\big\{\parbranching x{l_i}{Q_i} \suchthat Q_i \in \encCPalt{\Gamma,x:S_i \compsi\Psi \s P_i}\big\}
\\[2mm]  
\encCPalt{\Gamma_1\restricted{\wt{x:S}} \compsi \Psi_1 \restricted{\Psi_2(\wt{x})} \circc 
\Gamma_2\restricted{\wt{y:T}} \compsi \Phi_1 \restricted{\Phi_2(\wt{y})} \s 
\res {\wt{x}\wt{y}:\wt{S}}(P_1&\para P_2)}
\defeq\\
\big\{C_1[Q_1] \para G_2  \suchthat 
G_2 \in  \chrpv{\Gamma_2}{}{\Phi_1}\, , 
C_1 \in \mathcal{C}_{\,\wt{x:T} \,|\, \wt{y:S}}^{\Phi_2(\wt{x})}&\,,
Q_1 \in  \encCPalt{{\Gamma_1, {\wt{x:S}}} \compsi\Psi_1,\Psi_2 \s P_1} \big\} \ \ \cup
\\[1mm]									
\big\{G_1 \para C_2[Q_2]  \suchthat  
G_1 \in \chrpv{\Gamma_1}{}{\Psi_1}\, ,
C_2 \in \mathcal{C}_{\,\wt{y:S} \,|\, \wt{x:T}}^{\Psi_2(\wt{y})}&\,,
Q_2 \in \encCPalt{\Gamma_2,\wt{y:T}\compsi\Phi_1,\Phi_2\s P_2}
\big\} 
\end{align*}
\end{mdframed}

\caption{Translation $\encCPalt{\cdot}{}$. \label{fig:typed_encr}} 
\end{figure}
\end{definition}

The first six cases of \Cref{fig:typed_encr} extend those in $\encCP{\cdot}{}$ (\mydefref{def:typed_enc}) with the dependency context $\Psi$ and with annotated types, using \Cref{not:paraV} in the case of output.

More prominent differences arise in the translation of $\res {\wt{x}\wt{y}:\wt{S}}(P_1 \para P_2)$.
The first line of the definition specifies how  the translation   leads to process of the form
$C_1[Q_1] \para G_2$, in which the right-hand side of the parallel  is generated:
\begin{enumerate}[label=$\bullet$]
\item Process $Q_1$ belongs to the set that results from translating sub-process $P_1$ with an appropriate typing judgment, which includes $\wt{x:S}$.

\item   $C_1$ belongs to the set of catalyzers that implement the 
context $\wt{x:T}$, i.e., 
the dual behaviors of the sessions implemented by $P_1$.
We use \Cref{def:catalyserv} to ensure that $C_1$
  preserves the relevant value dependencies in $P_2$, denoted $\Phi_2$. 
  We write $\Phi_2(\wt{x})$ to specify that the dependencies refer to names in $\wt{x}$ (as required by $Q_1$)
  rather than to names in $\wt{y}$ (as present in $P_2$). 
  As before, this step removes the double restriction operator.

\item  $G_2$ belongs to the set of characteristic processes for $\Gamma_2$, which describes the sessions implemented by $P_2$. We use \Cref{def:charenvv} to preserve the relevant value dependencies in $P_2$, denoted $\Phi_1$.

  \end{enumerate}
Similarly, the second line of the definition specifies how  the translation leads to processes of the form
$G_1 \para C_2[Q_2]$. 
Notice that $G_1$, $C_2$, and $Q_2$ are obtained as before.
Overall, we preserve value dependencies both when generating
the new portion of the parallel process (i.e., $G_1$ and $G_2$) but also when 
constructing the catalyzer contexts (i.e.,  $C_1$ and $C_2$) that 
close the portions of the given process that are produced by translation (i.e.,  $Q_1$ and $Q_2$).

As a sanity check, we now state the analogues of 
\Cref{thm:L0-L2} (type preservation)
and
\Cref{thm:oc} (operational correspondence)
for the second translation.
We start by presenting an auxiliary lemma.
\begin{lemma}\label{lem:myplus}
Let $S$ and $\Gamma$ be respectively, an annotated session type and an annotated typing context, and let
$\myplus S$ and $\myplus\Gamma$ be the corresponding type and context where the annotation is incremented according to \Cref{def:myplus}.
Then, $\encob{S}{\ct}  = \encob{\myplus S}{\ct}$
and
$\encob{\Gamma}{\ct}  = \encob{\myplus\Gamma}{\ct}$.
\end{lemma}
\begin{proof}
It follows immediately by the encoding given in \Cref{f:enctypesct} and \Cref{def:myplus}.
\end{proof}


\begin{theorem}[$\encCPalt\cdot$ is Type Preserving]
\label{thm:rwtp}
Let $\Gamma \compsi \Psi\s P$.
Then, for all $Q\in \encCPalt{\Gamma \compsi \Psi\s P}$, we have that $Q\cp \encob{\Gamma}{\ct}$.
\end{theorem}

\begin{proof}
The proof proceeds by cases on the judgement $\Gamma \compsi \Psi\s P$ used in the translation  (\Cref{def:typed_encr}), and by inversion on the last typing rule applied (cf.~\Cref{fig:sess_typingmod}).
There are seven	 cases to consider:
\begin{enumerate}
\item
$\Gamma \compsi \emptyset \s \nil$. Then, by inversion the last typing rule applied is \Did{T-NilD}, with $\eend{\Gamma}$.
By \Cref{def:typed_encr} we have
$\encCPalt{\Gamma^\eende \compsi \emptyset \s \nil} = \{ \nil \}$.
Then, by applying Rule~{\Did{T-$\one$}} we obtain 
$$
\inferrule*[Right = {\Did{T-$\one$}}]
{ }
{\nil\cp  x: \bullet}
$$
The thesis follows immediately by encoding of types $\encob{\cdot}{\ct}$ given in \Cref{f:enctypesct},
and \Cref{def:enc_env_sc}, which states that
$\encob{x:\nilT}{\ct}=x: \bullet$.

\smallskip

\item
${\Gamma, x:\oc^0 T.S, v:T \compsi\Psi,\langle x, v, 0 \rangle \s\dual{x}\out v.P'}$,
where by inversion and  Rule \Did{T-OutD} we have $\Gamma = \myplus{\Gamma'}$, $\Psi = \myplus{\Psi'},\langle x, v, 0 \rangle$, and $S= \myplus{S'}$ as in the following derivation.
$$
\inferrule*[right={\Did{T-OutD}}]
	{
	   {\Gamma'}, x:S' \compsi \Psi' \s P'
	}
	{
	   \myplus{\Gamma'}, x:\oc^0 T.\myplus{S'}, v:T\compsi \myplus{\Psi'},\langle x, v, 0 \rangle \s \ov x\out{v}.P'
	}
$$
By \Cref{def:typed_encr},
$$\encCPalt{
	   \myplus{\Gamma'}, x:\oc^0 T.\myplus{S'}, v:T\compsi \myplus{\Psi'},\langle x, v, 0 \rangle \s \ov x\out{v}.P'
	}=
\big\{ \dual{x}(z).(\linkr{v}{z}\para Q) \suchthat Q \in \encCPalt{{\Gamma'}, x:S' \compsi \Psi' \s P'}\big\}$$
By Rule~\Did{T-$\tid$} and by \Cref{lem:dualenc} we have
\begin{equation}
\inferrule*[Right =  \Did{T-$\tid$}]
{ }
{\linkr{v}{z}\cp v:\encob{T}{\ct}, z:\encob{\overline T}{\ct}}
\label{eq:link}
\end{equation}
By induction hypothesis, 
$Q \in \encCPalt{{\Gamma'}, x:S' \compsi \Psi' \s P'}\big\}$ implies $Q\cp\encob{{\Gamma'}}{\ct}, x:\encob{S'}{\ct}$.
Let $Q'$ be an arbitrary process in this set.
By applying Rule \Did{T-$\otimes$} on \eqref{eq:link} and $Q'$ we have:
$$
\inferrule*[right= \Did{T-$\otimes$}]
{\linkr{v}{z}\cp v:\encob{T}{\ct}, z:\encob{\overline T}{\ct}  \\ 
Q'\cp\encob{{\Gamma'}}{\ct}, x:\encob{S'}{\ct}
}
{\dual{x}(z). \big(\linkr{v}{z}\para Q' \big)
\cp \encob{{\Gamma'}}{\ct},x: \encob{\dual T}{\ct}\otimes \encob{S'}{\ct}, v:\encob{T}{\ct} }
$$
By the encoding of types $\encob{\cdot}{\ct}$ in \Cref{f:enctypesct} we have that
$\encob{\oc T.S'}{\ct} =  \encob{\dual T}{\ct}\otimes \encob{S'}{\ct}$, and
by \Cref{def:enc_env_sc}, we have
$\encob{{\Gamma'}, x:\oc T.S', v:T}{\ct} = \encob{\Gamma'}{\ct},x: \encob{\dual T}{\ct}\otimes \encob{S'}{\ct}, v:\encob{T}{\ct}$, and we conclude this case by applying \Cref{lem:myplus}.

\smallskip

\item
${{\Gamma_1,\Gamma_2}, x:\oc^0 T.S\compsi\Psi_1, \Phi_1  [\Psi_2, \Phi_2] \s\res{zy}\dual{x}\out y.(P_1\para P_2)}$.
By inversion, this judgement is derived by a sequence of rules:
the last rule applied is \Did{T-ResD}, before that \Did{T-OutD} and \Did{T-ParD}.
We have:
$\Gamma_i = \myplus{\Gamma_i'}$, for $i\in \{1,2\}$;
$\Phi_j = \myplus{\Phi_j'}$, for $j\in \{1,2\}$; 
$\Psi_k = \myplus{\Psi_k'}$, for $k\in \{1,2\}$;
$S = \myplus{S'}$.
$$
\inferrule*[Right = \Did{T-ResD}]
{
\inferrule*[Right = \Did{T-OutD}]
{\inferrule*[Right = \Did{T-ParD}]
{\Gamma'_1, z:\ov T \compsi \Psi_1', \Psi_2'\s P_1 \and \Gamma'_2, x:S' \compsi \Phi_1',\Phi_2'\s P_2}
{\Gamma'_1, z:\ov T, \Gamma'_2, x:S' \compsi \Psi_1',\Psi_2', \Phi_1', \Phi_2'\s P_1\pp P_2}}
{\myplus{\Gamma_1'},\myplus{\Gamma_2'}, x:\oc^0 T.\myplus{S'}, y:T, z: \myplus{\ov T} \compsi \myplus{\Psi_1'}, \myplus{\Psi_2'}, \myplus{\Phi_1'}, \myplus{\Phi_2'}, \langle x, y, 0 \rangle\s {\overline{x}\out y.(P_1\para P_2)}}
}
{{\Gamma_1},{\Gamma_2}, x:\oc^0 T.S\compsi\Psi_1, \Phi_1  [\Psi_2, \Phi_2] \s\res{zy}\dual{x}\out y.(P_1\para P_2)}
$$
By \Cref{def:typed_encr},
$$
\begin{array}{lll}
\ {\encCPalt{{\Gamma_1,\Gamma_2}, x:\oc^0 T.S\compsi\Psi_1, \Phi_1  [\Psi_2, \Phi_2] \s\res{zy}\dual{x}\out y.(P_1\para P_2)}}	=\\
\hspace{2em}
\big\{ \dual{x}(y).(Q_1\para Q_2) \ \suchthat\ 
Q_1 \in \encCPalt{{{\Gamma_1},  z:\ov T\compsi\Psi_1, \Psi_2 \s P_1}}
\ \wedge \
Q_2 \in \encCPalt{{{\Gamma_2},x:{S}\compsi \Phi_1,\Phi_2 \s P_2}}
\big\}
\end{array}
$$
By induction hypothesis,
for all processes $Q_1 \in \encCPalt{{{\Gamma_1},  z:\ov T\compsi\Psi_1, \Psi_2\s P_1}}$, we have $Q_1\cp \encob{\Gamma_1}{\ct}, z:\encob{\ov T}{\ct}$. 
Let $Q_1'$ be an arbitrary process in this set.
Also,
for all processes $Q_2 \in \encCPalt{{{\Gamma_2},x:{S}\compsi \Phi_1,\Phi_2\s P_2}}$, we have $Q_2\cp \encob{\Gamma_2}{\ct}, x:\encob{S}{\ct}$. 
Let $Q_2'$ be an arbitrary process in this set.
By applying  Rule~\Did{T-$\otimes$} on $Q_1'$ and $Q_2'$ we have the following derivation:
$$
\inferrule*[right= \Did{T-$\otimes$}]
{Q_1'\cp \encob{\Gamma_1}{\ct}, z:\encob{\ov T}{\ct}  \\ 
Q_2' \cp \encob{\Gamma_2}{\ct}, x:\encob{S}{\ct}
}
{\dual{x}(z). \big(Q_1'\para Q_2' \big)
\cp \encob{\Gamma_1}{\ct}, \encob{\Gamma_2}{\ct}, x:\encob{\ov T}{\ct}\otimes \encob{S}{\ct}
}
$$
By the encoding of types $\encob{\cdot}{\ct}$ in \Cref{f:enctypesct},  
$\encob{\oc T.S}{\ct} =  \encob{\dual T}{\ct}\otimes \encob{S}{\ct}$, and
by \Cref{def:enc_env_sc}, we have
$\encob{{\Gamma_1}, \Gamma_2, x:\oc T.S}{\ct} = \encob{\Gamma'}{\ct},x: \encob{\Gamma_1}{\ct}, \encob{\Gamma_2}{\ct}, x:\encob{\ov T}{\ct}\otimes \encob{S}{\ct}$, and we conclude this case by applying \Cref{lem:myplus}.

\smallskip

\item
${\Gamma, x:\wn^0 T.S \compsi\Psi, (x, y, 0) \s x\inp {y:T}.P'}$,
where by inversion and Rule~\Did{T-InD} we have
$\Gamma = \myplus{\Gamma'}$, $\Psi = \myplus{\Psi'},(x, y, 0)$, and $S= \myplus{S'}$ as in the following derivation:
$$
\inferrule*[right={\Did{T-InD}}]
	{
 	  {\Gamma'}, x:S', y:T \compsi\Psi' \s P'
	}
	{
	  \myplus{\Gamma'}, x:\wn^0 T.\myplus{S'} \compsi\myplus{\Psi'},(x, y, 0) \s x\inp{y:T}.P'
	}
$$
By \Cref{def:typed_encr}, we have
$$
\encCPalt{\myplus{\Gamma'}, x:\wn^0 T.\myplus{S'} \compsi\myplus{\Psi'},(x, y, 0) \s x\inp{y:T}.P'}
=
\big\{x(y:T).Q \suchthat Q \in 
\encCPalt{{\Gamma'}, x:S', y:T \compsi\Psi' \s P'}\big\}	
$$
By induction hypothesis,  
$Q \in \encCP{{\Gamma'}, x:S', y:T \compsi\Psi' \s P'}$ implies
$Q\cp \encob{{\Gamma'}}{\ct}, x:\encob{S'}{\ct}, y:\encob{T}{\ct}$.
Let $Q'$ be an arbitrary process in this set.
By applying Rule~\Did{T-$\parl$} on $Q'$ we have:
$$
\inferrule*[right = \Did{T-$\parl$}]
{Q'\cp \encob{{\Gamma'}}{\ct}, x:\encob{S'}{\ct}, y:\encob{T}{\ct} }
{x(y:T).Q' \cp \encob{{\Gamma'}}{\ct}, x:\encob{T}{\ct}\parl \encob{S'}{\ct}}
$$
where by the encoding of types $\encob{\cdot}{\ct}$ in \Cref{f:enctypesct} we have that
$\encob{\wn T.S'}{\ct} =  \encob{T}{\ct}\parl \encob{S'}{\ct}$ and
by \Cref{def:enc_env_sc}, we have
$\encob{{\Gamma'}, x:\wn T.S'}{\ct} = \encob{{\Gamma'}}{\ct}, x:\encob{T}{\ct}\parl \encob{S'}{\ct}$, and we conclude this case by applying \Cref{lem:myplus}.

\smallskip 

\item
The cases for branching and selection follow the same lines as the corresponding ones in \Cref{thm:L0-L2}, whose proof is given in \Cref{appthm:L0-L2}, hence we omit them.

\smallskip 
\item
$\Gamma_1\restricted{\wt{x:S}} \compsi \Psi_1 \restricted{\Psi_2(\wt{x})} \circc 
\Gamma_2\restricted{\wt{y:T}} \compsi \Phi_1 \restricted{\Phi_2(\wt{y})} \s \res {\wt{x}\wt{y}:\wt{S}}(P_1\para P_2)$.
Then, by inversion we have following derivation:
$$
\inferrule*[Right = \Did{T-ResD}]
{
\inferrule*[Right = \Did{T-ParD}]
{{\Gamma_1, {\wt{x:S}} \compsi\Psi_1,\Psi_2\ \s P_1}
\\ {\Gamma_2,\wt{y:T}\compsi\Phi_1,\Phi_2\s P_2}}
{\Gamma_1 ,  \Gamma_2,  \wt{x:S},\wt{y:T} \compsi\Psi_1,\Psi_2,  \Phi_1, \Phi_2 \s P_1\para P_2}
}
{\Gamma_1\restricted{\wt{x:S}} \compsi \Psi_1 \restricted{\Psi_2(\wt{x})} \circc 
\Gamma_2\restricted{\wt{y:T}} \compsi \Phi_1 \restricted{\Phi_2(\wt{y})}  \s \res {\wt{x}\wt{y}:\wt{S}}(P_1\para P_2)}
$$

such that
$\wt T\defeq \wt{\dual S}$.

By \Cref{def:typed_encr}, we have
$$\encCPalt{\Gamma_1\restricted{\wt{x:S}} \compsi \Psi_1 \restricted{\Psi_2(\wt{x})} \circc 
\Gamma_2\restricted{\wt{y:T}} \compsi \Phi_1 \restricted{\Phi_2(\wt{y})}  \s \res {\wt{x}\wt{y}:\wt{S}}(P_1\para P_2)}$$
is the set of processes
\begin{align}
\label{eq:union_par}
&\big\{C_1[Q_1] \para G_2  \suchthat 
G_2 \in  \chrpv{\Gamma_2}{}{\Phi_1}\, , 
C_1 \in \mathcal{C}_{\,\wt{x:T} \,|\, \wt{y:S}}^{\Phi_2(\wt{x})}\,,
Q_1 \in  \encCPalt{{\Gamma_1, {\wt{x:S}}} \compsi\Psi_1,\Psi_2 \s P_1}
\big\}
\\[1mm]									
\cup ~~
&\big\{G_1 \para C_2[Q_2]  \suchthat 
G_1 \in \chrpv{\Gamma_1}{}{\Psi_1}\, ,
C_2 \in \mathcal{C}_{\,\wt{y:S} \,|\, \wt{x:T}}^{\Psi_2(\wt{y})}\,,
Q_2 \in \encCPalt{\Gamma_2,\wt{y:T}\compsi\Phi_1,\Phi_2\s P_2}
\big\}
\label{eq:union_par2}
\end{align}

We start by inspecting the set given by \eqref{eq:union_par}.
By induction hypothesis on the left-hand side premise of Rule~\Did{T-ParD} we have:
$$ Q_1\in \encCP{\Gamma_1, {\wt{x:S}} \compsi\Psi_1,\Psi_2\ \s P_1} \text{ implies }
Q_1\cp \encob{\Gamma_1}{\ct},  \wt{x: {\encob{S}{\ct}}}$$
Let $Q_1'$ be an arbitrary process in this set.
By 
\Cref{lem:typability_catal_rev}
we have that $C_1[Q_1'] \cp \encob{\Gamma_1}{\ct}$.

By \Cref{lem: wt_characteristic_new}(b),
since $G_2 \in  \chrpv{\Gamma_2}{}{\Phi_1,\Phi_2}$, we have that $G_2 \cp \encob{\Gamma_2}{\ct}$.
Since $\Gamma_1$ and $\Gamma_2$ are disjoint, by Rule~\Did{T-$\mix$} we have the following derivation, which concludes this case:
$$
\inferrule*[Right=\Did{T-$\mix$}]
{C_1[Q'] \cp \encob{\Gamma_1}{\ct} \\ G_2\cp  \encob{\Gamma_2}{\ct}}
{C_1[Q']\para G_2\cp \encob{\Gamma_1}{\ct}, \encob{\Gamma_2}{\ct}}
$$
The proof for the set given by (\ref{eq:union_par2}) follows the same line as the proof for \eqref{eq:union_par}, so we omit the details here.
We thus conclude that every process belonging to the set in (\ref{eq:union_par}) or (\ref{eq:union_par2}) is typed under the typing context
$\encob{\Gamma_1}{\ct}, \encob{\Gamma_2}{\ct}$, concluding this case.
\end{enumerate}
\end{proof}

\begin{theorem}[Operational Correspondence for $\encCPalt{~\cdot~}$]\label{thm:rwoc}
Let $P$ be  
such that $\Gamma \compsi \Psi\s P$ for some $\Gamma, \Psi$. Then we have:

\begin{enumerate}
\item
\nurev{If $P\to P'$, then for all 
$Q \in \encCPalt{\Gamma \compsi \Psi\s P}$
 there exist $Q', R$ such that
$Q\to\fred Q'$, $Q' \uptok R$, and 
$R \in \encCPalt{\Gamma \compsi \Psi\s P'}$
}

\item
\nurev{If 
$Q \in \encCPalt{\Gamma \compsi \Psi\s P}$, such that $P\in \lkoba$, and $Q \to\fred Q'$, then there exist
$P', R$ such that
$P\to P'$,
$Q' \uptok R$, and 
$R \in \encCPalt{\Gamma \compsi \Psi\s P'}$.
}
\end{enumerate}

\end{theorem}

\begin{proof}
By induction on the length of the derivations ($P\to P'$ and $Q \to Q'$), following the structure of the corresponding proof for $\encCP{\cdot}{}$ (cf. 
\Cref{thm:oc} and \Cref{appthm:oc}).
\end{proof}